\documentclass[11pt]{article}
\usepackage[margin=1in]{geometry}

\usepackage[english]{babel}

\usepackage{amsmath,amsthm,amssymb}
\usepackage{thmtools}
\usepackage{thm-restate}
\usepackage{graphicx}
\usepackage[colorlinks=true, allcolors=blue]{hyperref}
\usepackage[capitalize]{cleveref}
\usepackage{libertine}
\usepackage[framemethod=TikZ]{mdframed}

\usepackage{datetime}
\usepackage{multirow}
\usepackage{booktabs}
\usepackage{enumitem}
\usepackage{wrapfig}
\usepackage{subcaption}
\usepackage[T1]{fontenc}
\usepackage[utf8]{inputenc}
\usepackage{pgfplots}
\usepackage{amsmath, amssymb}
\usepackage{amsthm}
\usepackage{color}
\usepackage{cases}
\usepackage{xparse}
\usepackage{xargs}
\usepackage{appendix}
\usepackage[boxed,commentsnumbered,noend]{algorithm2e}
\usepackage{algpseudocode}
\usepackage{verbatim, xspace}
\usepackage{tikz}
\usepackage{mathtools}

\newtheorem{theorem}{Theorem}[section]
\newtheorem{claim}[theorem]{Claim}
\newtheorem*{claim*}{Claim}
\newenvironment{claimproof}[1]{\par\noindent\emph{Proof.}\space#1}{\hfill $\blacksquare$\\}
\newtheorem{lemma}[theorem]{Lemma}
\newtheorem{corollary}[theorem]{Corollary}
\newtheorem{observation}[theorem]{Observation}
\theoremstyle{definition} 
\newtheorem{definition}[theorem]{Definition}
\newtheorem{remark}[theorem]{Remark}
\newtheorem{proposition}[theorem]{Proposition}

\numberwithin{theorem}{section}

\newcommand{\Int}{\mathrm{Int}}

\newcommand{\eps}{\varepsilon}

\newcommand{\Oh}{\mathcal{O}}
\newcommand{\dist}{\mathrm{dist}}
\newcommand{\tw}{\mathrm{tw}}

\newcommand{\Otilde}{\widetilde{\Oh}}
\newcommand{\Ot}{\Otilde}

\newcommand{\nat}{\mathbb{N}}

\newcommand{\Z}{\mathbb{Z}}
\newcommand{\Bb}{\mathcal{B}}
\newcommand{\Ff}{\mathcal{F}}

\newcommand{\Cc}{\mathcal{C}}
\newcommand{\Tt}{\mathcal{T}}

\newcommand{\Uu}{\mathcal{U}}

\newcommand{\Gg}{\mathcal{G}}

\newcommand{\bnd}{\partial}
\newcommand{\real}{\mathbb{R}}

\newcommand{\poly}{\mathrm{poly}}
\newcommand{\polylog}{\mathrm{polylog}}

\newcommand{\bd}{\partial}
\newcommand{\sph}{\mathbb{S}}
\newcommand{\etal}{\emph{et~al.}~}
\newcommand{\outcome}[1]{\textbf{Outcome~#1}}
\newcommand{\grid}{\mathrm{Grid}}

\renewcommand{\root}{\mathrm{root}}

\renewcommand{\leq}{\leqslant}
\renewcommand{\geq}{\geqslant}
\renewcommand{\le}{\leqslant}
\renewcommand{\ge}{\geqslant}
\renewcommand{\hat}{\widehat}

\title{Separator Theorem and Algorithms for Planar Hyperbolic Graphs
\thanks{This research was partially carried out during the Parameterized Algorithms Retreat of the University of Warsaw, PARUW 2022, held in B\k{e}dlewo in April 2022.}
}
\author{
    S\'andor Kisfaludi-Bak\footnote{Department of Computer Science, Aalto University, Finland, \textsf{sandor.kisfaludi-bak@aalto.fi}}
    \and
    Jana Masa\v r\'ikov\'a\footnote{Institute of Informatics, University of Warsaw, Poland, \textsf{jnovotna@mimuw.edu.pl}. 
    }
    \and
    Erik Jan van Leeuwen\footnote{Department of Information and Computing Sciences, Utrecht University, The Netherlands, \textsf{e.j.vanleeuwen@uu.nl}.}
    \and
    Bartosz Walczak\footnote{Department of Theoretical Computer Science, Faculty of Mathematics and Computer Science, Jagiellonian University, Krak\'ow, Poland, \textsf{bartosz.walczak@uj.edu.pl}.  Partially supported by the National Science Center of Poland under grant No.\ 2019/34/E/ST6/00443.}
    \and
    Karol W\k{e}grzycki\footnote{Saarland University and Max Planck Institute for Informatics,
        Saarbr\"ucken, Germany, \textsf{wegrzycki@cs.uni-saarland.de}. 
    This work is part of the project TIPEA that has
    received funding from the European Research Council (ERC) under the European Unions Horizon
    2020 research and innovation programme (grant agreement No.\ 850979).}
}
\date{}

\begin{document}
\maketitle

\thispagestyle{empty}
\begin{abstract}
The hyperbolicity of a graph, informally, measures how close a graph is (metrically) to a tree. Hence, it is intuitively similar to treewidth, but the measures are formally incomparable. Motivated by the broad study of algorithms and separators on planar graphs and their relation to treewidth, we initiate the study of planar graphs of bounded hyperbolicity.

Our main technical contribution is a novel balanced separator theorem for planar
$\delta$-hyperbolic graphs that is substantially stronger than the classic
planar separator theorem. For any fixed $\delta \geq 0$, we can find balanced
separator that induces either a single geodesic (shortest) path or a single
geodesic cycle in the graph.

An important advantage of our separator is that the union of our separator
(vertex set $Z$) with any subset of the connected components of $G - Z$ induces
again a planar $\delta$-hyperbolic graph, which would not be guaranteed with an
arbitrary separator. Our construction runs in near-linear time and guarantees
that size of separator is $\mathrm{poly}(\delta) \cdot \log n$.

As an application of our separator theorem and its strong properties, we obtain
two novel approximation schemes on planar $\delta$-hyperbolic graphs. We prove
that {\sc Maximum Independent Set} and the {\sc Traveling Salesperson} problem
have a near-linear time FPTAS for any constant $\delta$, running 
in $n\, \mathrm{polylog}(n) \cdot 2^{\Oh(\delta^2)} \cdot \eps^{-\mathcal{O}(\delta)}$ time.

We also show that our approximation scheme for {\sc Maximum Independent Set} has essentially the best possible running time under the Exponential Time Hypothesis (ETH). This immediately follows from our third contribution: we prove that {\sc Maximum Independent Set} has no $n^{o(\delta)}$-time algorithm on planar $\delta$-hyperbolic graphs, unless ETH fails.
\end{abstract}

\begin{picture}(0,0)
\put(462,-200)
{\hbox{\includegraphics[width=40px]{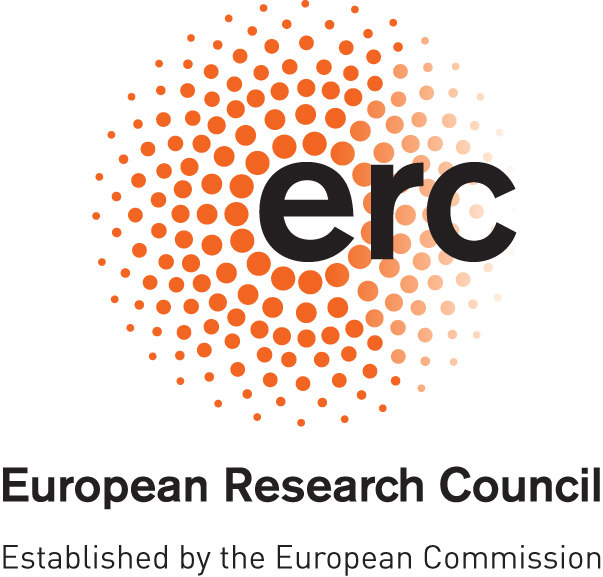}}}
\put(452,-260)
{\hbox{\includegraphics[width=60px]{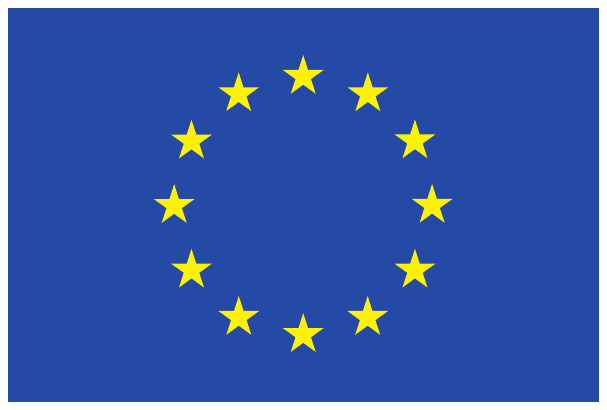}}}
\end{picture}

\clearpage
\setcounter{page}{1}
\section{Introduction}
Many graph problems are known to be efficiently solvable on trees. A substantial
research effort has long been underway to transfer this simple insight to more
complex graphs that are somehow `tree-like'. While many measures have been proposed (see e.g.~\cite{HlinenyOSG08,Vatshelle12}), one of the most successful measures
of tree-likeness has arguably been treewidth. We refer to the surveys of Bodlaender~\cite{Bodlaender93,Bodlaender98,Bodlaender05} or the recent book by Fomin et al.~\cite{2020bodlaender} for an overview of treewidth.
The study of graph structure and algorithms
on graphs of bounded treewidth has led to many celebrated results (see,
e.g.,~\cite{Baker94,Bodlaender96,Courcelle90,CyganNPPRW22,DemaineFHT05,RobertsonS86a,RobertsonS86b})
that are important in their own way or as a subroutine in other algorithms.
Treewidth has also been highly useful in practice, for example, for probabilistic inference in Bayesian networks~\cite{LauritzenS88}. For many other real-world networks and standard random models of them, treewidth is unfortunately very high~\cite{AdcockSM16,Gao12,ManiuSJ19,MontgolfierSV11}. This makes algorithms for graphs of bounded treewidth not very useful in this context and seems to cast doubt on the tree-likeness of such networks.

It has been shown, however, that many real-world networks are \emph{metrically} close to a tree. This idea can be formally cast to Gromov's notion of hyperbolicity. The graph $G$ is \emph{$\delta$-hyperbolic} if for all quadruples $x,y,z,w\in V(G)$ we have that the greater two among the sums 
\[\dist(x,y)+\dist(z,w),\;\dist(x,z)+\dist(y,w),\; \dist(x,w)+\dist(y,z)\]
differ by at most $2\delta$, where $\dist(.,.)$ denotes the shortest path
distance in $G$. The \emph{hyperbolicity} of a graph $G$ is the smallest $\delta$ such
that $G$ is $\delta$-hyperbolic. It is difficult to build good intuition for
this definition without diving into hyperbolic geometry, but the alternative
notion of $\delta'$-slimness can be used instead and is easier to picture.
Consider a triplet $(x,y,z)$ of vertices and take any three shortest paths
connecting the pairs $(x,y)$, $(y,z)$, and $(z,x)$. We call the vertices of these three paths the sides of the triangle $xyz$. We then say that the graph is $\delta'$-\emph{slim} if for each triangle the side $xy$ is within distance $\delta'$ from the union of the sides $yz$ and $zx$. 
The \emph{slimness} of a graph $G$ is then the smallest $\delta'$ such that $G$ is $\delta'$-slim. 
It is known that the hyperbolicity and slimness of a graph differ by a constant factor from each other~\cite{gromov87,BH99}, so they can almost be considered equivalent for the purposes of this section.

Since the notions of hyperbolicity and slimness are central to our work, let us develop some intuition for them through some simple examples. We can observe that the hyperbolicity and the slimness of a tree are both $0$, which is in line with the idea that hyperbolicity measures tree-likeness in a metric sense. On the other hand, cycles are not metrically close to trees. Indeed, the cycle $C_n$ is (roughly) $(n/4)$-hyperbolic and $(n/6)$-slim, which can be checked by placing the quadruple or the triangle vertices at equal distances along the cycle. In general, the hyperbolicity of a graph is bounded by its diameter or even its treelength\footnote{For readers familiar with tree decompositions, the \emph{treelength} of a tree decomposition of $G$ is the maximum diameter of any bag, where distances are measured in $G$ (and not in the metric of the subgraph induced by the bag). The \emph{treelength} of a graph is the minimum treelength of any of its tree decompositions.}~\cite{ChepoiDEHV08}. This is in a stark contrast with treewidth --- note that $C_n$ has treewidth~$2$. As another instructive example, observe that both the treewidth and the hyperbolicity of the $n \times n$ grid are $n$~\cite{MontgolfierSV11}. However, treewidth and hyperbolicity are not comparable, as the complete graph $K_n$ has treewidth $n-1$ but is $0$-hyperbolic and $0$-slim. Finally, it has been observed that many real-world networks are $\delta$-hyperbolic for some small constant~$\delta$ (see, e.g.,~\cite{Abu-AtaD16,AdcockSM16,ChenFHM12,MontgolfierSV11,NarayanS11} and the discussion in~\cite{Abu-AtaD16} for other relevant measures). There is evidence that the reason behind such a phenomenon is that real-world networks often have a large core of high-degree vertices through which most shortest paths pass~\cite{LeskovecLDM09,NarayanS11,ShavittT08}. Hence, hyperbolicity seems to be a useful measure by which to study networks (see also, e.g.,~\cite{BorassiCC15}).

While a substantial research effort has focused on algorithms for treewidth and
many other measures of tree-likeness, hyperbolicity has received comparatively limited
attention. Chepoi et al.~\cite{ChepoiDEHV08} studied spanners and the
computation of the center and diameter of $\delta$-hyperbolic graphs, while
Chepoi and Estellon~\cite{ChepoiE07} considered packing and covering problems
for balls. There have also been several investigations on how to compute the
hyperbolicity of a graph and improve on the naive $\Oh(n^4)$-time
algorithm~\cite{FournierIV15,BorassiCH16,CohenCDL17,CohenCL12,CoudertDP19,FluschnikKMNNT19}. 
Some research has gone into studying algorithms for objects in hyperbolic geometry, e.g., point sets or graphs embedded in hyperbolic space. For example, Krauthgamer and Lee~\cite{KrauthgamerL06} and Kisfaludi-Bak~\cite{Kisfaludi-Bak21} studied the {\sc Traveling Salesperson} problem in this context. Using
hyperbolic space, one can also define a graph where embedded vertices are
adjacent if they are `close' in hyperbolic space, which gives rise to hyperbolic ball graphs. Hyperbolic random graphs,
where the vertices are embedded randomly, are a particularly well-studied case of hyperbolic ball graphs (see,
e.g.,~the survey of Friedrich~\cite{Friedrich19} and the works of Bl{\"a}sius et
al.~\cite{BlasiusFFK20,BlasiusFK16,BlasiusFK18}).
Kisfaludi-Bak~\cite{Kisfaludi-Bak20} studied the complexity of several NP-hard
problems on hyperbolic ball graphs. These results based on hyperbolic geometry are related
but they have a different flavor. In particular, the geometric graphs are often but not always $\delta$-hyperbolic graphs. Moreover, in the geometric setting one can always rely on the homogeneity of the underlying hyperbolic space. In contrast, $\delta$-hyperbolic graphs (for small $\delta$) have no underlying space that can be utilized. We note however that embedding $\delta$-hyperbolic graphs into high-dimensional hyperbolic space with low distortion is possible~\cite{bonk2011embeddings}.

\begin{figure}
    \centering
    \includegraphics{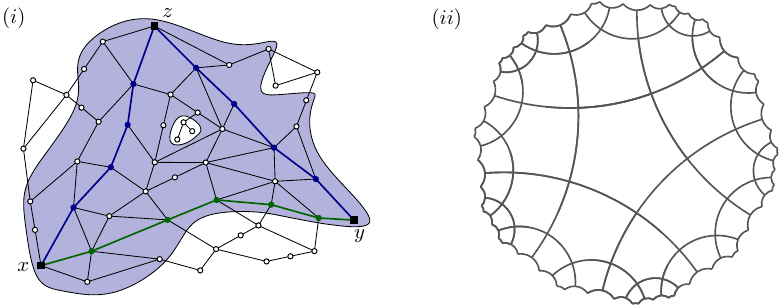}
    \caption{(i) The triangle $xyz$ is $2$-slim as the $2$-neighborhood of the sides $yz$ and $zx$ (vertices in the blue shaded region) covers the side $xy$ (green). (ii) Example of a planar constant-hyperbolic graph that does not directly resemble a tree. This graph is a part of the pentagonal tiling of the hyperbolic plane.}
    \label{fig:slimpatch}
\end{figure}

To advance research on algorithms for graphs of bounded hyperbolicity, we consider it in the context of planar graphs. Treewidth is already well-studied in this context. Indeed, some problems have more efficient algorithms on bounded treewidth graphs if the graph is planar see, e.g., Dorn et al.~\cite{DornPBF10}. (Recent work suggests that similar improvements may even be possible on general graphs~\cite{BodlaenderCKN15,CyganNPPRW22}.)
Treewidth is also an important tool in known approximation schemes for planar graphs~\cite{Baker94,DemaineH05,FominLS18}. 
For hyperbolicity, Cohen~\cite{CohenCDL17} developed an algorithm to compute the hyperbolicity of an outerplanar graph in linear time. We are unaware, however, of any studies on general planar graphs in relation to hyperbolicity. Motivated by this gap in our knowledge, this paper initiates research on planar $\delta$-hyperbolic graphs.

\subsection{Main Contribution: A Novel Separator Theorem}

A crucial tool for the algorithmic study of planar graphs has always been a
balanced separator. Recall that for a graph class ${\cal G}$ and a function
$\alpha : \mathbb{N} \rightarrow \mathbb{R}^+$ (possibly depending on ${\cal
G}$), we say that an $n$-vertex graph $G \in {\cal G}$ has a \emph{separator} of
\emph{balance} $\alpha(n)$ if there is a vertex set $Z \subset V(G)$ such that
the number of vertices in any connected component of $G- Z$ is at most
$(1-\alpha(n)) \cdot n$. Typically, the balance is a constant independent of
$n$. We say that a subgraph $H$ of $G$ is a geodesic path (cycle) if $H$ is a
path (cycle) where, for any $u,v\in V(H)$, we have $d_H(u,v)=d_G(u,v)$. 

Famously, Lipton and Tarjan~\cite{LiptonT79} proved that planar graphs have
a $\frac{1}{2}$-balanced separator of size~$\Oh(\sqrt{n})$. In fact, Lipton and
Tarjan~\cite{LiptonT79} showed these there exists such separator that consists
of \emph{two} geodesic paths.

In this paper, we develop a balanced separator theorem for planar
$\delta$-hyperbolic graphs. Importantly, our separator consists of single
geodesic path or geodesic cycle.  This shows that hyperbolic planar graphs offer
significantly more structure than general planar graphs.

\newcommand\separatorlogexponent{4}
\newcommand\separatorlogexponentplus[1]{\the\numexpr\separatorlogexponent+#1\relax}

\begin{restatable}[Separator for planar $\delta$-hyperbolic graphs]{theorem}{theoremsep} \label{thm:separator}
    Let $G$ be a connected planar $\delta$-hyperbolic graph on $n$ vertices.
    Then $G$ has a geodesic path separator $X$  and a constant balance or $G$
    has a geodesic cycle separator $Y$  and balance $2^{-\Oh(\delta)}/\log n$.
    Given $G$, such a separator $X$ or $Y$ can be computed in $\Oh(\delta^2
    n\log^\separatorlogexponent n)$ time.

    Additionally, $|X|=\Oh(\delta^2 \log n)$ and $|Y| = \Oh(\delta)$.
\end{restatable}

The proof of our separator theorem leans heavily on several novel techniques that we propose for planar $\delta$-hyperbolic graphs. We first present a new variant of the well-known isoperimetric inequality for $\delta$-hyperbolic graphs~\cite{BH99}. Then, we develop an iterative procedure that tries to construct a partition of a plane $\delta$-hyperbolic graph into regions that are bounded by cycles of length $\Oh(\delta)$ and that contain a small number of vertices. We either obtain a balanced geodesic cycle separator during the execution of this procedure or, if the procedure finishes, we show how we can exploit the imbalance of the cycles to find a geodesic path separator by cutting `straight through' the embedding. In the latter case, it is highly non-trivial to argue that this separator is both short and balanced; this argument crucially relies on our new isoperimetric inequality. If the procedure cannot do anything, because all faces have length $\Oh(\delta)$, then we show how any separator can be turned into a balanced cycle separator using deep insights into the structure of such separators. Finally, we prove that we can shorten certain balanced cycle separators into geodesic cycle separators without losing our balance. As such, we develop new tools that combine insights into both planar and $\delta$-hyperbolic graphs.

For a more detailed overview of the proof of \cref{thm:separator}, we refer to \cref{sec:overview}. The full proof is presented in \cref{sec:structure} and \cref{sec:sep}.

\paragraph{Usefulness of single geodesic-cycle/geodesic-path separator}

For a graph class ${\cal G}$, we say that a separator $Z$ of a graph $G \in
{\cal G}$ is \emph{in-class} if $G[Z \cup \bigcup_{C \in {\cal C}} V(C)] \in
{\cal G}$ for every subset ${\cal C}$ of the set of connected components of $G -
Z$. Note, that any separator for a \emph{hereditary} graph class is
automatically in-class; for example, any separator for planar graphs. However,
(planar) $\delta$-hyperbolic graphs are \emph{not hereditary}, as the deletion
of any vertex can substantially alter distances. Indeed, a wheel graph $W_{n+1}$
(consisting of $C_{n}$ plus a central universal vertex) has hyperbolicity at
most~$2$ (as hyperbolicity is upper bounded by diameter~\cite{ChepoiDEHV08}),
but removing its central vertex increases the hyperbolicity to (roughly) $n/4$.
Hence, we need a different separator theorem that does guarantee the in-class
property.

We now observe (and later prove formally) that if we separate a planar
$\delta$-hyperbolic graph along a geodesic path or cycle $Z$, then a shortest
path $P$ in $G$ between two vertices in a component $C$ of $G-Z$ such that $P$
is fully not contained in $C$ can be `rerouted' along $Z$ and remain shortest in
$G[V(C) \cup Z]$. Therefore, as a consequence of \cref{thm:separator}, we obtain
the desired in-class separator:

\begin{restatable}{corollary}{corollarypump} \label{cor:oursep-pump}
For any $\delta \geq 0$, the class of connected planar $\delta$-hyperbolic graphs has a $1/2$-balanced in-class separator of size $2^{\Oh(\delta)} \log n$ that can be computed in $2^{\Oh(\delta)} \cdot n\log^{\separatorlogexponentplus{1}} n$ time.
\end{restatable}

\paragraph{Discussion about the size of the separator}
In the applications of~\cref{thm:separator} and \cref{cor:oursep-pump}, discussed later, we will mainly use the property that our separator is \emph{in-class} and has \emph{sublinear size}. Note that \cref{thm:separator} even guarantees that the returned separator has size $\Oh(\delta^2 \log(n))$.
The separator of \cref{thm:separator} thus is significantly smaller (for small values of $\delta$) than the general planar separator of size~$\Oh(\sqrt{n})$~\cite{LiptonT79} (which is not even in-class in our case). On the other hand, the size bound is reminiscent of known separators for hyperbolic ball graphs and random graphs, which also have size~$\Oh(\log n)$ in certain regimes~\cite{BlasiusFK16,Kisfaludi-Bak20,kopczynski2020hyperbolic}. Recall, however, that hyperbolic
ball graphs differ substantially in nature from the (planar) $\delta$-hyperbolic graphs we study in this paper.

We observe that a separator of size $\Oh_\delta(\log(n))$ that is \emph{not in-class} can be easily obtained by combining two existing results. Chepoi et
al.~\cite[Proposition~13]{ChepoiDEHV08} bounded the treelength of
$\delta$-hyperbolic graphs, and Dieng and Gavoille~\cite{DiengG09} (see
also~\cite{Dieng09}) bounded the treewidth of a planar graph in terms of its
treelength, which gives the following bound on the treewidth of planar
$\delta$-hyperbolic graphs:

\begin{proposition}[\cite{ChepoiDEHV08} and \cite{DiengG09}] \label{prp:tw-hyper}
For any $\delta \geq 0$, the treewidth of any $n$-vertex planar $\delta$-hyperbolic graph is $\Oh(\delta \log n)$.
\end{proposition}
Observe that a constant-factor approximation of the treewidth $\tw$ of a planar
graph can be computed in $\Oh(n \cdot \tw^2 \log \tw)$
time~\cite{KammerT16,GuX19}. Using standard arguments (see,
e.g.,~\cite[(2.5)]{RobertsonS86a}), \cref{prp:tw-hyper} and the fact that
$\delta<n$ immediately implies the existence of a balanced separator:

\begin{corollary}\label{cor:tw-hyper-sep}
For any $\delta \geq 0$, the class of planar $\delta$-hyperbolic graphs has a $1/2$-balanced separator of size $\Oh(\delta \log n)$ that can be computed in $\Oh(\delta^2\, n\log^3 n)$ time.
\end{corollary}

We stress again that the effectiveness of \cref{cor:tw-hyper-sep} alone is
somewhat doubtful. When attempting to employ it in recursive algorithms (a
common approach for utilizing separators), the separator fails to guarantee that
its components are again $\delta$-hyperbolic.  \cref{thm:separator} guarantees
that the separator consists of a single geodesic path or a single geodesic cycle, which allows
us to develop novel algorithmic applications, which we discuss now.

\subsection{Applications of our Separator Theorem}
We present two applications of our separator theorem to well-known graph problems. Recall that an \emph{independent set} of a graph $G$ is a set $I \subseteq V(G)$ such that $uv \not\in E(G)$ for any $u,v \in I$. Then the {\sc Maximum Independent Set} problem asks to find an independent set of maximum size in a given graph $G$. We give a near-linear time FPTAS for {\sc Maximum Independent Set} on planar $\delta$-hyperbolic graphs (for any fixed $\delta$).

\begin{restatable}{theorem}{theoremmis} \label{thm:mis}
For any $\delta \geq 0$ and any $\eps > 0$, the class of planar
$\delta$-hyperbolic graphs has a $(1-\eps)$-approximation algorithm for {\sc
Maximum Independent Set} running in $2^{\Oh(\delta)}  n
\log^{\separatorlogexponentplus{2}} n + 2^{\Oh(\delta^2)} n/\eps^{\Oh(\delta)}$
time.
\end{restatable}

It is important to compare our approximation scheme to the known EPTAS for {\sc
Maximum Independent Set} on planar graphs, which runs in time $2^{\Oh(1/\eps)}
n$~\cite{Baker94} and is asymptotically optimal~\cite{Marx07a}. Our algorithm will be
substantially faster for small values of $\delta$. We also observe that the
usual approach to planar approximation schemes that uses a treewidth bound (e.g.,~\cref{prp:tw-hyper}), as pioneered by Baker~\cite{Baker94}, is likely not possible here.
Indeed, recall that the class of $\delta$-hyperbolic graphs is not hereditary
and thus removing BFS-layers does not necessarily preserve hyperbolicity.
Hence, our algorithm uses the separator of \cref{thm:separator} in the same way 
Lipton and Tarjan~\cite{LiptonT80} did in their pioneering work. In particular,
we show that we can compute a (weak) $r$-division of which each part induces a planar
$\delta$-hyperbolic graph. Our algorithmic approach is actually more general (using
ideas of Chiba et al.~\cite{ChibaNS81}) and allows us to prove approximation
schemes for several other problems (including e.g.\ the {\sc Maximum Induced Forest} problem).

We next consider the {\sc Traveling Salesperson} problem. We only consider the variant on undirected, unweighted graphs. We define a \emph{tour} in a graph $G$ to be a closed walk in $G$ that visits every vertex of $G$ at least once. Then the {\sc Traveling Salesperson} problem (also known as the {\sc Traveling Salesman} problem or {\sc Graph Metric TSP}) asks, given an unweighted, undirected graph $G$, to find a shortest tour in $G$. We give a near-linear time FPTAS for the {\sc Traveling Salesperson} problem on planar $\delta$-hyperbolic graphs (for any fixed $\delta$).

\begin{restatable}{theorem}{theoremtsp}\label{thm:tsp}
For any $\delta \geq 0$ and any $\eps > 0$, the class of planar $\delta$-hyperbolic graphs has a $(1+\eps)$-approximation algorithm for the {\sc Traveling Salesperson} problem running in $2^{\Oh(\delta)} \cdot n \log^{\separatorlogexponentplus{2}} n + 2^{\Oh(\delta^2)} n/\eps^{\Oh(\delta)}$ time.
\end{restatable}
We again compare our approximation scheme to the known approximation schemes for the {\sc Traveling Salesperson} problem on planar graphs. A first PTAS for this problem, running in time $n^{\Oh(1/\eps)}$, was designed by Grigni et al.~\cite{GrigniKP95}. This later improved to an EPTAS running in time $2^{\Oh(1/\eps)}n$ by Klein~\cite{Klein08}. (For later generalizations, to the weighted case and $H$-minor-free graphs, see e.g.~\cite{Klein08,Le18,Cohen-AddadFKL20} and references therein.) Our scheme will be substantially faster for small values of $\delta$. A crucial element in all these schemes is the definition of appropriate subproblems and the patching of partial solutions to form a general solution. Grigni et al.\ and Klein use different approaches to address these challenges: the former uses a recursive separator approach whereas the latter combines a spanner with a Baker-style shifting technique. Like for {\sc Maximum Independent Set}, we must be careful that planar $\delta$-hyperbolic graphs are not hereditary. Therefore, our approach again relies on the recursive separator approach of Lipton and Tarjan~\cite{LiptonT80}, although some of its ideas feel reminiscent of those underlying the previous schemes~\cite{GrigniKP95,Klein08}.

Finally, we note that an FPTAS for {\sc Maximum Independent Set} or the {\sc
Traveling Salesperson} problem is generally not possible, unless P=NP. However,
the dependence on $\delta$ in \cref{thm:mis} means that our schemes do not
disprove the standard complexity theory assumptions. We note that our approximation schemes can also be seen as parameterized approximation schemes (see e.g.~\cite{FeldmannSLM20}), in particular as EPASes, with parameter~$\delta$.

\subsection{Connection to Exact Algorithms}
To build a connection to exact algorithms, we first observe that the following results are immediate from \cref{prp:tw-hyper} combined with known algorithms on graphs of bounded treewidth for {\sc Maximum Independent Set}~\cite{ArnborgP89} and the {\sc Traveling Salesperson} problem~\cite[Appendix~D]{Le18} respectively.

\begin{corollary} \label{cor:is-exact}
For any $\delta \geq 0$, the class of planar $\delta$-hyperbolic graphs has an algorithm for {\sc Maximum Independent Set} running in time $n^{\Oh(\delta)}$.
\end{corollary}

\begin{corollary} \label{cor:tsp-exact}
For any $\delta \geq 0$, the class of planar $\delta$-hyperbolic graphs has an algorithm for the {\sc Traveling Salesperson} problem running in time $n^{\Oh(\delta)}$.
\end{corollary}
Note that, alternatively, these results follow from our approximation schemes (with an extra factor $2^{\Oh(\delta^2)}$ in the running time) by setting $\eps = 1/\Omega(n)$.

For {\sc Independent Set}, we prove a lower bound matching \cref{cor:is-exact}, conditional on the Exponential Time Hypothesis (ETH)~\cite{ImpagliazzoP01}, which asserts that there is no $2^{o(n)}$-time algorithm for the 3-\textsc{Satisfiability} problem. We prove:

\begin{restatable}{theorem}{theoremlower}\label{thm:lower_bound}
There is no $n^{o(\delta)}$-time algorithm for \textsc{Independent Set} in planar $\delta$-hyperbolic graphs, unless ETH fails.
\end{restatable}
This result immediately implies that the running time of \cref{thm:mis} is also essentially optimal, in the sense that there is no $(1-\eps)$-approximation scheme running in time $\poly(n)/\eps^{o(\delta)}$, unless ETH fails.

The lower bound of \cref{thm:lower_bound} also stands in contrast to our knowledge of graphs of bounded treewidth. It is known that {\sc Maximum Independent Set} can be solved in $2^{\Oh(\tw)} n$ time on $n$-vertex graphs of treewidth~$\tw$, but such a result (fixed-parameter tractability) will be unlikely by \cref{thm:lower_bound}.

\subsection{Organization}
We first give an overview of the main ideas of our paper in \cref{sec:overview}, particularly those behind \cref{thm:separator} and \cref{thm:lower_bound}. 
We then start the full paper with some preliminaries in \cref{sec:prelim}, including formal definitions of hyperbolicity and slimness. Then \cref{sec:structure} delivers several important structural analyses of $\delta$-hyperbolic graphs, such as the computation of a filling and the structure of geodesics. Using these ideas, we develop our separator theorem in \cref{sec:sep} and prove \cref{thm:separator}. In \cref{sec:schemes}, we apply our separator theorem. We first prove \cref{cor:oursep-pump} and then use it to prove \cref{thm:mis} and \cref{thm:tsp}. We prove the lower bound of \cref{thm:lower_bound} in \cref{sec:lower}. Finally, we discuss our results and ask open questions in \cref{sec:discussion}.

\section{Overview of Main Ideas and Techniques}\label{sec:overview}
In this section, we discuss the combinatorial observations and ideas behind the proofs of \cref{thm:separator} and \cref{thm:lower_bound}.

\subsection{Main Ideas and Techniques for the Separator Theorem}
For the sake of convenience, we briefly restate \cref{thm:separator}.

\theoremsep*

To grasp this theorem and how we prove it, it is important to understand why the dichotomy of the two types of separators in this theorem is unavoidable, and why the separator requires size $\Omega(\log n)$ and $\Omega(\delta)$, respectively. To this end, we give two illustrative examples.

First, there exist planar hyperbolic graphs of treewidth $\Omega(\log n)$. For example, Kisfaludi-Bak~\cite[Lemma~28]{Kisfaludi-Bak20-arxiv} showed that a size-$n$ patch of the pentagonal tiling of the hyperbolic plane (see \cref{fig:hyperbolic-examples}) contains a plane subgraph that is a subdivision of a $\log n \times \log n$ grid. Thus, a path separator of a large balance must be of length $\Omega(\log n)$. (A stronger lower bound of $\Omega(\delta\log n)$ can be derived from \cref{sec:lower}.) The pentagonal grid example (combined with the isoperimetric inequality~\cite{gromov87,BH99} discussed later) demonstrates that a geodesic cycle separator alone cannot always lead to a balanced separator, as any cycle of length $\ell$ in this graph can cut away only $\Oh(\ell)$ vertices.

Second, we consider the simple example of a $\delta$-cylinder: a graph consisting of $n$ copies of a circle $C = \{c_1,\ldots,c_\delta\}$ and for each $i$, a path through the vertices $c_i$ of each copy (see \cref{fig:hyperbolic-examples}). The $\delta$-cylinder has hyperbolicity $\Theta(\delta)$. In the $\delta$-cylinder, any geodesic cycle that would be a balanced separator has at least~$\delta$ vertices. This example also demonstrates that a geodesic path alone cannot lead to a balanced separator.

By these examples, our algorithm needs to output either a geodesic path or a geodesic cycle as a separator. Moreover, they need to be of the size as stated in the theorem, apart from a possible factor $\delta$ overhead in the size of the geodesic path separator. While the examples served as the starting point for our thinking in the proof of \cref{thm:separator}, arriving at the proof required significant technical effort.

\begin{figure}
\centering
\includegraphics{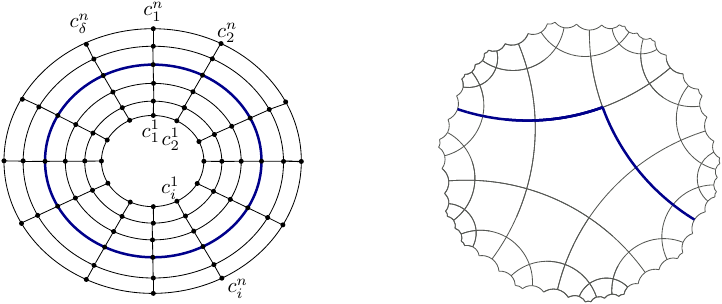}
\caption{Left: the $\delta$-cylinder with hyperbolicity $\Theta(\delta)$, and a typical geodesic cycle separator (in blue). Right: the pentagonal tiling graph and a shortest path separator (in blue).}
\label{fig:hyperbolic-examples}
\end{figure}

\medskip

We now give a high-level overview of the proof of \cref{thm:separator}. 
Let $G$ be a planar $\delta$-hyperbolic graph embedded on a sphere $\sph^2$. For simplicity, we assume first that $G$ is $2$-connected; we later argue how we can reduce to this case.

Our algorithm heavily relies on a known procedure to create $\Oh(\delta)$-fillings in $\delta$-hyperbolic graphs (see~\cite{BH99}). In the context of planar graphs, an $\Oh(\delta)$-\emph{filling} of a cycle $C$ of $G$ is a $2$-connected planar subgraph $H$ of $G$ that has $C$ as a face, and where each face of $H$ (except $C$) has length $\Oh(\delta)$. The so-called \emph{isoperimetric inequality}~\cite{gromov87,BH99} asserts that the minimum number of faces in a filling of $C$ is $\Oh(|C|)$. The procedure essentially `chops off' parts of the region enclosed by $C$ in a greedy manner, where each part is bounded by a cycle of length $\Oh(\delta)$.

We apply this filling procedure in an iterative manner with the goal of arriving at an $\Oh(\delta)$-filling where each face of the filling contains `few' vertices of $G$. First, we apply the filling procedure on the longest face in $G$ itself. If one of the faces of the resulting $\Oh(\delta)$-filling covers a region $F$ that contains `many' vertices of $G$, then we apply it iteratively on the longest face of this region. In this new iteration, the rest of the graph (formed by vertices in the interior of $\mathbb{S}^2 \setminus F$) is removed. This region becomes a face, which is assigned a weight equal to the number of vertices of $G$ outside $F$, of which there are `few' by the fact that there are `many' vertices of $G$ inside $F$. In this way, we slowly and iteratively proceed towards our stated goal. We only terminate prematurely if along the way a suitable separator is already found (see Outcome~1 below).

Let $G'$ be the graph after the final iteration of our algorithm (this may be after the above procedure fully finishes or is terminated prematurely). We can terminate with one of three outcomes:
\begin{itemize}
        \setlength\itemsep{0.2em}
    \item[]
\outcome{1}: One of the cycles of the current filling of $G'$ already has a good enough balance.
    \item[]
\outcome{2}: The maximum face length of $G'$ is $\Oh(\delta)$, i.e., a new greedy filling procedure would terminate with the trivial filling consisting only of the initial face cycle.
    \item[]
\outcome{3}: All faces of the current filling have few vertices of $G$ inside.
\end{itemize}
Note that when looking at the number of vertices of $G$ inside a face of the current filling, or in other words at the balance of this face, we also account for the newly assigned weights to (some of) the faces. We describe how we deal with each of these outcomes in turn.

\smallskip 

In \outcome{1}, there is a cycle in the filling with a good enough balance. In general, we can prove that if we encounter a cycle separator of length $\ell$ with balance $\alpha$, then we can compute a geodesic cycle separator of length $\Oh(\delta)$ with balance $\alpha/2^{\Oh(\ell)}$. This can be obtained by iteratively carving away a constant fraction of the vertices inside (or outside) the cycle while reducing the length of the separating cycle by at least~$1$. Applying this shortcutting procedure to the assumed cycle, we obtain a geodesic cycle separator.

\smallskip 

In \outcome{2}, all faces of $G'$ have length $\Oh(\delta)$. This is the case, for example, for $\delta$-cylinders. In a $\delta$-cylinder, we can directly find a geodesic cycle separator roughly in the middle of the cylinder. However, this intuition does not immediately carry over to  general planar $\delta$-hyperbolic graphs with short faces.

We first find a separator $S$ of size $\Oh(\delta \log n)$ by \cref{cor:tw-hyper-sep}. However, this separator is possibly not geodesic nor a path or cycle. Next, our goal is to transform the separator $S$ into a separator $\bar S$ that has a good \emph{split balance}, meaning that each face of the graph induced by $\bar S$ contains at most a constant proportion of the vertices of $G'$. This transformation is non-trivial and is done with the help of an auxiliary graph. Once the separator $\bar S$ with a good split balance is found, we can find some collection of faces in $G'[\bar S]$ whose union $U$ has a boundary $\bd U$ that gives a constant-balanced separator. However, the boundary $\bd U$ may consist of several cycles. We then find a single component cycle of $\bd U$ with a good split balance. Here, we need to offset the split balance of the cycle against its length. Hence, we select the component cycle $\gamma_i$ of $\bd U$ with the best ratio of balance to length. We then use a more involved shortening procedure on this cycle to find a geodesic cycle with the desired balance. For details, see~\cref{sec:outcome-2}.

\smallskip 

Finally, we can end up in \outcome{3}. In this case we think of the graph as embedded on the plane with the outer face being the starting cycle of the final filling. Recall that in this case, each filling face (except the outer face) has only a few vertices of $G$ inside. This case would be the outcome if the initial graph is a patch of the pentagonal grid, and the filling is based on the cycle around the perimeter of the patch, i.e., the boundary of the outer face. We now claim that we can find two `antipodal' vertices on the outer face such that some shortest path between them is a balanced separator.

It is far from clear in general why some shortest path connecting two `antipodal' vertices of the outer face has constant balance. We begin this proof by defining \emph{layers} on the filling faces: a face is in layer $i$ if its distance to the outer face is $i$. Roughly, we aim to show that there are only $\Oh(\log n)$ layers. 

To bound the number of layers, we prove a variant of the isoperimetric inequality for our purposes, which may be of independent interest. In general, consider a planar $\delta$-hyperbolic graph with a $\Oh(\delta)$-filling $H$ of a cycle $C$ of $G$, where $C$ is the cycle along a face in some fixed embedding of $G$. We then consider an arbitrary cycle $\gamma$ in $G$. In our \emph{isoperimetric inequality over greedy filling} we show that the interior of $\gamma$ intersects $\Oh(\delta |\gamma|)$ faces of the filling $H$. To bound the number of layers, one can show that the outer face cycle touches a constant proportion of all faces, i.e., the outermost layer has a constant proportion of all the faces of the filling. Iterating this argument shows that the number of layers is $\Oh(\log n)$.

Then, in the last layer, we find two vertices that are as far from the outer face as possible. Using an auxiliary tree in the planar dual graph, we can select a good balanced cut edge. The endpoints of the corresponding primal edge are connected to their respective nearest vertices $a$ and $b$ on the outer face. The path we obtain this way from $a$ to $b$ is a balanced separator, but unfortunately, it is not a shortest path. We need to argue about the balance of a shortest path from $a$ to $b$ in $G$ instead. To prove that some shortest $ab$ path is also a balanced separator, we crucially rely on the isoperimetric inequality on greedy fillings again. The relatively short closed walk given by the initial path and the shortest path can only interact with a small number of filling faces due to the isoperimetric inequality. Since we are in the case where each filling face contains only a few vertices inside, we can upper bound the balance shift between the initial path and the shortest path. This then gives the geodesic path separator with the desired balance. For details, see~\cref{sec:outcome-3}.

\medskip

With the above outcomes handled, the only missing piece of the proof of~\cref{thm:separator} is the case of graphs that are not $2$-connected. Again, since hyperbolicity is very sensitive to changes in the graph, we have a slightly more technical reduction from the general case to the $2$-connected case. Intuitively, it is enough to find a separator of a `central' 2-connected component, but this would not immediately account for the number of vertices in other components and thus potentially lead to an imbalanced separator. We represent all the non-central $2$-connected components of $G$ by attaching wheel graphs to the central $2$-connected component, which (i) ensures that a balanced separator of the reduced $2$-connected instance corresponds to a balanced separator in the original graph and (ii) does not increase the hyperbolicity significantly.

\subsection{Main Ideas and Techniques of the Lower Bound}
For the sake of convenience, we restate the lower bound:

\theoremlower*

The proof is based on embedding a subdivision of a Euclidean grid (with some diagonals) into a planar $\delta$-hyperbolic graph. It is known that \textsc{Independent Set} in subgraphs of the $n\times n$ grid (with some diagonals) has a $2^{o(n)}$-time lower bound under ETH~\cite{BergBKMZ20}, which we use as our starting point.

Let $G$ be a given subgraph of an $n\times n$ grid
with diagonals, which we denote by $\grid_n$. Observe that subdividing each edge of $G$ an even number of times gives an instance that is equivalent to $G$~\cite{Po74}. 
We are then left with two tasks: (a) create a planar $\delta$-hyperbolic host
graph of size $2^{\Oh(n/\delta)}$ that `surrounds' an even subdivision of
$G$, and (b) make sure that the surrounding parts of the host graph created in (a) do not impact the hardness proof.

\begin{figure}
    \centering
    \includegraphics{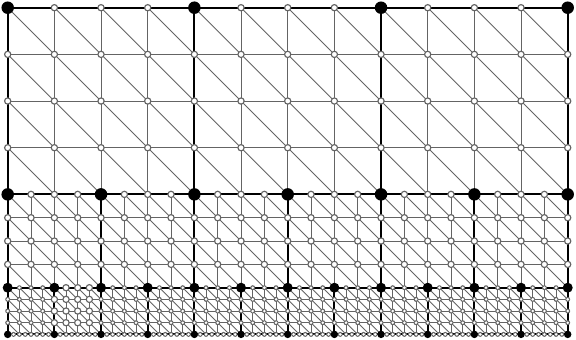}
    \caption{A part of a binary tiling (thick edges) with small $\delta\times\delta$ grids (with diagonals) embedded in each face except the outer face.}
    \label{fig:binarygrid_simple}
\end{figure}

Task (a) requires a thorough, technical approach. It is in fact easier to think of embedding $\grid_n$ itself instead of focusing on some custom graph $G$. Our construction is based on the so-called \emph{binary tiling} of B\"or\"oczky~\cite{Boro}, which is a tiling of the hyperbolic plane; see \cref{fig:binarygrid_simple}. The underlying infinite planar graph of this tiling is known to be constant hyperbolic.
We carefully choose a portion $B_1$ of this graph that has size $2^{\Oh(n/\delta)}$ and is still constant-hyperbolic. By inserting $\delta \times \delta$ grids into the faces of $B_1$ (except its outer face), we can show that we get an $\Oh(\delta)$-hyperbolic graph $B_\grid$, which contains some subdivision of $\grid_n$. A further modification of $B_\grid$ leads to a graph $B$ that is $\Oh(\delta)$-hyperbolic, is of size $2^{\Oh(n/\delta)}$, and contains a subgraph that is an even subdivision of $G$. It follows that $B$ contains a hard instance of \textsc{Independent Set} as a subgraph and we have achieved (a). 

To achieve (b), we cannot just remove unwanted parts of the $2^{\Oh(n/\delta)}$-hyperbolic host graph we just constructed, as that would change the hyperbolicity. However, in the case of {\sc Independent Set}, we show how to achieve (b) using a simple local modification, attaching a small gadget to vertices of $B$ that are not in the hard instance, that does not impact hyperbolicity. This gives the desired graph whose maximum independent sets can be related to the maximum independent sets of~$G$ and completes the reduction.

\section{Preliminaries} \label{sec:prelim}

For $n \in \nat$ we let $[n] \coloneqq \{1,\ldots,n\}$.  The $\poly(n)$ is a
shorthand notation for $n^{\Oh(1)}$.  The notation $\Ot(T)$ means $\Oh(T \cdot
\polylog(T))$. Sometimes, we use $o(1)$ notation which means a function in $n$
that asymptotically goes to $0$.  All the logarithms are base $2$ unless stated
otherwise.

For a path $P$ and two vertices $a$ and $b$ in $P$, we define subpath $P[a,b]$
to be the path between $a$ and $b$ inside $P$ (including vertices $a$
and $b$).  For two vertices $u$ and $v$ in $V(G)$, we denote $\dist(u,v)$ to be
the length of their shortest path in the graph $G$, and $G[u,v]$ to be a
shortest path itself. For a subgraph $H$ of $G$, we denote $\dist_H(u,v)$ to be
the distance between vertices $u$ and $v$ in $H$.
For $\ell \in \nat$, the $\ell$-neighborhood of $H$ is $\{ v \in G \mid
\exists u \in H \text{ such that } \dist(u,v) \le \ell\}$.

We use standard graph notation (see for example~\cite{graphtheoryDiestel}).  In
particular, $V(G)$ and $E(G)$ are sets of vertices and edges respectively.  For
a set $X \subseteq V(G)$ we use $G[X]$ to denote the subgraph induced in $G$ by
$X$.  Whenever, we are given a planar graph $G$, we also assume that it is also
equipped with the \emph{combinatorial embedding} or an embedding on a sphere $\sph^2$. We refer
the reader to the textbook~\cite{klein2014optimization} for introduction to
planar graphs and formal definition of combinatorial embedding.

\subparagraph*{Hyperbolicity}
Let $(X,\dist)$ be a metric space. The Gromov product of $y, z \in X$ with respect to $w \in X$ is 
\[(y,z)_w = \frac 1 2 \left( \dist(w, y) + \dist(w, z) - \dist(y, z) \right).\]
The metric space $(X,\dist)$ is $\delta$-\emph{hyperbolic} if and only if for all $x,y,z,w \in X$ we have
\[(x,z)_w \ge \min \left( (x,y)_w, (y,z)_w \right) - \delta.\]
The hyperbolicity of $(X,d)$ is the minimum value $\delta$ for which it is
$\delta$-hyperbolic. We note that this definition is equivalent to the definition given in the introduction~\cite{BH99}.

Consider now a graph $G=(V,E)$. This induces the metric space on $V$ with the shortest
path distance. A \emph{geodesic} is defined as the vertex set of a shortest
path, and a \emph{triangle of geodesics} is a set of three geodesics between three
vertices in the graph; each side of the triangle is a shortest path. A 
triangle of geodesics is said to be \emph{$\delta$-slim} if each of its sides is contained in
the $\delta$-neighborhood of the union of the other two sides, that is, for any
shortest paths $P:a\rightarrow b; Q:b \rightarrow c; R:c\rightarrow a$ and any
$v\in V(P)$ we have $\min_{x\in V(Q)\cup V(R)} \dist_G (v,x) \leq \delta$.  The
slimness of $(X,d)$ is defined as the minimum $\delta$ for which all triangles
of $(X,d)$ are $\delta$-slim.

A simple graph $G=(V,E)$ can be also considered as a metric space in the
following sense: We can consider $X$ to be the $1$-dimensional cell complex
$\Cc(G)$ defined by $G$, and here we set each edge to be a unit length interval.
The distance of two points is then the shortest curve connecting them in
$\Cc(G)$, i.e., the geodesic distance. This metric space we denote by
$(\Cc(G),\rho_G)$. Note that when we restrict $\rho$ to the vertices, then we
get $\dist_G$. In the metric space $(V,\dist_G)$ a geodesic is defined as the
vertex set of a shortest path. These two metrics are very similar.

It is well-known that in so-called \emph{geodesic metric spaces} (such as
$(\Cc(G),\rho_G)$) the space is $\delta$-hyperbolic if and only if all triangles
are $\delta'$-slim where $\delta' =\Theta(\delta)$~\cite{gromov87,BH99}. Notice that if
$(\Cc(G),\rho_G)$ is $\delta$-hyperbolic, then so is $(V,\dist_G)$, since its
distance function is a restriction of $\rho_G$. On the other hand, when
$(V,\dist_G)$ is $\delta$-hyperbolic, then $(\Cc(G),\rho_G)$ is
$\delta+\Oh(1)$-hyperbolic.  Similarly, the maximum slimness of a triangle in
$(V,\dist_G)$ is within $1/2$ additive distance of the maximum slimness of a
triangle in $(\Cc(G),\rho_G)$.

We remark, that there is a crucial distinction regarding $(\Cc(G),\rho_G)$ and
$(V,\dist_G)$. It is well known (and easy to verify) that $(\Cc(G),\rho_G)$ have
hyperbolicity and slimness $0$ if and only if $G$ is a tree, and both
hyperbolicity and slimness is at least $1/2$ otherwise~\cite{BH99}. This is not
true when it comes to the hyperbolicity and slimness of $(V,\dist_G)$: when $G$
is a triangle graph, then $(V,\dist(G))$ has hyperbolicity and slimness~$0$. It
is known that the hyperbolicity of $(V,\dist(G))$ is $0$ if and only if $G$ is a
block graph~\cite{BandeltM86}, i.e., a graph where all $2$-connected components
are cliques.

Because of that, in the remainder of this paper we will consider the simple graph $G$ with the
metric $(V,\dist_G)$. We will assume that the triangles of $(V,\dist_G)$ are
$\delta$-slim. As a consequence of $\delta$-slimness of $G$, we also have that
$G$ is $\delta'=\Theta(\delta)$-hyperbolic.

\section{Structures in Planar \texorpdfstring{$\delta$}{delta}-Hyperbolic Graphs} \label{sec:structure}

Now, we propose several properties of planar $\delta$-hyperbolic graphs.
Suppose that we are given a $2$-connected planar graph
$G$ with fixed combinatorial embedding. Moreover we assume that $G$ is
$\delta$-hyperbolic, where $\delta$ is given as a parameter. Without loss of
generality, we can assume that $\delta \geq 1$, as every $0$-hyperbolic graph is
also $1$-hyperbolic.

\subsection{Local Geodesics}
A path or cycle (or subgraph) $P$ of $G$ is a \emph{$k$-local geodesic} if for
any $u,v \in V(P)$ with $d_P(u,v)\leq k$, we have $d_P(u,v)=d_G(u,v)$.

One can show that in $\delta$-slim graphs, the $\Oh(\delta)$-local
geodesic cycles have bounded length. This lemma is well known and we include the proof for completeness.

\begin{lemma}[cf., Theorem 1.13 in Chapter III.H of \cite{BH99}]
    \label{lem:local-geo-paths}
    Let $G$ be a $\delta$-slim graph. Let $a$ and $b$ be arbitrary
    vertices of $G$, and let $Q$ be a shortest path from $a$ to $b$ in $G$. If a path
    $P$ from $a$ to $b$ is $(8\delta+2)$-local geodesic, then $P$ is contained
    in the $4\delta$-neighborhood of $Q$.
\end{lemma}

\begin{proof}
    Let $x \in P$ be a vertex on $P$ with the maximum distance from $Q$. For
    the sake of contradiction, let us assume that $\dist(x,Q) > 4\delta$. This
    implies that $\dist_P(x,a)$ and $\dist_P(x,b)$ are both greater than
    $4\delta$. Therefore, there exist two distinct vertices $y$ and $z$ on $P$
    such that the subpaths $P[y,x]$ and $P[x,z]$ have length $4\delta+1$. Let
    $y',z' \in Q$ be the vertices on $Q$ that are closest to $y$ and $z$,
    respectively; see Figure~\ref{fig:local_geodesic}. Fix $G[y,y']$,
    $G[z,z']$ and $G[z,y']$ to be some shortest paths
    between $y,y'$, between $z,z'$ and between $z,y'$ respectively.

    \begin{figure}
        \centering
        \includegraphics{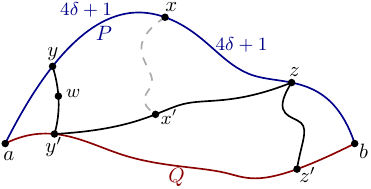}
        \caption{An $8\delta+2$-local geodesic $P$ stays within distance $4\delta$ of a shortest path $Q$ with the same endpoints.}
        \label{fig:local_geodesic}
    \end{figure}

    First, we will show that there exists $w \in G[y,y'] \cup
    G[z,z']$ with $\dist(x,w) \le 2\delta$. Consider the triangle
    $\Delta(y,y',z)$ determined by paths $P[y,z]$, $G[y,y']$, and
    $G[z,y']$. Observe that because $P$ is $(8\delta+2)$-local geodesic
    and $\dist_P(y,z) \le 8\delta+2$, this triangle is geodesic. Moreover,
    $x$ lies on $P[y,z]$, which is a side of this triangle. Therefore, by the
    $\delta$-slim triangle property, there exists a vertex $x'$ on either
    $G[y,y']$ or $G[z,y']$ with $\dist(x',x) \le \delta$. If $x' \in
    G[y,y']$, we are done as we can set $w=x'$. Therefore, from now on,
    let us assume $x' \in G[z,y']$.

    Next, consider the triangle $\Delta(z,z',y')$ that consists of paths
    $G[z,y']$, $G[z,z']$, and $Q[y',z']$. Observe that this is a
    triangle of geodesics, because $Q$ is a shortest path. Moreover, it holds that
    $x' \in \Delta(z,z',y')$ as it lies on $G[z,y']$. Therefore, by the $\delta$-slim triangle property, there exists $x''$
    on one of the opposite sides of this triangle with $\dist(x',x'') \le
    \delta$. By the triangle inequality, we have $\dist(x,x'') \le \dist(x,x') +
    \dist(x',x'') \le 2\delta$. If $x'' \in Q[y',z']$, then it means that
    $\dist(x,Q) \le 2\delta$, which contradicts the opening assumption
    $\dist(x,Q) > 4\delta$. Therefore, $x'' \in G[z,z']$, and we can
    simply select $w=x''$.

    Hence, there exists $w \in G[y,y'] \cup G[z,z']$ with
    $\dist(x,w) \le 2\delta$. Let us assume that $w \in G[y,y']$ (the
    remaining argument is symmetrical when $w \in G[z,z']$). 
    Next, we will show $\dist(x,y') < \dist(y,y')$. Indeed,
    \begin{align*}
        \dist(x,y') - \dist(y,y') & \le \left(\dist(x,w) + \dist(w,y')\right) -
        \left(\dist(y,w) + \dist(w,y')\right)\\
        & = \dist(x,w) - \dist(y,w) \\
        & \le \dist(x,w) - \left(\dist(y,x) - \dist(x,w) \right)\\
        & = 2 \cdot \dist(x,w) - \dist(x,y) \\
        & < 4\delta - 4\delta = 0.
    \end{align*}
    Since $y'$ is the closest vertex on $Q$ to $y$, this means that
    $\dist(x,Q) \le \dist(x,y') < \dist(y,Q)$. This contradicts the
    initial choice of $x$ and concludes the proof.
\end{proof}

By setting $a = b$ in the above lemma we get the following corollary.

\begin{corollary}\label{cor:local-geo-cycle}
    Any $(8\delta+2)$-local geodesic cycle in a $\delta$-slim graph has
    length at most $(8\delta+2)$ (and thus it is a geodesic cycle).
\end{corollary}
\begin{proof}
    Let $C$ be an $(8\delta+2)$-local geodesic cycle in a $\delta$-slim
    graph, and let $a \in C$ be an arbitrary vertex on $C$. For the sake of
    contradiction, let $x \in C$ be a vertex on $C$ such that $\dist_C(a,x) =
    4\delta+1$. Because $C$ is $(8\delta+2)$-local geodesic, it holds that
    $\dist(a,x) = \dist_C(a,x)$. If we set $a=b$ in \Cref{lem:local-geo-paths},
    we see that the cycle $C$ is contained in the ball of radius $4\delta$
    centered at $a$. Hence, we have $\dist(a,x) \le 4\delta$, which contradicts
    the choice of $x$.
\end{proof}

In hyperbolic graphs, geodesic separators are useful because they (unlike other separators) can preserve hyperbolicity in the following sense.

\begin{lemma}\label{lem:geod-sep}
    Let $G$ be a $\delta$-slim graph  and let $S$ be a geodesic subgraph
    of $G$. Let $\mathcal{C}$ be the set of connected components of $G -
    S$. For any $X \subseteq \mathcal{C}$ the graph $G[S \cup V(X)]$ is
    $\delta$-slim.
\end{lemma}
\begin{proof}
    Since $S$ is geodesic, it is also connected.
    Let $G' = G[S \cup V(X)]$. We begin by proving the following claim.
    \begin{claim*}\label{claim:local-sep}
        For every $a,b \in V(G')$ it holds that $\dist_{G'}(a,b) = \dist_G(a,b)$, i.e., $G'$ is a geodesic subgraph.
    \end{claim*}
    \begin{claimproof}
        Since $G'$ is a subgraph of $G$, we have $\dist_{G'}(a,b) \ge
        \dist_G(a,b)$. Therefore, it suffices to show that $\dist_{G'}(a,b) \le
        \dist_G(a,b)$.

        Fix $a,b \in V(G')$ and some shortest path $P$ between them in $G$. Let
        $x \in P \cap S$ be the vertex with minimum $\dist_P(a,x)$ and $y \in P
        \cap S$ be the vertex with minimum $\dist_P(y,b)$. Note that it is
        possible that such $x,y$ do not exist. In this case, observe that $P$ is
        a path in $G'$ and the claim follows.

        Consider a path $P' = P[a,x] \oplus S[x,y] \oplus P[y,b]$ (here $\oplus$
            denotes the join of the paths, and $H[a,x]$ is a shortest path in
        graph $H$). Observe that since $S$ is a separator and $V(X)$ are
        vertices in some connected components, it follows that $P[a,x]$ and
        $P[y,b]$ are paths in $G'$. Therefore, $P'$ is a path inside the graph
        $G'$.

        Finally, observe that because $S$ is a geodesic subgraph, it holds that
        $\dist_S(x,y) \le \dist_P(x,y)$. 
        Therefore, the length of $P'$ is at most the length
        of $P$, and the claim follows.
    \end{claimproof}

    Let $P,Q$, and $R$ be the sets of vertices of the three sides of a triangle of geodesics in $G'$.  We need to show that for any $x\in P$ there is some $y\in Q\cup R$ such that $\dist_{G'}(x,y)\leq \delta$. Notice first that $P,Q,R$ form shortest paths also with respect to $G$, and thus they form a triangle of geodesics also with respect to $G$. Thus there exists $y_G\in Q\cup R$ such that $\dist_G(x,y_G)\leq \delta$. It follows that $y_G\in V(G')$ as well, and since $G'$ is geodesic, we have $\dist_{G'}(x,y_G)=\dist_G(x,y_G)\leq \delta$. Consequently $y=y_G$ satisfies the desired property.
\end{proof}

\subsection{Fillings and greedy fillings}

\begin{definition}[$k$-filling]
    Let $C$ be a cycle in a graph $G$, and let $k > 0$ be a positive
    integer. We say that a $2$-connected planar subgraph $H$ of $G$ is an
    \emph{$k$-filling of $C$} if $H$ admits a planar embedding where all faces except $C$ have size at most $k$.

    The \emph{area} of a $k$-filling is defined as the number of edges in $H$.
\end{definition}

Note that when $|C| \leq k$, then $C$ trivially admits a
$k$-filling of area at most $k$ by taking $H = C$. The notion of a
$k$-filling is usually defined in the general context of metric spaces
(not just planar graphs). For these more general definitions, it is true that
when $G$ is $\delta$-hyperbolic, then the minimum $k$-filling of any loop
$C$ has area $\Oh(\delta |C|)$ (see Proposition 2.7 in Chapter III.H
of~\cite{BH99}). Next, based on~\cite{BH99} we present a simple procedure to construct
fillings in near-linear time for face cycles of planar graphs.

\subparagraph*{Greedy Filling Procedure}

We begin by describing the following \emph{greedy filling procedure}. This
procedure is given a planar $2$-connected graph $G$ that is $\delta$-hyperbolic. 
For a given face cycle $C$ of $G$, it constructs a graph
$H$ that is a \emph{greedy filling} of $C$. We fix a combinatorial embedding of $G$, i.e., an embedding on the unit sphere $\sph^2$.

The procedure is as follows: Initially, we set $H = C$, $C_1 = C$, and $G_1 =
G$. The procedure is iterative, and $i$ denotes the current iteration (initially
set to $1$). If a cycle $C_i$ is $10\delta$-local geodesic, we add $C_i$ to $H$
and terminate.

\begin{figure}
    \centering
    \includegraphics[width=\textwidth]{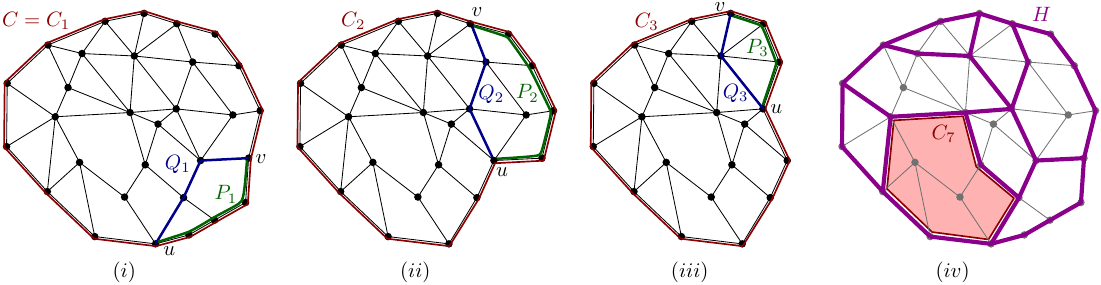}
    \caption{Greedy filling for the face with boundary cycle $C$. (i)-(iii) depicts the first few steps of the filling. (iv) shows the final filling $H$. The greedy filling procedure terminated when it found $C_7$, which is a  geodesic cycle.}
    \label{fig:greedyfilling}
\end{figure}

Otherwise, $C_i$ is not $10\delta$-local geodesic, and we find a path $P_i$ that
is a minimal subpath of $C_i$ that is not a shortest path in $G_i$, see \cref{fig:greedyfilling}. Note that
$|P_i| < 10\delta$. Let $u,v \in C_i$ denote the endpoints of $P_i$, and let
$Q_i$ be a shortest path from $u$ to $v$ in $G_i$. Observe that, by the
minimality of $P_i$, the paths $P_i$ and $Q_i$ are vertex disjoint (except for
the endpoints $u,v$). Next, we add the path $Q_i$ to the filling graph $H$.
Then, we shorten the cycle based on $Q_i$. Namely, we set $C_{i+1} = (C_i
- P_i) \cup Q_i$. We set $G_{i+1}$ to be the subgraph of $G_i$ with the
boundary $C_{i+1}$ (the side that does not contain the internal vertices of $P_i$)
and continue the procedure with $G_{i+1}$ and $C_{i+1}$. 
Notice that $P_i$ is a geodesic path separator of $G_i$, thus by \Cref{lem:geod-sep} the $\delta$-slimness of $G$ is retained in each $G_i$. 

This concludes the description of the greedy filling procedure. We say that the
graph $H$ returned by this procedure is a \emph{greedy filling} of $C$. Next, we show
that this procedure can be executed efficiently, and that it indeed returns an
$\Oh(\delta)$-filling of $C$ with small area.

\begin{lemma}[Greedy-Filling Procedure]\label{lem:compute_filling}
    Let $G$ be a $2$-connected planar $\delta$-slim graph on $n$ vertices, and let $C$ be a face of $G$.
    There exists a data structure that 
    for a cycle $C$ inside $G$ given to it as a query, can construct a
    greedy-filling $H_C$ of $C$ inside $G$. The
    returned graph $H_C$ is $21\delta$-filling of $C$ and has area
    $\Oh(\delta |C|)$.
    The initialization time of the data structure is $\Oh(\delta^2 n \log^4 n)$.
    The query time is $\Oh(\delta |C| \log\log n)$ if a cycle $C$ is given as a query.

    The data structure can also answer shortest-path queries about distances of
    length at most $10\delta$ in $\Oh(\delta+\log\log(n))$ time.
\end{lemma}
\begin{proof}
    During the initialization, we preprocess the planar graph $G$ and create a
    shortest-path oracle from Thorup~\cite{Thorup04}. After
    $\Oh(\delta^2 n \log^4 n)$ time preprocessing, this data structure allows us to check in
    $\Oh(\delta+\log\log n)$ time if a given pair of vertices are at a distance at most
    $10\delta$ in $G$. (In case of very small $\delta$, see~\cite{DBLP:journals/talg/KowalikK06} for a faster distance oracle.)

    During a query, our data structure invokes the \emph{greedy-filling procedure}.
    Observe that in each iteration, the length of the cycle $C_i$ decreases.
    Therefore, the total number of iterations is $\Oh(|C|)$. In each
    iteration, we need to find a path $P_i$ and find a shortest path between the
    endpoints of $P_i$. This step takes $\Oh(\delta+\log\log n)$ time as it is a query to
    distance oracle. Therefore, the total runtime of the algorithm is
$\Oh(\delta |C|\log\log n)$. 

    It remains to argue about the correctness of the algorithm, i.e., that $H$
    is a $21\delta$-filling of $C$. Notice that after $i$ iterations
    the faces of $H$ (with the exception of $C$) are $P_i \cup Q_i$ or $C_\ell$ (where
    $\ell$ is the number of the last iteration). The size of $P_i \cup Q_i$ is
    at most $20\delta$. After the last iteration, $C_\ell$ is a $10\delta$-local
    geodesic of the $\delta$-slim graph $G_\ell$.
    By~\cref{cor:local-geo-cycle}, $C_\ell$ has size at most $10\delta$.
    Hence, $H$ is a $20\delta$-filling of $C$. The area of $H$ is at most
    $21\delta |C|$ because, in each of the $|C|$ iterations, we add at most
    $20\delta$ edges to the graph $H$.
\end{proof}

We will call the data structure of \cref{lem:compute_filling} the \emph{greedy
filling data structure}. 

A \emph{filling face} of the greedy filling $H$ of $C$ is a face $F$ of $H$ that is not $C$. We say that a cycle $\gamma$ in $G$ \emph{interacts} with a filling face $F$ of $H$ (where $F$ is now thought of as a region of $\sph^2$) if $F$ does not lie between $C$ and $\gamma$; that is, either $\gamma$ intersects the interior of the region $F$ or $\gamma$ separates the interior of $F$ from $C$. We prove the following variant of an isoperimetric inequality for cycles in the plane graph where a fixed greedy filling serves as the measure of `area'. See ~\cite[Chapter III.H Proposition 2.7]{BH99} for a similar claim about fillings.

\begin{lemma}[Isoperimetric Inequality for Greedy Filling]\label{lem:isoperimetry}
    Let $G$ be $2$-connected planar graph\footnote{We note that the lemma does not use the hyperbolicity of $G$ directly, only indirectly, namely in the fact that the greedy filling $H$ turned out to be a $k$-filling.} and let $H$ be
    some \emph{greedy $k$-filling} for some face $C$ of $G$ inside some subgraph\footnote{We need to consider a subgraph $G'$ for technical reasons; in a typical application one should think of $G'=G$.} $G'$ of $G$. Then any simple cycle in $G$ of
    length $\ell$ interacts with at most $(k+1)\cdot\ell$ faces of $H$. Moreover, if the cycle is in $H$, then it interacts with at most $\ell$ faces.
\end{lemma}

\begin{proof}
    Let $F_1,F_2,\ldots$ be the
    filling faces of $H$ in the order that they are added to $H$ in the
    greedy filling procedure. Note that filling face $F_i$ is a simple cycle in $G$,
    and $Q_i \subseteq F_i$ is the shortest path replacing the path $P_i$.
    Consider a plane embedding of $G$ where $C$ is fixed to be the outer face. As a result, each cycle $\psi$ of $G$ defines a bounded and an unbounded open region.
    
    Consider now a simple cycle $\gamma$ of length $\ell$ in $G$. Let $\Ff_{\gamma}$
    be the set of filling faces of $H$ whose interior is intersected by $\gamma$ or share an edge with $\gamma$.
    Let $\Ff_{\Int}$ be the set of faces disjoint from $\gamma$ that $\gamma$ separates from $C$,
    i.e., the set of filling faces contained entirely in the interior of the bounded region of $\gamma$.
    Our goal now is to show that $|\Ff_\gamma|+|\Ff_\Int|\leq (k+1)\ell$. First,
    notice that $\gamma$ must have at least one edge in the interior of each
    $F\in \Ff_{\gamma}$, thus $|\Ff_{\gamma}|\leq |\gamma| =\ell$.

    It remains to bound $|\Ff_{\Int}|$. Consider the region given by the union of the faces
    $R=\bigcup_{F\in \Ff_{\gamma}} F$. The region $R$ is not necessarily
    connected and $\bd R$ may have several components. Let $\Psi$ be the set of
    connected components (cycles) in $\bd R$ such that each $\psi\in \Psi$ is
    the boundary shared between some connected component of $\bigcup_{F\in
    \Ff_{\Int}} F$ and $R$. Notice that the cycles of $\Psi$ interact 
    with exactly the set of faces $\Ff_\Int$.
    Moreover, observe that $\sum_{\psi\in \Psi} |\psi| \leq k\ell$. Indeed, $\bigcup_{\psi\in \Psi} \psi\subset \bd R$ and $|\bd R|\leq \sum_{F\in \Ff_{\gamma}} |F| \leq k \ell$.
    
    We now prove a more general statement for any set $\Phi$ of edge-disjoint simple cycles of $H$ whose bounded regions are pairwise disjoint. 
    The set $\Psi$ defined in the previous paragraph is such a set. We will prove by induction on $L(\Phi):=\sum_{\phi\in \Phi} |\phi|$ that altogether the cycles of $\phi$ interact with at most $L(\Phi)$ filling faces of $H$.
    
    If $|\phi|=  3$ for each $\phi\in \Phi$, then each of them interacts
    with exactly one face, and altogether they interact with $L(\Phi)/3<L(\Phi)$ faces. 
    Suppose now that there is some $\phi\in \Phi$ with $|\phi|\geq 4$ and it interacts with at least two faces of $H$. Let $F_j$ be the filling face with the minimum index $j$ that interacts with $\phi$. Recall that $F_j$ is bounded by the paths $P_j$ and the shortest path $Q_j$. Since $\phi$ separates $F_j$ from $C$ (i.e., $F_j$ is inside the bounded region of $\phi$), and $j$ is the smallest index of a face with this property, we know that $F_1,\dots,F_{j-1}$ lie outside $\phi$. Since all edges of $P_j$ are edges with $F_j$ on one side and some $F_t$ with $t<j$ on the other, we have that $P_j\subset \phi$. Consider now the closed walk $\phi'$ obtained by replacing $P_j$ with $Q_j$ in $\phi$; this move essentially removes $F_j$ from the bounded region of $\phi$. (Note that since $Q_j$ is in the bounded region of $\phi$, it cannot contain vertices from the cycles $\Phi\setminus \{\phi\}$.) We can decompose the remainder of $\phi'$ into several simple cycles that are edge-disjoint and their bounded regions are pairwise disjoint from each other (as well as from all other cycles $\Phi\setminus \{\phi\}$). Let $\Phi'$ be the cycle collection where we replace $\phi$ in $\Phi$ with these new cycles. Since $|Q_j|<|P_j|$, we have $L(\Phi')<L(\Phi)$. On the other hand, the cycles of $\Phi'$ interact with the same faces as the cycles of $\Phi$ except $F_j$. Thus $\Phi$ interacts with $F_j$ and (by induction) at most $L(\Phi')\leq L(\Phi)-1$ other filling faces, so altogether with at most $L(\Phi)$ filling faces.

    This concludes the proof of the first part of the lemma statement: we have $|\Ff_{\gamma}|\leq |\gamma| =\ell$
    and by the above argument $|\Ff_\Int|\leq L(\Psi)\leq k\ell$.

    For the second part of the lemma statement, if the cycle $\gamma$ is in $H$, then $\gamma$ itself is a simple cycle of $H$ and $\Psi = \{\gamma\}$ by the definition of $F_\gamma$. Moreover, we can bound $L(\Psi)$ by $\ell$. Following the above proof, we obtain that $\gamma$ intersects with at most $\ell$ faces.
\end{proof}

\subsection{Geodesic Cycle Separators}
To analyze the structure of (geodesic) cycle separators, we start with two definitions. 
Let $G$ be a $2$-connected graph with a fixed embedding on $\sph^2$. Let $F$ be a face of $G$ in this
embedding, and let $C$ be a cycle where $C$ is not the boundary of $F$. Notice
that in the embedding $C$ cuts the plane into two parts, one containing the face
$F$ and the other not containing $F$. We say that $C$ \emph{covers $k$ vertices
opposite $F$} if the part not containing $F$ has $k$ vertices of $G$ in its
interior.

Let $G$ have a fixed embedding on the unit sphere $\sph^2$. For a parameter $\alpha \in
(0,1)$, we say that a cycle $C$ \emph{splits} the set of vertices
$V(G)\setminus V(C)$ with \emph{balance} $\alpha$ into two regions $R_1,R_2 \subseteq V(G) \setminus
V(C)$ if $|R_1|$ and $|R_2|$ are both at most $(1-\alpha) \cdot |V(G)|$. Observe that
such a cycle is also a separator of balance at least $\alpha$, since each
connected component of $G-V(C)$ is either contained in $R_1$ or $R_2$. The converse is not true: a cycle separator of balance $\alpha$ need not induce a split of balance $\alpha$, as all connected components of $G-C$ may be in $R_1$.
More generally, a subgraph $H$ of $G$ has \emph{split balance} $\alpha$ if in each region defined by the faces of $H$ there are at most $(1-\alpha)n$ vertices from $V(G)\setminus V(H)$.

\begin{figure}[t!]
    \centering
    \includegraphics
    {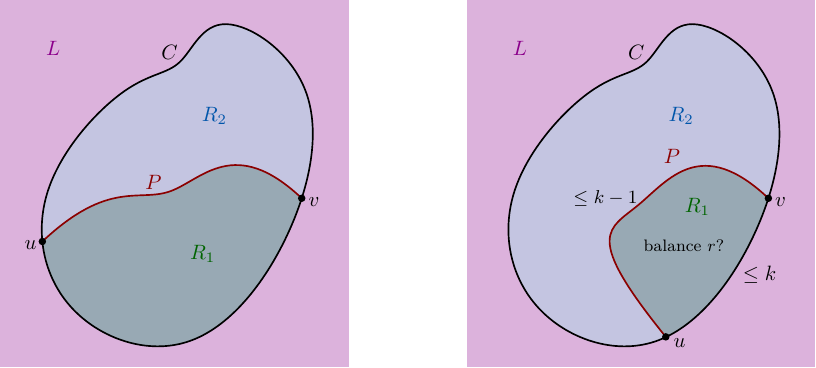}
    \caption{Left: the cycle shortening process from
    \cref{claim:cycle_to_geodesiccycle_aux}. Right: In 
\cref{claim:cycle_to_geodesiccycle_aux2}, by finding two vertices $u$ and $v$
with a shorter path $P$ outside the cycle, we can create two cycles, one with a
length $\leq 2k$. If that cycle has split balance $\geq r$, then we return it. Otherwise, we
return the second cycle and lose at most a constant factor in the split balance.}
    \label{fig:cycle-shortening}
\end{figure}

\begin{lemma}\label{lem:cycle_to_geodesiccycle}
    Let $\alpha \in (0,1)$. Let $C$ be a cycle separator of length $\ell$ that splits the $n$ vertices
    of a planar graph $G$ with balance $\alpha$. If $\alpha n \ge
    \ell 4^\ell$, then there is a \emph{geodesic} cycle that splits
    $G$ with balance at least $\alpha/2^{\Oh(\ell)}$. Given $G$ and $C$, such a geodesic
    cycle can be found in 
         $\Oh(\ell^2 n\log n)$ time.
\end{lemma}
\begin{proof}
    We start with the following subroutine.

    \begin{claim}\label{claim:cycle_to_geodesiccycle_aux}
         Assume $\alpha n \ge 4\ell$, and let $C$ be a non-geodesic cycle of length
         $\ell$ that splits $G$ with balance $\alpha$.  In 
         $\Oh(\ell n\log n)$ time, we can find cycle $C'$ with length $\ell-1$ that splits $G$
         with balance at least $\alpha/4$.
    \end{claim}
    \begin{claimproof}
        Since $C$ is not a geodesic, there exist vertices $u$ and $v$ in $C$
        such that $C[u,v]$ is not a shortest path in $G$. Let $u,v$ be such a pair where $|C[u,v]|$ is 
        minimized. Let $P$ be a shortest path between $u$ and $v$ in $G$.
        We construct two cycles: $C_1 \coloneqq P \cup C[u,v]$ and $C_2
        \coloneqq P \cup (C - C[u,v])$. We select $C'$ to be the cycle
        among $C_1,C_2$ with higher split-balance.

        This concludes the description of the procedure
        (see~\cref{fig:cycle-shortening}). Observe that the
        running time follows, since we can compute $\ell$ shortest path trees in
        $\Oh(\ell n\log n)$ time, then compute the number of vertices that
        are split by a given separator in linear time. It remains to prove
        the correctness.

        Both $C_1$ and $C_2$ are cycles of smaller length than $C$,
        so it remains to argue about the split balance of~$C'$.
        Let $L,R \subseteq V(G) \setminus V(C)$ be the vertices split by cycle
        $C$. Note that, by the minimality, the path $P - \{u,v\}$ is
        contained in either $L$ or $R$. Without loss of generality, assume that
        it is contained in $R$.

        Let $R_1,L_1$ and $R_2,L_2$ be the vertices split by $C_1$ and $C_2$
        respectively. Observe that $R_1,R_2 \subseteq R$ and thus they both have size at most $|R|\leq (1-\alpha)n$. Therefore, it remains to show that at least one of $|L_1|$ and $|L_2|$
        is at most $(1-\alpha/4)n$. Since $L_1 \subset V(G)\setminus R_2$ and $L_2 \subset V(G)\setminus R_1$, it is sufficient to show that $\max(|R_1\cup C|, |R_2\cup C|)\leq \alpha/4$.

        Observe that $R_1$ and $R_2$ are disjoint, and $R_1 \cup R_2 \cup (P
        - \{u,v\}) = R$, and also $|R|\geq \alpha n$. Therefore, $|R_1| + |R_2| = |R| - |P| + 2$. The
        larger of $R_1$ or $R_2$ has size at least $|R|/2 - |P|/2+1 \ge \alpha n/2
        - \ell/2 +1 \ge \alpha n /4$.
    \end{claimproof}
    
    Now, we apply~\cref{claim:cycle_to_geodesiccycle_aux} exhaustively on $C$.
    The process terminates when $C$ is a geodesic cycle. Each iteration
    of~\cref{claim:cycle_to_geodesiccycle_aux} decreases the length of the
    cycle; therefore, we will perform less than $\ell$ iterations. The initial
    assumption, $\alpha n \ge \ell 4^\ell$, guarantees that $\alpha$ is
    large enough in each iteration. Therefore, the final cycle has split balance
    at least $\alpha/4^\ell$.

    The running time bound follows, as each iteration takes 
    $\ell n\log n$
    time, and the number of iterations is at most $\ell$.
\end{proof}

\begin{lemma}\label{lem:cycle_to_geodesiccycle2}
    Let $r \in (0,1)$ and let $t$ be an integer.
    Assume that there exists a cycle $C$ of length $t$ in a planar $\delta$-slim graph $G$ with $n$
    vertices that
    splits $G$ with balance at least $2rt$, where $rn \ge 2^{\Omega(\delta)}$. Then, there is a
    geodesic cycle separator with split balance $r/2^{\Oh(\delta)}$ in $G$. Given $G$ and $C$, such a geodesic
    cycle can be found in $\Oh(\delta n\log\log n)$ time if we are given an
    access to the greedy filling data structure from~\cref{lem:compute_filling}.
\end{lemma}

\begin{proof}
    We start with the following claim.

    \begin{claim}\label{claim:cycle_to_geodesiccycle_aux2}
        Let $\alpha,r \in (0,1)$ with $r < \alpha$ and $k$ be an integer such that $rn \ge k$. Let
        $C$ be a cycle with split balance $\alpha$ that is not a $k$-local
        geodesic. Then, one of the following exists:
        \begin{itemize}
            \item A cycle with split balance at least $r$ and length at most $2k$, or
            \item A cycle with split balance at least $(\alpha - 2r)$ and shorter length than $C$.
        \end{itemize}
    \end{claim}
    \begin{claimproof}
        Because $C$ is not a $k$-local geodesic, there exist two vertices $u,v
        \in C$ with $\dist_C(u,v) \le k$ and $\dist_G(u,v) < \dist_C(u,v)$. Let
        $u,v$ be a closest (w.r.t. $\dist_C$) pair of such vertices, and let $P$ be a
        shortest path between $u$ and $v$ in $G$
        (see~\cref{fig:cycle-shortening}).

        Consider the cycle $C_1 \coloneqq P \cup C[u,v]$. Note that by the
        minimality of $|C[u,v]|$, this is a simple cycle of length at most
        $2k$. If $C_1$ splits $G$ with balance $r$, then we can just return $C_1$ and
        satisfy the first point of the statement. Otherwise, consider a cycle
        $C_2 \coloneqq P \cup (C - C[u,v])$. We will show that $C_2$
        splits $G$ with a balance of at least $(\alpha - 2r)$, which will
        conclude the proof, as $C_2$ is shorter than $C$.

        Let $L,R \subseteq V(G) \setminus V(C)$ be the set of vertices split by
        $C$. Observe that $P- \{u,v\}$
        is contained in either $R$ or $L$. Without loss of
        generality, we can choose the regions such that $P - \{u,v\}
        \subseteq R$. 

        Let $R_1,L_1$ and $R_2,L_2$ be the vertices split by $C_1$ and $C_2$
        respectively. 
        Observe that $R_1,R_2\subset R$.  
        We also have $|R|\leq(1-\alpha)n$ and thus both $R_1$ and $R_2$ have size less than $(1-\alpha)n<(1-r)n$. Since $C_1$ has balance less than $r$, we must have $|L_1|>(1-r)n$. It follows that $|R_1|<rn$. It remains to prove that $L_2$ is not too big: indeed, $|L_2|<|L|+|R_1|+ C[u,v]| \leq (1-\alpha)n + rn +k <(1-\alpha+2r)n$. Hence, $C_2$ splits $G$ with balance at least $(\alpha-2r)$.
    \end{claimproof}

    We exhaustively invoke~\cref{claim:cycle_to_geodesiccycle_aux2}. Initially,
    we set $C$ to be the original cycle, $\alpha$ to be the split balance of $C$,
    $k = 10\delta$ (as in~\cref{cor:local-geo-cycle}), and $r$ to be the
    original $r$.

    If, at some point, \cref{claim:cycle_to_geodesiccycle_aux2} returns a cycle
    of length at most $2k$ and split balance $r$, then we terminate and use~\cref{lem:cycle_to_geodesiccycle}. This returns a
    geodesic cycle with a split balance of $r/2^{\Oh(\delta)}$. Otherwise,
    \cref{claim:cycle_to_geodesiccycle_aux2} progressively returns cycles of
    smaller lengths, until it terminates with a $10\delta$-local geodesic
    cycle $C'$, which is by \cref{cor:local-geo-cycle} a geodesic cycle. In each iteration the split balance changes by at most $2r$, and the number of iterations is at most $t-1$. Hence, the final geodesic cycle $C'$
    has a split balance more than~$r$.

    In order to achieve our running time, observe that \cref{claim:cycle_to_geodesiccycle_aux2} is essentially a greedy filling step on the cycle $C$, so we can use the same data structure and algorithm as in \cref{lem:compute_filling}, and use that we may assume that $t<n$.    
\end{proof}

\section{Separator for Planar Hyperbolic Graphs} \label{sec:sep}

In this section we provide a proof for~\cref{thm:separator}. First, we note that if $G$ has no cycles, then it has a $1/2$-balanced separator consisting of a single vertex, which is a shortest path of length $0$. Hence, we may assume that $G$ has a cycle. We will also assume that $n$ is large enough. If $n$ is below some constant threshold, then since $G$ has a cycle, by \cref{lem:cycle_to_geodesiccycle} it has a geodesic cycle of length at least $3$, which can be used as a geodesic cycle separator of balance $\frac{1}{2^{\Oh(\delta)}\log n}$ by increasing the hidden constant.
The same argument shows that without loss of generality, we may assume that $\delta$ is small, i.e., $\delta=\Oh(\log n)$, as otherwise the balance requirement for cycles is satisfied by any cycle.

We start our arguments by
proving the separator in the case when the input graph is $2$-connected. Namely, we show the following.

\begin{lemma}\label{lem:separator-2-connected}
    Let $G$ be a $2$-connected planar $\delta$-hyperbolic graph on $n$ vertices.
    Then $G$ has a geodesic path separator $X$ of size $|X|=\Oh(\delta^2 \log
    n)$ and balance $1/7$ or $G$ has a geodesic cycle separator $Y$ of
    size $|Y| = \Oh(\delta)$ and balance $2^{-\Oh(\delta)}/\log n$. Given $G$,
    such a separator $X$ or $Y$ can be computed in $\Oh(\delta^2 n\log^4 n)$ time.
\end{lemma}

Set $r=\frac{1}{c\delta^3 \log n}$, where $c$ is a large enough constant that will be specified later. 

Next, for the fixed embedding of $G$ in the
unit sphere $\sph^2$ (where vertices are identified with points on the sphere) we will be able to associate integer weights with the faces
of some subgraph $G'$ of $G$. In our setting the weight of a (closed) face region $F$ of $G'$ is
the number of vertices in the original graph $G$ that lie in the
interior of $F$. Note that when referring to the vertex count or
weight of a particular region, we include the sum of face weights, i.e., we are
thinking of the number of vertices of the original graph $G$ that lie in the region. This influences our definition of split balance, that is, we are always considering split balance of cycles with respect to the original graph $G$.

Our algorithm will begin by setting up a data structure for greedy fillings in
$\Oh(\delta^2 n \log^4 n)$ time as described in~\cref{lem:compute_filling}. We note that setting up this data structure is the dominant term in the running time of the separator computation.
We then proceed with the following iterated greedy filling procedure.
Recall that a filling face of a greedy filling $H$ for a cycle $C$ is any face of $H$ that is not the face bounded by~$C$.

\subparagraph{Iterated greedy filling:}

Let $G_1=G$, and assign a weight of $0$ to each face. We repeat the following
steps, starting at $i=1$: Let $F_i$ be a face cycle of maximum size in $G_i$. If $|F_i|\leq 20\delta$ then terminate; otherwise let $H_i$ be a greedy
$10\delta$-filling of $F_i$ in $G_i$. Since $|F_i|>20\delta$, we know that $F_i$ is not a $10\delta$-local geodesic cycle and thus $H_i\neq F_i$. If the boundary cycle of some filling face of $H_i$ has balance at least $r$ in $G$, then terminate.  If
each filling face of $H_i$ has fewer than $rn$ vertices of
$G$ in its interior, then terminate. Otherwise, there is a filling face $F$ in
$H_i$ such that $F$ has more than $(1-r)n$ vertices in its interior.
Then let $G_{i+1}$ be the subgraph of $G_i$ induced by the vertices of $G_i$ inside and on the boundary of $F$, i.e., the closure of $\sph^2\setminus F$ now becomes a face in $G_{i+1}$ assigned weight
$|V(G)\setminus F|$, which is by definition less than $rn$. Notice that $|F_i|>20\delta$ and the procedure assigns positive weights only to filling faces that have length at most $20\delta$. Consequently, the fillings are made only for faces $F_i$ of weight $0$, and in particular, only for faces of the original graph $G$.

We note that the above definition of $G_i$ preserves $2$-connectivity, thus all $G_i$ are $2$-connected.
Observe that the total number of vertices in $G_i$ plus the total face weight of
$G_i$ remains $n$ throughout the procedure. Since the faces $F_i$ on which we do
greedy fillings are distinct faces of $G$, their total length is
$\Oh(n)$, thus by \cref{lem:compute_filling} the entire procedure (including the
initialization of the data structure) can be run in $\Oh(\delta n\log n)$ time.

The possible outcomes of the above procedure are the following:

\begin{itemize}
    \item[] \outcome{1}: The boundary cycle of some filing face $F$ has split balance at least $r$.
    \item[] \outcome{2}:  $|F_i|\leq 20\delta$. This means that the longest face
        cycle of $G_i$ has length $\Oh(\delta)$. Consequently, all faces of $G_i$
        have length $\Oh(\delta)$, and an integer weight between $0$ and $rn$.
    \item[] \outcome{3}: Every filling face of $H_i$ has fewer than $rn$ vertices in its interior. 
\end{itemize}

It remains to find a suitable separator in each of these possible outcomes.
Notice that in case of \outcome{1} we can directly invoke
\Cref{lem:cycle_to_geodesiccycle}. 

Now, we analyse the other outcomes and show
that in each of them a suitable separator can be found. This will conclude the
proof of~\cref{lem:separator-2-connected}.

\subsection{\outcome{2}}\label{sec:outcome-2}

\begin{lemma}[\outcome{2}]\label{lem:out-1}
    Let $G'$ be the face-weighted graph generated by the iterated greedy filling of~$G$. 
    If each face of $G'$ has length $\Oh(\delta)$, then there is a geodesic
    cycle separator of balance at least $\frac{1}{2^{\Oh(\delta)}\log n}$ in $G$,
    and it can be computed in $\Oh(\delta^2 n \log^3 n)$ time
    (given access to the data structure from~\cref{lem:compute_filling}).
\end{lemma}

\begin{proof}
Observe that by construction $G'$ has total weight $n$ (i.e., the number of vertices plus the weight of all faces is equal to the number of vertices of $G$).
Moreover, each weight is an integer between $0$ and $rn$.

\begin{claim}\label{cl:initial_sep}
There is a $\frac 12$-balanced separator $S$ in $G'$ of size $\Oh(\delta^2\log n)$ such that $G'[S]$ has $\Oh(\delta\log n)$ connected components.
Moreover, such a separator can be computed in $\Oh(\delta^2 n \log^3 n)$ time.
\end{claim}

\begin{claimproof}
By Corollary~\ref{cor:tw-hyper-sep}, since $G$ is a planar $\delta$-hyperbolic graph, a $\frac 12$-balanced separator $S_G$ of $G$ of size $\Oh(\delta\log n)$ can be computed in $\Oh(\delta^2 n \log^3 n)$ time.
We will now turn the separator $S_G$ of $G$ into a requested separator $S$ of $G'$.
Initially, we let $S=S_G\cap V(G')$.
Then, for each vertex $v\in S_G\setminus V(G')$, we pick the face of $G'$ that covers $v$ and add all vertices of its boundary cycle to $S$.
Clearly, $S$ is a separator of $G'$ with balance (defined with respect to $G$) at least $\frac 12$.
The size of $S$ is $\Oh(\delta^2\log
n)$, because each vertex of $S_G\setminus V(G')$ has been replaced by at most $\Oh(\delta)$ vertices of a face cycle of $G'$.
Moreover, $G'[S]$ has at most $|S_G|=\Oh(\delta\log n)$ connected components.
The procedure that turns $S_G$ into $S$ can be clearly performed in $\Oh(n)$ time.
\end{claimproof}

\begin{figure}[t]
    \centering
    \includegraphics{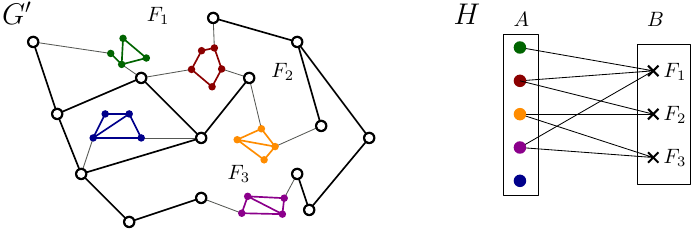}
    \caption{The graph $G'$ with separator vertices $S$ marked with black circles. The components of $G'-S$ and faces of $G'$ incident to at least two such components form the two parts of the auxiliary graph $H$.}
    \label{fig:auxgraph}
\end{figure}

For any separator $S$ in $G'$, we define an auxiliary bipartite graph $H=H(S)$ as follows.
The vertices of $H$ are partitioned into two sets $A$ and $B$.
The vertices $a\in A$ correspond to the connected components $X_a$ of $G'-S$, and the vertices $b\in B$ correspond to the faces $F_b$ of $G'$ that are incident to at least two distinct components of $G'-S$.
There is an edge $ab$ in $H$ whenever the corresponding component $X_a$ and the corresponding face $F_b$ share a vertex.
See \cref{fig:auxgraph} for an illustration.

\begin{claim}\label{cl:aux_forest}
Let $S$ be a separator of $G'$ such that $G'[S]$ has $k$ connected components.
Then $G'$ has a separator $S'$ such that $S\subseteq S'$, $|S'|\leq|S|+\Oh(k\delta)$, and the graph $H(S')$ is a forest.
Moreover, such a separator $S'$ can be constructed from $S$ in $\Oh(n)$ time.
\end{claim}

\begin{figure}
    \centering
    \includegraphics{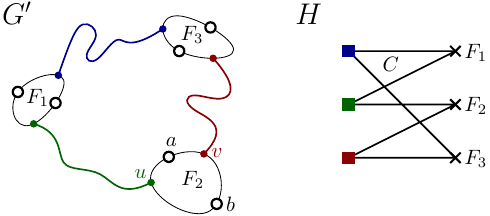}
    \caption{Illustration of the proof of \cref{cl:aux_forest}. By the choice of $S'$ the vertices $a$ and $b$ must be in the same component of $G'[S']$, but this is not possible as the union of the red, green, and blue components of $G'-S'$ separate $a$ and $b$.}
    \label{fig:aux_cycle}
\end{figure}

\begin{claimproof}
We construct $S'$ by the following greedy procedure.
Initially, let $S'=S$.
While $G'[S']$ has at least two connected components $X$ and $Y$ such that some face $F$ of $G'$ contains a vertex from both $X$ and $Y$, we add the face cycle of $F$ to $S'$.
Since the number of connected components of $G'[S']$ decreases with each face cycle added to $S'$, we add fewer than $k$ face cycles in total, each of length at most $\Oh(\delta)$.
Thus $|S'|\leq|S|+\Oh(k\delta)$, and $G'[S']$ has at most $k$ connected components.
Clearly, this greedy procedure can be performed in $\Oh(n)$ time.

We claim that the graph $H(S')$ is a forest.
Suppose to the contrary that there is a cycle $C$ in $H(S')$.
Let $A$, $B$, $X_a$, and $F_b$ be as in the definition of $H(S')$.
Consider a vertex $b\in V(C)\cap B$.
Since the corresponding face $F_b$ is incident to two distinct components of $G'-S'$ (say, one containing a vertex $u$ and another containing a vertex $v$), at least two vertices $a$ and $b$ in $S'$ must separate these
two components along the face cycle of $F_b$, i.e., without loss of generality the
vertices $a,u,b,v$ appear in this order along the face cycle of $F_b$; see Figure~\ref{fig:aux_cycle}. Notice that
$a$ and $b$ must be in the same component of $S$ as otherwise the boundary cycle
of $F$ would have been added to $S'$. On the other hand, notice that the union
of the components corresponding to the $A$-side vertices of $C$ form a separator
of $G'$ where $a$ and $b$ are in different components; in any fixed embedding, the
separation is demonstrated by a closed curve made up of paths inside these
components connecting the two neighboring faces along $C$ and curves inside the
faces corresponding to the $B$-vertices of $C$. Since the separator separates
$a$ and $b$, and it is disjoint from $S'$, it follows that $a$ and $b$ cannot
be in the same component of $G'[S']$, which is a contradiction.
\end{claimproof}

Combining \cref{cl:initial_sep} and \cref{cl:aux_forest}, we can find, in $\Oh(\delta^2 n \log^3 n)$ time, a $\frac 12$-balanced separator $S$ in $G'$ such that $|S|=\Oh(\delta^2\log n)$ the graph $H(S)$ is a forest.
Let $H=H(S)$, and let $A$, $B$, $X_a$, and $F_b$ be as in the definition of $H(S)$.

We assign weights to the vertices of $H$ as follows.
The weight of a vertex $b\in B$ is simply the weight of the corresponding face $F_b$ of $G'$.
The weight of a vertex $a\in A$ is defined as the number of vertices in the component $X_a$ corresponding to $a$ plus the total weight of the faces of $G'$ that are incident only to $X_a$ among the components of $G'-S_1$.
Such faces are said to \emph{belong to} $X_a$.
For a component (tree) $T$ of $H$, a face $F$ of $G'$ is said to \emph{belong to} $T$ if it is represented by some $b\in V(T)\cap B$ or if it belongs to the component $X_a$ for some $a\in V(T)\cap A$.
Observe that the faces of $G'[S]$ are exactly the unions of the faces that belong to the individual components of $H$, and the weights of the faces of $G'[S]$ correspond to the weights of the components of $H$.

\begin{claim}
In $\Oh(n)$ time, the separator $S$ can be turned into a vertex set $\bar S$ of
size $\Oh(\delta^3\log n)$ such that $G'[\bar S]$ has split balance at least
$\frac 13$ (given access to the data structure from~\cref{lem:compute_filling}).
\end{claim}

\begin{claimproof}
Consider a component $T$ of $H$.
We claim that the union of the faces belonging to $T$ is a connected region $R_T$ whose boundary vertices are in $S$.
Since $T$ is a component of $H$, no face of $G'$ can contain simultaneously a vertex in $X_a$ and a vertex in $X_{a'}$ for some $a\in V(T)\cap A$ and $a'\in V(H-T)\cap A$, because such a face would give rise to a vertex in $B$ connected to both $a$ and $a'$ in $H$.
Thus, if $v\in X_a$ for some $a\in V(T)\cap A$, then all faces incident to $v$ belong to $R_T$, so $v\notin\bd R_T$.
Similarly, if $v\in X_{a'}$ for some $a'\in V(H-T)\cap A$, then no face incident to $v$ belongs to $R_T$, so $v\notin\bd R_T$.
This shows that indeed $\bd R_T\subseteq S$.

We conclude that a face of $G'[S]$ is either a face of $G'$ and thus has weight at most $rn<\frac n2$, or it belongs to a component of $H$.
One can check that constructing $H$ and finding its maximum weight component can be implemented in $\Oh(n)$ time.
Consequently, when all components $T$ of $H$ have weight at most $\frac n2$, then $\bar S=S$ is a suitable separator.

Suppose that this is not the case, and we have a component $T$ of $H$ that has weight greater than $\frac n2$.
In the vertex-weighted tree $T$ (with non-negative weights) of total weight $w>\frac n2$, we can greedily find a vertex $v$ whose removal cuts $T$ into parts of weight at most $\frac w2\leq\frac n2$.
(This can be done by rooting $T$ somewhere and computing all subtree weights in $\Oh(n)$ time.)
We distinguish two cases based on the type of $v$.

\begin{description}
\item[Case 1.] $v\in B$.\\
Then there is a face $F_v$ corresponding to $v$.
By adding the vertices of the cycle of $F_v$ to $S$, we get a new separator $\bar S$ that is larger than $S$ by at most $\Oh(\delta)$, and its split balance is at least $\frac12$.

\item[Case 2.] $v\in A$.\\
Then the corresponding component $X_v$ of $G'-S$ has weight at most $\frac n2$, since it is a component of $G'-S$.
We claim that $v$ has degree $\Oh(|S|)$ in $H$.
Recall that each face that is incident to at least two components of
$G'-S$ must contain at least two vertices of $S$ that are separated
along the face cycle by two vertices from distinct components of $G'-S$,
i.e., there are two vertices of $S$ on the face that are not adjacent in
$G'$. We can draw a curve inside each such face that connects two vertices
of $S$ that are not adjacent in $G'$. Since faces are pairwise internally disjoint,
these curves form a planar graph on $S$. By a consequence of Euler's
formula, there are at most $3|S|-6$ edges in this graph. It follows that
there are at most $3|S|-6$ faces that contain simultaneously a vertex from $X_v$ and a vertex from
another component of $G'-S$.
Thus the degree of $v$ in $H$ is $\Oh(|S|)$.

We add the vertices on the faces corresponding to the neighbors of $v$ in $H$ to the separator $S$, resulting in a new separator $\bar S$.
Since $v$ has degree $\Oh(|S|)$, the set $\bar S$ is at most $\Oh(\delta)$ times larger than $S$.
Moreover, the split balance of $\bar S$ is at least $\frac12$.
\end{description}

In both cases we arrived at a separator $\bar S$ with size $\Oh(\delta^3\log n)$ and split balance at least $\frac12$, and all the manipulations could be executed in $\Oh(n)$ time.
\end{claimproof}

We can now continue the proof of \Cref{lem:out-1} with the separator $\bar S$ computed above.
Thus every face of $G'[\bar S]$ has weight at most $\frac n2$.
If there is a face of $G'[\bar S]$ of weight at least $\frac n3$, then let $U$ be that face.
Otherwise, start with any face of $G'[\bar S]$, and greedily add adjacent (i.e., edge-sharing) faces of $G'[\bar S]$ until a collection $\mathcal F$ of faces is found with sum of weights exceeding $\frac n3$.
The collection $\mathcal F$ can be built in $\Oh(n)$ time using a
breadth-first search on the dual graph of $G'[\bar S]$.
Since each new addition changes the total weight by less than $\frac n3$, the sum of the weights of the faces in $\Ff$ is between $\frac13 n$ and $\frac23 n$.
Let $U$ be the union of the faces in $\Ff$.
Thus the weight of $U$ is at most $\frac 23n +|\bar S| =\frac23 n + o(n)< \frac34 n$.
Note that $\bd U$ is covered by edges of $G'[\bar S]$ and therefore has at most $\Oh(\delta^3\log n)$ edges.

If $\bd U$ is a simple cycle, then it is a cycle separator of length at most $\Oh(\delta^3\log n)$ with split balance at least $\frac 14$.
We can apply \Cref{lem:cycle_to_geodesiccycle} to get the desired separator.

Suppose now that $U$ has `holes', i.e., its boundary is not a simple cycle. Let
$\gamma_1,\dots, \gamma_k$ be the boundaries of the connected components of $\sph^2
\setminus U$ (as simple cycles induced by $G'[\bar S]$), with the corresponding
region covering total weight $w_1,\dots,w_k$, respectively. Notice that
$\sum_{i=1}^k w_i \in [\frac14 n, \frac 34 n]$ and $\sum_{i=1}^k |\gamma_i| = \Oh(\delta^3\log n)$. Thus $\frac{\sum_i w_i}{\sum_i |\gamma_i|} =\Omega \left(\frac{n}{\delta^3\log n}\right)$.

Let $j\in \{1,\dots,k\}$ be the index that maximizes the fraction
$\frac{w_i}{|\gamma_i|}$. Then $w_j|\gamma_i|\geq w_i|\gamma_j|$ for all $i$. Summing over all
$i$, we get $w_j\sum_i |\gamma_i|\geq |\gamma_j|\sum_i w_i$, and thus
$\frac{w_j}{|\gamma_j|}\geq  \frac{\sum_i w_i}{\sum_i |\gamma_i|}
=\Omega\left(\frac{n}{\delta^3\log n}\right)$. Since the balance of $C_j$ is at
least $\min(\frac{w_i}{n}, \frac14)$, the previous inequality gives that its
balance is $\Omega\left(\frac{|\gamma_j|}{\delta^3\log n}\right)=\Omega(r |\gamma_j|)$.
We now apply Lemma~\ref{lem:cycle_to_geodesiccycle2} on the cycle $\gamma_j$ in the
original graph $G$. This yields a geodesic cycle separator $\gamma$ that splits with
balance $r/2^{\Oh(\delta)}=\frac{1}{2^{\Oh(\delta)}\log n}$, as required.
\end{proof}

\subsection{\outcome{3}}\label{sec:outcome-3}

In this section we prove the following lemma.

\begin{lemma}[\outcome{3}]\label{lem:out-3}
    Let $G'$ be the face-weighted graph and $H$ be the greedy filling returned
    by the last iteration of the iterated greedy filling procedure. Let $C$ be
    the face-cycle of $G$ and $G'$ based on which $H$ was created, where each
    filling face of $H$ has at most $rn$ vertices of $G$ in its interior, and
    the boundary of each filling face has length at most $k=\Oh(\delta)$. 

    Then in the original graph $G$, there exists a shortest path of length
    $\Oh(\delta^2 \log n)$ that is a $1/7$-balanced separator. Moreover, this path can
    be found in $\Oh(\delta^2 n \log n)$ time
    (given access to the data structure from~\cref{lem:compute_filling}).
\end{lemma}

\begin{proof}
    Through the proof we assume the black-box access to the greedy filling data
structure from~\cref{lem:compute_filling}).
Consider now the plane embedding of $H$ where $C$ is the boundary of the unbounded face.
In the first step, we add a vertex to the interior of each face, and connect it to all vertices of the face; let $\root$ denote the vertex added to the unbounded face. Note that this operation preserves planarity; let $H_\Delta$ denote the resulting graph. Let $\Tt$ be the breadth-first search tree of $H_\Delta$
starting from $\root$.

\begin{claim}\label{claim:depth-bfs}
    The depth of $\Tt$ is $\Oh(\delta \log n)$.
\end{claim}
\begin{claimproof}
    We will iteratively construct \emph{layers} in $H$, see \cref{fig:path-sep} (i).
    Initially, set $C_1 = C$ and $H'_1=H$, and let $L_1$ be the set of faces of $H'_1$ that share an edge with $C_1$.
    Next, we iterate over index $i=2,3\dots$ with the following procedure. Let $H'_i$ be
    the graph obtained by removing the edges that are shared only by the faces
    in $L_{i-1}$ from $H'_{i-1}$. Let $C_i$ be the set of edges on the unbounded face of $H'_i$. We
    continue this procedure until the graph $H'_i$ is empty. Let $\ell$ be the
    largest layer index in this procedure, and let $L_{\geq i}=L_i\cup L_{i+1}\cup\dots\cup L_\ell$ denote the set of faces of $H'_i$.
    Since each face of $L_{i+1}$ shares a vertex with some face of $L_i$, and we can get from any vertex of a face to any other vertex of the face in two steps inside $H_\Delta$, we have that any vertex of $L_i$ is within distance $2i+1$ to root.

    Observe that it remains to prove that $\ell = \Oh(\delta \log n)$.
    Since each filling face of $H$ has length at most $k$ and each edge of $C_i$ is shared with some face of $L_i$,  and face of $L_i$ are no longer than $k$, it means that $|C_i|\leq k \cdot |L_i|$.
    On the other hand, we can apply \cref{lem:isoperimetry} on each component cycle of $C_i$ and get $|L_{\geq i}| \leq |C_i|$. Thus $|L_{\geq i}|\leq k |L_i|$, which gives $|L_i|\geq\frac{1}{k} |L_{\geq i}|$ and $|L_{\geq i+1}|\leq \left(1-\frac{1}{k}\right) |L_{\geq i}|$. Consequently, after $\ell = \log_{\frac{1}{1-1/k}} n = \Oh(k \log n)=\Oh(\delta\log n)$ iterations, the graph $H_\ell$ is empty, and the procedure terminates.
\end{claimproof}

\begin{figure}[t]
    \centering
    \includegraphics[width=\textwidth]{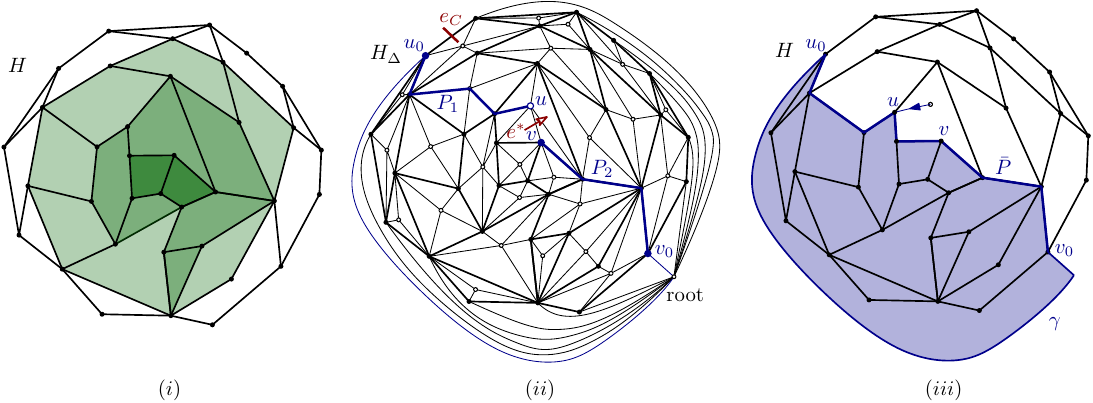}
    \caption{(i) The layers of $H$. (ii) The graph $H_\Delta$ and a balanced separator edge $e^*$ of the dual tree $\Tt^*$. The paths $P_1$ and $P_2$ are connecting $u$ and $v$ inside the BFS tree $\Tt$ to $\root$. (iii) Since $u\not \in V(H)$, it had to be moved to the next vertex on $P_1$. The bounded region of $P^\circ =\bar P\cup \gamma$ (shaded blue) has those faces of $H_\Delta$ whose dual vertices are behind $e^*$ in $\Tt^*$, except some of those where the containing $H$-face shares an edge with $\bar P_0$.}
    \label{fig:path-sep}
\end{figure}

Next, we continue our construction. Let $H_\Delta^\ast$ be the dual graph of $H_\Delta$. For each filling face $F$ of $H$, let $F^\ast$ be some fixed vertex in $H_\Delta^\ast$ whose dual triangle is inside $F$. Similarly, for each vertex $v\in V(H)$ let $v^\ast$ be the vertex of $H_\Delta^\ast$ corresponding to some triangular face that is incident to $v$ in $H_\Delta$. Initially we set the weight of all vertices in $H_\Delta^\ast$ to $0$, and then for each $F$ with weight $w(F)\leq rn$ (recall that the weight of $F$ is the number of $G$-vertices in its interior on $\sph^2$) we increase the weight of $F^\ast$ by $w(F)$. Similarly, for each $v\in V(H)$ we increase the weight of $v^\ast$ by one. Note that the vertices $F^\ast$ are pairwise distinct, but the vertices $v^\ast$ may coincide with each other and with the vertices $F^\ast$. Since each dual vertex can receive at most $3$ unit weights (from the vertices of the corresponding triangle) and at most one face weight of $H$, the maximum weight of a vertex is at most $rn+3$.

Recall that each edge $e$ in $H_\Delta$ corresponds to an edge $e^\ast$ in the
dual graph. The endpoints of $e^\ast$ are the dual vertices corresponding to the
faces on opposite sides of the edge $e$. Let $\Tt^\ast$ be the graph
spanned by the edges $e^\ast$ of the dual graph $H^\ast_\Delta$ whose corresponding
edges in the primal graph are not in $\Tt$. 

\begin{claim}\label{clm:cubic-tree}
    $\Tt^\ast$ is a cubic tree on the faces of $H_\Delta$ whose dual edges are not incident to $\root$.
\end{claim}
\begin{claimproof}
    The maximum degree of the graph $\Tt^\ast$ is three because it is a
    subgraph of the dual graph of a triangulated graph.
    Observe that $\Tt$ is a spanning tree of $H_\Delta$, therefore,
    by~\cite[Theorem XI.6]{tutte2001graph} $\Tt^\ast$ is also a spanning
    tree of $H_\Delta^\ast$. Finally, since $\Tt$ is a BFS tree from $\root$, all edges incident to $\root$ are included in it, thus $\Tt^\ast$ cannot include the dual of such an edge, which concludes the proof.
\end{claimproof}

\begin{claim}\label{claim:dual-tree-balanced}
    There exists an edge $e^\ast$ in $\Tt^\ast$ that is a $(1/4)$-balanced
    separator in $\Tt^\ast$.
\end{claim}
\begin{claimproof}
    According to Claim~\ref{clm:cubic-tree}, $\Tt^\ast$ is a cubic tree.
    Each vertex in this tree has weight at most $rn + 3<3+ n/\log n$ and the total weight
    is $n$.  Therefore, there exists a vertex $v$ such that each connected component of
    $\Tt^\ast-v$ has at most $2/3$ fraction of the total weight. (Such a vertex could be found for example by traversing $\Tt^\ast$ from some arbitrary leaf, moving towards the heavier subtree). Among the three edges incident to such a vertex select the one that points towards the heaviest component which therefore has weight at least $(n-3-n/\log n)/3>(n-3n/\log n)/3$, leaving weight at most $2n/3+3n/\log n<3n/4$ weight for the other component (as $n$ is large enough).
\end{claimproof}

Now, consider an edge $e^\ast$ that serves as a balanced separator in
$\Tt^\ast$, see \cref{fig:path-sep} (ii). Let $e$ be the edge corresponding to $e^\ast$ in the primal
graph $H_\Delta$, and let $u$ and $v$ be the endpoints of $e$. Let $u_0 \in V(C)$ be
the vertex on the outer face that lies on the path from the vertex
$\root$ to $u$ in the BFS tree $\Tt$. Note
that, due to $\Tt$ being a BFS tree, there exists exactly one such
vertex. Similarly, let $v_0 \in V(C)$ be the vertex on the outer face that lies
on the path from $u$ to $\root$ in $\Tt$.

\begin{claim}\label{cl:isoperimetry}
    Let $P$ be the shortest path between $a$ and $b$ in $G$, and let $\bar{P}$
    be an $\alpha$-balanced path separator of $G$ with endpoints $a,b$. 
    If $\bar{P}$ has length $\Oh(\delta^2 \log n)$ then $P$ is an
    $(\alpha-1/20)$-balanced separator.
\end{claim}
\begin{claimproof}
    Consider the starting face of the greedy filling $H$ as the outer face.
    Define the closed walk $W = P \cup \bar{P}$ in $G$. Let $\gamma_1,\ldots,\gamma_k$ be
    the boundaries of faces of $W$, which are simple cycles. By assumption we
    have that $|P|\leq |\bar P|=\Oh(\delta^2\log n)$, thus $|W| = \Oh(\delta^2 \log n)$.  Based on
    Lemma~\ref{lem:isoperimetry}, the cycle $\gamma_i$ interacts
    with at most $\Oh(\delta |\gamma_i|)$ faces of $H$. Each face has a weight of at
    most $rn$. Hence, at most $w^* = rn\cdot \Oh\left(\sum_{i=1}^k \delta|\gamma_i|\right) \le \Oh(rn\delta |W|)$ weight is placed on different sides of the
    separators $P$ and $\bar{P}$. Since $rn\delta|W| < \Oh(n/c)$. By setting $c$ large enough, we have that $w^*=\Oh(n/c)<n/20$, and each
    connected component of $G - P$ has weight at most $(1-\alpha)n + w^*< (1-\alpha+1/20)n$.
\end{claimproof}

Consider a shortest path $P$ between $u_0$ and $v_0$. Next, we show that this
path is indeed a balanced separator.

\begin{claim}\label{claim:balance-in-g'}
    The shortest path from $u_0$ to $v_0$ in $G$ is a $(1/7)$-balanced separator
    of length $\Oh(\delta^2 \log n)$.
\end{claim}
\begin{claimproof}
    Let $P_1$ be the path from $u$ to $u_0$, and let $P_2$ be the path from $v$
    to $v_0$ (both in $H_\Delta$). If $u$ or $v$ is not a vertex of $H$ (i.e., it is the vertex added to the middle of the face in $H$), then exchange it for its neighbor along $P_1$ or $P_2$; as a result we may assume that $u,v$ are on the boundary cycle of the same filling face of $H$. Let $P_3$
    be the shortest path between $u$ and $v$ along this filling face.

    By \cref{claim:depth-bfs} the walk $\bar{P}_0 = P_1 \cup P_2 \cup P_3$ has length
    $\Oh(\delta \log n)$, 
    with endpoints $u_0$ and $v_0$ in $H$.
    We now lift $\bar{P}_0$ to $H$, i.e., a vertex $v$ of
    $\bar{P}_0$ is either already a vertex of $G$, or $v$ corresponds to a face
    $F$ in $H$, and then we connect the neighbors of $v$ (which are vertices of
    $G$) with a shortest path along the boundary of $F$. Since the neighbors of
    $v$ are on the same filling face, this operation increases the length by a
    factor of at most $k=\Oh(\delta)$, so $|\bar P|=\Oh(\delta^2\log n)$. We then proceed to remove any loops (i.e., intervals of the walk with the same start and end vertex) to get the final path $\bar{P}$.  See \cref{fig:path-sep} (iii) for an illustration. 

    Recall that all edges incident to $\root$ are in $\Tt$, and one of the incident faces to $\root$ is the outer face of $H_\Delta$. Let $e_C$ be the unique edge of $\Tt^*$ that ends in the dual of the outer face of $H_\Delta$.
    
    We add a curve $\gamma$ of the path $u_0\,\root\, v_0$ from $u_0$ to $v_0$ inside the face $C$ to create a closed curve $P^\circ=\bar P \cup \gamma$. We fix an orientation on $P^\circ$, and let us orient $e^\ast$ to point towards $e_C$ inside $\Tt^*$. Observe moreover that if a vertex $F^\ast$ of $H^\ast_\Delta$ is behind $e^\ast$ in $\Tt^\ast$, then the face $F$ in the primal graph is typically in the bounded face of $P^\circ$; similarly, if $F^\ast$ is after $e^\ast$ in $\Tt^\ast$, then $F$ is typically in the unbounded face of $P^\circ$. The only possible exception to this is when $F$ is inside a face of $H$ that is incident to $\bar{P}_0$, as the lifting step may place such faces on the other side of $P^\circ$. (Note that when $u_0=v_0$, then this curve has length $0$, and one of its sides is empty.) 
    
    We will now show that the path $\bar{P}$ serves as a $1/5$-balanced
    separator in $G$. Let $\Ff$ denote the set of filling faces in $H$ that
    share an edge with $\bar{P}_0$. \cref{claim:dual-tree-balanced} implies that
    each side of $e^\ast$ has weight at most $3n/4$. By the above
    observation, we know that face weights and vertex weights can differ between
    the $e^\ast$-separation and the $\bar P$-separation only when the face is in
    $\Ff$ or the vertex is incident to some face in $\Ff$. Notice that
    $|\Ff|\leq 2|\bar P_0| = \Oh(\delta^2\log n)$, thus the total face weight and
    vertex weight of $\Ff$ is at most $(rn+k)\cdot \Oh(\delta^2\log n)<\Oh(n/c)<n/10$ when $c$ is large enough, so each side of $\bar P$ has weight at most $3n/4+n/20$, i.e., $\bar P$ has balance at least $n/5$.
    
    Let $P$ be the shortest path between $u_0$ and $v_0$ in $G$.
    By~\cref{cl:isoperimetry} applied to $\bar P$ and $P$ we have that the shortest path $P$ is a separator of $G$ with balance $n/5-n/20>n/7$.
\end{claimproof}

It is important to observe that the construction of the BFS tree $\Tt$ can be
done in linear time. Similarly, triangulating a planar graph, finding the dual
graph, and obtaining the tree $\Tt^\ast$ can also be achieved in linear time. To
find an edge $e^\ast$ that serves as a balanced separator in $\Tt^\ast$, we can
precompute the subtree weights of $\Tt^\ast$ and greedily walk to a balanced
edge, doing $\Oh(1)$ weight checks per traversed vertex, requiring a total time
complexity of $\Oh(n)$. The shortest path between $u_0$ and $v_0$ can be found
in $\Oh(n \log(n))$ time. This concludes the proof of~\cref{lem:out-3}.    
\end{proof}

\begin{remark}
    If we choose $c$ large enough, then we can achieve a balance arbitrarily close to $1/3$ for the path separator case of \cref{lem:separator-2-connected}.
\end{remark}

\subsection{Proof of~\cref{thm:separator}}

 Recall that \cref{lem:out-1} and~\cref{lem:out-3} together conclude the proof
of~\cref{lem:separator-2-connected}.
This gives us a construction of a balanced
separator when the input graph is $2$-connected. Now, we conclude the proof
of~\cref{thm:separator} by reducing the case of connected graphs to the case of $2$-connected graphs.

Let $\mathbb{T}$ be the block-cut tree of $G$. The vertices in this tree
correspond to maximal 2-connected components of $G$ or cut-vertices of $G$. If
$x \in V(\mathbb{T})$ corresponds to a 2-connected component of $G$, then it
identifies a set $B_x \subseteq V(G)$ that induces a maximal 2-connected
component of $G$. Otherwise, $x$ simply corresponds to a cut-vertex of $G$. For
an edge $uv \in E(\mathbb{T})$, it corresponds to a bridge of $G$ if in the
original graph $G$, there was a bridge connecting $G[B_u]$ to $G[B_v]$.
Otherwise, an edge $uv \in E(\mathbb{T})$ corresponds to the 2-connected
component $G[B_v]$ and a cut-vertex $u$ or component $G[B_u]$ and a cut-vertex
$v$, see \cref{fig:2conn_reduction}. Note that a cut-vertex may be inside several 2-connected components at once.

Next, we direct the edges of $\mathbb{T}$ as follows. For an edge $uv \in
E(\mathbb{T})$, let $\mathbb{T}_{uv}^L$ and $\mathbb{T}_{uv}^R$ be two connected
components obtained by deleting $uv$ from $\mathbb{T}$, such that $u \in
\mathbb{T}_{uv}^L$ and $v \in \mathbb{T}^R_{uv}$. For $i \in \{L,R\}$, let
$\omega(\mathbb{T}_{uv}^i)$ be defined as $|\bigcup_{x \in \mathbb{T}_{uv}^i}
B_x|$. The edge $uv$ is directed from $u$ to $v$ if $\omega(\mathbb{T}_{uv}^L) >
\omega(\mathbb{T}_{uv}^R)$. Otherwise, we direct the edge $uv$ from $v$ to $u$. See \cref{fig:2conn_reduction}(ii) for an illustration.
It is important to observe that each edge of $\mathbb{T}$ is now directed, and
since $\mathbb{T}$ is a tree, it has exactly $|V(\mathbb{T})| - 1$ edges. The
sum of the outdegrees is also $|V(\mathbb{T})| - 1$. Consequently, there exists
a vertex $\alpha \in V(\mathbb{T})$ that has no outgoing directed edge.
If some $2$-connected component $B_v$ has more than $n/20$ vertices, then we can use~\cref{lem:separator-2-connected} directly on $G[B_v]$. The resulting path or cycle has balance at most $20$ times worse on $G$, i.e., we either have a geodesic cycle with balance $\frac{1}{2^{\Oh(\delta)}\log n}$ or a shortest path with balance $n/140$. For the rest of the proof, suppose that each component $B_v$ has size at most $n/20$.

\begin{figure}[t]
    \centering
    \includegraphics[width=\textwidth]{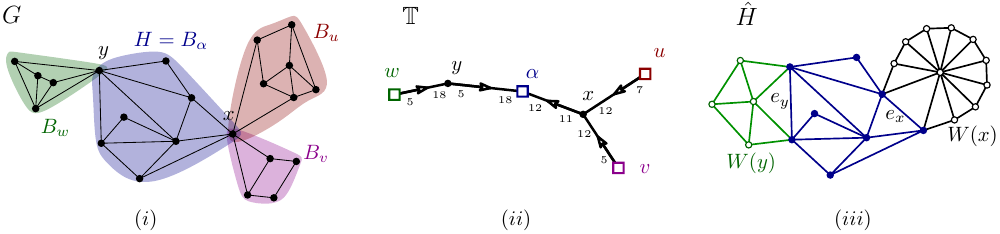}
    \caption{(i) A graph $G$ and its $2$-connected blocks with cut vertices $x,y$. (ii) The block-cut tree $\mathbb{T}$ of $G$ and the weights $\omega(\mathbb{T}^{L,R}_{uv})$ on each edge $uv$. (iii) The graph $\hat H$ that we get by attaching wheels. A vertex that cuts off $k$ vertices is associated with a wheel $W_{k+1}$.}
    \label{fig:2conn_reduction}
\end{figure}

Observe that when $\alpha$ corresponds to a cut-vertex in $G$, we are done
because the length-0 shortest path from $\alpha$ to $\alpha$ serves as a $1/3$-balanced
separator. In the other case, $\alpha$ corresponds to a $2$-connected component
of $G$. It is important to note that we can assume $|B_\alpha| > 1$, as if
$|B_\alpha| = 1$, a single vertex would act as a $1/3$-balanced separator. We
now need to analyze the case when $|B_\alpha| > 1$.

Let $H = G[B_\alpha]$. Observe that the graph $H$ is 2-connected;
however, we cannot directly apply~\cref{lem:separator-2-connected} as $H$ itself
may be small. The idea is to modify $H$ in such a way that the separator, after
applying~\cref{lem:separator-2-connected}, corresponds to the balanced separator
in $G$. More concretely, for any integer $k \ge 1$, let $W_k$ denote the
\emph{wheel graph} with $k$ vertices. In other words, $W_1$ consists of a single
vertex; $W_2$ is an edge; $W_3$ is a triangle; and for every $k \ge 4$, let
$W_k$ be a cycle with $k-1$ vertices and an additional universal vertex connected to all $k-1$ vertices of the cycle. It should be noted that $W_k$ is a planar graph for all $k \in \nat$, and we will consider a plane embedding where the universal vertex is drawn inside the cycle.

Let $N(\alpha) \subseteq V(H)$ be the set of cut vertices of $G$ that are in $V(H)$. For each vertex $x \in N(\alpha)$, we arbitrarily
select an edge $e_x \in E(H)$ that is adjacent to $x$. Next, for every $x \in
N(\alpha)$, we create a copy of $W_{k+1}$, where $k = \omega(\mathbb{T}_{x\alpha}^L)$, and
name it $W(x)$. We connect the aforementioned copy to $H$ by identifying two adjacent vertices of the $k+1$-length cycle of $W(x)$ with $e_x$, i.e., the resulting graph has $k-1$ new vertices (recall that the cut vertex has been shared) and $|E(W_{k+1})-1|$ new edges. This is illustrated in \cref{fig:2conn_reduction}(iii).

Observe that we can naturally define the bijection $\eta : V(G) \to V(\hat{H})$.
If $x \in V(H)$, then $\eta(x) = x$. Next, for every $x \in N(\alpha)$, we map
all vertices from $\bigcup_{x \in \mathbb{T}_{x\alpha}^L} B_x$ to the vertices of
$W(x)-e_x$ in the order of their labels. In particular, the number of vertices in
$\hat{H}$ is equal to the number of vertices in $G$. Now, we prove the
following properties of $\hat{H}$.

\begin{proposition}
    If $G$ is a connected planar $\delta$-slim graph, then  $\hat{H}$ is a 2-connected planar $(\delta+6)$-slim graph. 
\end{proposition}
\begin{proof}
    First, observe that $\hat{H}$ is planar, because we simply add a copy of the
    planar graph $W_k$ and connect its outer-face vertices to the endpoints of
    some edges in the planar graph $H$.
    Similarly, notice that $\hat{H}$ is 2-connected because the original graph
    $H$ is 2-connected, and we add a 2-connected graph $W_k$ to two different
    edges of $H$. One can verify that this construction also works for $k=1$
    and $k=2$.

    Therefore, it remains to show that $\hat{H}$ is
    $(\delta+6)$-slim. To do this, observe that $H$ is
    $\delta$-slim because it is a $2$-connected component of $G$, consequently, any shortest path in $H$
    is also a shortest path in $G$.

    Next, observe that each vertex of $W_k$ is within distance $3$ from the vertices of $H$ that $W_k$ was connected to. Therefore, any
    shortest path in $\hat{H}$ lies within the $3$-neighborhood of some shortest
    path in $H$, and by the same argument, each side of a triangle of $\hat H$ is in the $3$-neighborhood of the corresponding side of a triangle in $H$. Hence, the graph $\hat{H}$ is
    $(\delta+6)$-slim.
\end{proof}

We observe that a separator of $\hat H$ that happens to fall into $H$ is a good separator for $G$.

\begin{observation}\label{obs:2-connectivity-case1}
    Let $\beta \in (0,1/3)$ be a fixed constant. Let $\hat{P}$ be a $\beta$-balanced
    separator in $\hat{H}$ such that $V(\hat{P}) \subseteq V(H)$. Then $\hat{P}$
    is also a $\beta$-balanced separator in $G$.
\end{observation}

\begin{claim}\label{claim:sp-hbar}
    Let $\hat{P}$ be a $1/7$-balanced shortest path separator in $\hat{H}$ guaranteed by \cref{lem:separator-2-connected}. Then,
    there exist a $1/140$-balanced shortest path separator in $G$, and based on $\hat P$ it can be computed in $\Oh(n)$ time.
\end{claim}
\begin{claimproof}
    If $V(\hat{P}) \subseteq V(H)$, then according to
    \cref{obs:2-connectivity-case1}, $\hat{P}$ is also a separator in
    $G$ with the same balance. Thus, there must be at least one vertex in $\hat{P}$ that belongs to a
    wheel graph $W$. Let $e \in H$ be the edge connecting the connection points of
    this $W$. We modify $\hat{P}$ by removing the vertices of $W$. Since a shortest path cannot enter and then exit $W$ (as that would be a detour), we are simply removing some starting or ending segment from the path, thus even after this modification, the path
    remains a shortest path. We do the same modification on the other end of $\hat P$ if necessary. Let $P_G$ be the resulting path, and let $a,a'\in V(H)$ be the points where $W$ is attached to $H$. We claim that $P_G$ is a $1/140$-balanced separator in $G$. Observe that the balance difference of $\hat P$ in $\hat H$ and of $P_G$ in $G$ is entirely due to the vertices in the (at most two) wheels at the ends of $\hat P$ that fall into $V(\hat H)\setminus V(H)$. For a wheel $W$ recall that its number of vertices outside $H$ is the number of vertices in some $2$-connected component $B_v$, which in turn has size at most $n/20$. Thus any component of $G-P_G$ has size at most $6n/7+2\cdot n/20<n-n/140$.
\end{claimproof}

\begin{claim}\label{claim:gc-hbar}
    Let $\hat{C}$ be a $\beta$-balanced geodesic cycle separator in $\hat{H}$. Then there exists a geodesic cycle separator $C$ in $G$ of balance at least $\beta$, or a single vertex (shortest path of length $0$) with balance $1/2$, and the separator can be computed in $\Oh(\delta n\log\log n)$ time.
\end{claim}
\begin{claimproof}
    Let $\hat{C}$ be the geodesic cycle separator in $\hat{H}$. First consider the case when the cycle contains some vertex of $V(\hat H)\setminus V(H)$, i.e., some newly added vertex of a wheel $W(x)$.
    If $\hat{C}$ were a cycle completely contained within $W(x)$,
    then it would have to be a triangle, so any geodesic cycle of $G$ is a suitable separator (as $n$ must be small). We can find a geodesic cycle by finding any cycle in linear time, and then applying \cref{lem:cycle_to_geodesiccycle2}.
    
    Suppose now that $C$ has some vertex $v\in W(x)\cap V(\hat H)\setminus V(H)$ but it is not entirely contained in $W(x)$. Since $C$ is a cycle, it must enter and exit $W(x)$ at the endpoints of the edge $e_x$ to which $W(x)$ was glued. Since $C$ is geodesic, it must contain the edge $e_x$. This is a contradiction to $C$ not being contained in $W_x$.

    Therefore, we can assume that $V(\hat{C}) \subseteq V(H)$. According to
    \cref{obs:2-connectivity-case1}, the balance of $\hat{C}$ in $\hat{H}$ is
    the same as the balance of $\hat{C}$ in $G$, which concludes the proof.
\end{claimproof}

We utilize \cref{lem:separator-2-connected} on the graph $\hat{H}$, which is
a $2$-connected $(\delta+6)$-slim (and thus planar $\Oh(\delta)$-hyperbolic) graph.
\cref{lem:separator-2-connected} provides two possible cases: (i) a shortest
path separator, or (ii) a geodesic cycle separator. In the first case, we
employ \cref{claim:sp-hbar} to establish that the original graph also
possesses a shortest path separator with balance $1/140$. In the latter case, we similarly use
\cref{claim:gc-hbar} to obtain a geodesic cycle separator with suitable
balance. It is worth noting that the constructions outlined in
\cref{claim:gc-hbar} can be implemented in linear time by scanning through
the edges of $G$. This concludes the proof of \cref{thm:separator}.

\begin{remark}
    By a more careful analysis of the constants, the above arguments yield a balance of $1/21-\eps$ for the path separation case.
\end{remark}

\section{Applications of the Separator Theorem} \label{sec:schemes}
In this section, we present two applications of our separator theorem. First, we show a near-linear time FPTAS for {\sc Maximum Independent Set}. Second, we prove a near-linear time FPTAS for the {\sc Traveling Salesperson} problem. To prove both of these results, however, we need to obtain an $r$-division, which we can obtain by strengthening the balance of our separator.

\subsection{Division by a Stronger Balance}

Our goal is to obtain a (weak) $r$-division, following the work of Frederickson~\cite{Frederickson87}, using the separator of \cref{thm:separator}. To this end, we first strengthen the balance of our separator theorem. 
Recall that \cref{thm:separator} yields an in-class separator by \cref{lem:geod-sep}. 
(For the {\sc Traveling Salesperson} problem, it is additionally important that the separator of \cref{thm:separator} yields connected separators, but we defer this discussion until later.)

We prove that most balanced separators can be `pumped' to be $1/2$-balanced,
while still being in-class. We call a function $g : \mathbb{N} \rightarrow
\mathbb{R}$ is \emph{sublinear} if $g(n) = \Oh(n^{1-\eps})$ for some fixed
$\eps > 0$.

\begin{lemma} \label{lem:sep-pump}
    Let ${\cal G}$ be a class of connected graphs and let $\mathcal{T},
    \mathsf{sep}: \mathbb{N} \rightarrow
    \mathbb{N}$ and $\alpha : \mathbb{N} \rightarrow \mathbb{R}^{+}$ be functions
    (possibly depending on ${\cal G}$) such that $\alpha(n) < 1$ for any $n \in
    \mathbb{N}$, $\mathsf{sep}(\cdot)$ is sublinear, $\alpha(\cdot)$ is monotone non-increasing,
    and $\mathcal{T}(\cdot),\mathsf{sep}(\cdot)$ are monotone non-decreasing. Suppose that for any $n$-vertex graph $G \in {\cal
    G}$ one can compute in $\mathcal{T}(n)$ time an $\alpha(n)$-balanced in-class separator of
    size $\mathsf{sep}(n)$. Then one can compute in $\Oh(\mathcal{T}(n)/\alpha(n))$ time
    a $1/2$-balanced in-class separator of size $\Oh(\mathsf{sep}(n)/\alpha(n))$.
\end{lemma}
\begin{proof}
    Let $G \in {\cal G}$ with $n$ vertices. We apply the following recursive
    procedure, starting with $G_0 = G$. Let $G_i \in {\cal G}$ and let $n_i$ be the
    number of vertices of $G_i$. Use the assumed algorithm to compute an
    $\alpha(n_i)$-balanced in-class separator $Z_i$ of $G_i$ of size
    $\mathsf{sep}(n_i)$. Let $C$ be the largest connected component of $G_i
    - Z_i$. If $|V(C)| \geq \frac{1}{2} n$, then let $G_{i+1} = G_i[V(C)
    \cup Z']$ and recurse. Otherwise, the procedure terminates.

    Let $k$ be the number of recursive calls. Note that $n_{i+1} \leq
    (1-\alpha(n_{i})) \cdot n_{i} + \mathsf{sep}(n_i) \leq (1-\alpha(n)) \cdot n_i
    + \mathsf{sep}(n)$ by the monotonicity of $\alpha$ and $\mathsf{sep}$. Hence, $n_k \leq
    (1-\alpha(n))^{k} \cdot n + \sum_{i=0}^{k-1} \left( (1-\alpha(n))^{i} \cdot
    \mathsf{sep}(n) \right) \leq (1-\alpha(n))^{k} \cdot n + \frac{1}{\alpha(n)}
    \cdot \mathsf{sep}(n)$. By the description of the procedure, $n_k \leq
    \frac{1}{2} n$. To obtain an upper bound of $k$, we instead seek the value of $k$ when $(1-\alpha(n))^{k} \cdot n + \frac{1}{\alpha(n)}
    \cdot \mathsf{sep}(n) \leq \frac{1}{2} n$. Then $k = \Oh(1/\ln (1-\alpha(n))) = \Oh(\frac{1}{\alpha(n)})$,
    where the hidden constant depends on (the sublinearity of)
    $\mathsf{sep}(\cdot)$.

    Let $Z = \bigcup_{i=0}^{k} Z_i$. We claim that $Z$ is the $1/2$-balanced
    in-class separator of size $\Oh(\frac{1}{\alpha(n)} \cdot \mathsf{sep}(n))$
    we are looking for. By the monotonicity of $\mathcal{T}$ and $\mathsf{sep}$, the bounds on the
    running time and $|Z|$ follow. By construction, no connected component of $G
    - Z$ has a size larger than $\frac{1}{2} n$. By induction on the
    recursive procedure and the fact that the separator is in-class, we can
    observe that $G[Z \cup \bigcup_{C \in {\cal C}} V(C)] \in {\cal G}$ for
    every subset ${\cal C}$ of the set of connected components of $G -
    Z$. Hence, the computed separator is in-class.
\end{proof}

In the following, we implicitly use that any $\delta$-hyperbolic graph is $\Oh(\delta)$-slim. Using the above lemma, we can immediately prove \cref{cor:oursep-pump}.

\corollarypump*
\begin{proof}
Observe that \cref{thm:separator} yields an in-class separator by
\cref{lem:geod-sep}. We note that \cref{thm:separator} returns either a
separator $X$ of size $\Oh(\delta^2 \log n)$ and a constant balance or a
separator $Y$ of size $\Oh(\delta)$ and balance $2^{-\Oh(\delta)}/\log n$. We
follow the proof of \cref{lem:sep-pump} using this separator. Indeed, the
running time, size bound, and balance factor satisfy the conditions of the
lemma. In the procedure of \cref{lem:sep-pump}, we note that the separator of
the first type is only used a constant number of times before the connected
components have a size at most $\frac{1}{2} n$. Hence, following the analysis of
\cref{lem:sep-pump}, the total size of the separator is $2^{\Oh(\delta)} \log n$. The running time is immediate from \cref{thm:separator} and \cref{lem:sep-pump}.
\end{proof}

We now describe how to obtain the (weak) $r$-division. The following definition and algorithm originates in the work of Frederickson~\cite{Frederickson87} (see also Lipton and Tarjan~\cite{LiptonT80}). 

\begin{definition}[Weak $r$-divison]
    For a set ${\cal P}$ of subsets of $V(G)$, we define its boundary to be the set of vertices that occur in more than one subset:
    $$\bnd(\mathcal{P}) \coloneqq \{ v \in P_1 \cap P_2 \mid P_1, P_2 \in \mathcal{P} \text{ with } P_1 \neq P_2 \}.$$
        For any integer $r\in \nat$, 
    a \emph{weak $r$-division} of a graph $G$ is a set ${\cal P}$
    of subsets of $V(G)$ such that (i) $\bigcup_{P \in {\cal P}} P = V(G)$, (ii)
    $|P| \leq r$ for each $P \in {\cal P}$, (iii) $|{\cal P}| = \Theta(n/r)$, and (iv)
        $|\bnd(\mathcal{P})| = \Oh(n/\sqrt{r})$.
    The stronger notion of an \emph{$r$-division} replaces the latter condition by the demand that each set in ${\cal P}$ shares at most $\sqrt{r}$ vertices in total with the other sets in ${\cal P}$.
\end{definition}

We call each set in ${\cal P}$ a \emph{group} of the (weak) $r$-division. The
vertices of a group that are unique to it are called \emph{interior} vertices,
while the vertices in $\bnd(\mathcal{P})$ are called \emph{boundary vertices}. 

\begin{lemma} \label{lem:division}
Let ${\cal G}$ be a class of connected graphs and let ${\cal T}, {\sf sep}: \mathbb{N} \rightarrow \mathbb{N}$ be functions (possibly depending on ${\cal G}$) such that ${\sf sep}(n) = \Oh(\sqrt{n})$ and ${\sf sep}(\cdot),{\cal T}(\cdot)$ are monotone non-decreasing. Suppose that for any $n$-vertex graph $G \in {\cal G}$ one can compute in ${\cal T}(n)$ time a $1/2$-balanced in-class separator of size ${\sf sep}(n)$. Then there exists an integer $s$ such that for any $n$-vertex graph $G \in {\cal G}$ and integer $r$ with $s < r < n$, one can compute in $\Oh({\cal T}(n) \log n)$ time a weak $r$-division ${\cal P}$ such that $G[P] \in {\cal G}$ for any $P \in {\cal P}$.
\end{lemma}
This procedure is the same as the one in Frederickson~\cite[Lemma~1]{Frederickson87}. We describe it explicitly here, because we also need to show that each group is in ${\cal G}$.

\begin{proof}
Let $G \in {\cal G}$ with $n$ vertices. We apply the following recursive
procedure, starting with $G' = G$. Let $G' \in {\cal G}$ and let $n'$ be the
number of vertices of $G'$. If $n' \leq r$, then we return $V(G')$ as a group. Otherwise, use the assumed algorithm to compute a $1/2$-balanced in-class separator $Z'$ of $G'$ of size $g(n')$. Split the vertices of the connected components of $G' - Z'$ into two parts $X$ and $Y$ such that $|X|,|Y| \leq \frac{2}{3} n'$ and any two vertices of a connected component of $G' - Z'$ are either both in $X$ or both in $Y$. Recurse on $G'[X \cup Z']$ and on $G'[Y \cup Z']$. Note that both induced subgraphs are in ${\cal G}$ by the assumption that the separator is in-class. Return the union of the sets of groups found. 

Let ${\cal P}$ denote the set of the groups computed by this procedure. We observe that every group of ${\cal P}$ is in ${\cal G}$. Furthermore, every group has at most $r$ vertices, and there are $\Theta(n/r)$ groups. Moreover, at most $\Oh(\log(n/r))$ levels of recursion are needed, and thus the running time is $\Oh({\cal T}(n) \log n)$.

It remains to analyze $|\bnd(\mathcal{P})|$. We follow the analysis of Klein et al.~\cite{KleinMS2012}. Following the recursive procedure and the assumption that ${\sf sep}(n) = \Oh(\sqrt{n})$, observe that $|\bnd(\mathcal{P})|$ is bounded by
$$
B_r(n) \leq \begin{cases}
  \rho\sqrt{n} + \max_{\{\alpha_i\}_{i=1}^2} \sum_{i=1}^2 B_r(\alpha_i n) & \mbox{if\ } n > r\\
 0 & \mbox{otherwise},
\end{cases}
$$
where $\alpha_i \leq 2/3 + c/\sqrt{n}$ for $i=1,2$ and $1 \leq \sum_{i=1}^2 \alpha_i \leq 1 + c/\sqrt{n}$ for some constants $c,\rho$. We can then apply Lemma~2 of~\cite{KleinMS2012} (using $\beta=1/2$ and with minor modification to account for two $\alpha$ values instead of eight) to show that $|\bnd(\mathcal{P})| \leq B_r(n) = \Oh(\rho n/\sqrt{r})$ when $r$ is larger than some (constant) integer $s$.
\end{proof}

\begin{corollary} \label{cor:ourdivision}
For some $\delta$, let $G$ be a $n$-vertex connected planar $\delta$-hyperbolic graph. Then there exists an integer $s$ such that for any integer $r$ with $s < r < n$, one can compute in $2^{\Oh(\delta)} \cdot n\log^{\separatorlogexponentplus{2}} n$ time a weak $r$-division ${\cal P}$ of $G$ such that $G[P]$ is a planar $\delta$-hyperbolic graph for each $P \in {\cal P}$.
\end{corollary}
\begin{proof}
    We apply \cref{cor:oursep-pump} in combination with \cref{lem:division}. Note
    that ${\sf sep}(n) = 2^{\Oh(\delta)} \log n \leq 2^{\Oh(\delta)} \sqrt{n}$ 
    and we can apply the analysis of \cref{lem:division} with $2^{\Oh(\delta)}$ as $\rho$.
\end{proof}

\begin{remark} \label{rem:ourdivision}
For the weak $r$-division ${\cal P}$ produced by \cref{cor:ourdivision}, it holds that $|\bnd({\cal P})| = \Oh(2^{\Oh(\delta)} n/\sqrt{r})$ by the analysis of \cref{lem:division}.
\end{remark}

\subsection{Near-linear Time FPTAS for {\sc Maximum Independent Set}}
We now present the near-linear time FPTAS for the {\sc Maximum Independent Set}
problem. In fact, we obtain an even more general result.

Following Chiba et al.~\cite{ChibaNS81}, we define a general class of problems.
For any graph property $\Pi$, an induced subgraph with property $\Pi$ of a graph
$G$ is a set $S \subseteq V(G)$ such that $G[S]$ has property $\Pi$. The {\sc
Maximum Induced Subgraph with Property $\Pi$} problem asks to find a largest
induced subgraph of a graph that has property $\Pi$. We refer to a graph
property $\Pi$ as \emph{determined by the components} if the property $\Pi$
holds for the disjoint union of any set of graphs that satisfy $\Pi$. We call
$\Pi$ \emph{hereditary} if any induced subgraph of a graph that satisfies $\Pi$
also satisfies $\Pi$. It is worth noting that {\sc Maximum Independent Set} is a
prime example of a {\sc Maximum Induced Subgraph with Property $\Pi$} problem:
the corresponding property $\Pi$ is that the graph is a set of pairwise
non-adjacent vertices, which is both hereditary and determined by the
components.

We refer to a class of graphs as \emph{$c$-colorable} if any graph in the class
can be properly colored using at most $c$ colors. We assume that $c$ is
constant. For instance, planar graphs are famously $4$-colorable \cite{AppelH1989,RobertsonSST97}.

\begin{lemma} \label{lem:div-to-algo}
Let ${\cal G}$ be a $c$-colorable graph class for some constant $c \in \mathbb{N}$ and let $\mathcal{T}: \mathbb{N} \times \mathbb{N} \rightarrow \mathbb{N}, \tw: \mathbb{N} \rightarrow \mathbb{N}$ be functions (where $c$, ${\cal T}$, and $\tw$ possibly depend on ${\cal G}$). Suppose that any connected $n$-vertex graph $G \in {\cal G}$ has treewidth at most $\tw(n)$. Moreover, suppose that there exists a constant integer $s$ such that for any $s < r < n$, one can compute in $\mathsf{div}(n)$ time a weak $r$-division of $G$ such that $G[P] \in {\cal G}$ for each group $P \in {\cal P}$. Let $\Pi$ be a graph property that is hereditary and determined by the components.

If the {\sc Maximum Induced Subgraph with Property $\Pi$} problem can be solved
in $\mathcal{T}(k,n)$ time on $n$-vertex graphs of treewidth $k$, then for any
$\eps > 0$, for ${\cal G}$ the {\sc Maximum Induced Subgraph with Property
$\Pi$} problem has a $(1-\eps)$-approximation algorithm running in $\mathsf{div}(n) +
\mathcal{T}(\tw(\Oh((c/\eps)^2)), \Oh((c/\eps)^2)) \cdot \Oh(n/(c\cdot\eps)^2)$ time.
\end{lemma}
\begin{proof}
The proof is adapted from Chiba et al.~\cite{ChibaNS81} and Lipton and Tarjan~\cite{LiptonT80}, strengthened to include the treewidth property.
Consider any $n$-vertex graph $G \in {\cal G}$. We may assume $G$ is connected. Otherwise, we run the algorithm on each connected component separately and
return the union of solutions. Since $\Pi$ is hereditary and determined by the components, this yields a valid solution. Moreover, the approximation ratio obtained for the instance is at most the maximum of the ratios obtained for each of the components. Finally, without loss of generality, $s < 1/\eps^2$ by slightly decreasing $\eps$ and $1/\eps^2 < n$ or we can just solve the entire instance directly using the assumed algorithm for graphs of bounded treewidth.

Compute a weak $(1/\eps^2)$-division ${\cal P}$ of $G$ in $\mathsf{div}(n)$ time
using the assumed algorithm. By definition, ${\cal P}$ has $\Oh(n/\eps^2)$
groups. Each group $P$ of ${\cal P}$ has $\Oh(1/\eps^2)$ vertices and $G[P] \in
{\cal G}$. By the assumption of the lemma, $G[P]$ has treewidth
$\tw(\Oh(1/\eps^2))$. Hence, the interior of $P$ has treewidth
$\tw(\Oh(1/\eps^2))$ as well. For each $P \in {\cal P}$, use the assumed
algorithm to solve the {\sc Maximum Induced Subgraph with Property $\Pi$}
problem on the subgraph induced by the interior of $P$. Finally, output the union of these solutions.

Since $\Pi$ is determined by the components, the output is an induced subgraph with property $\Pi$. Clearly, the running time of the above algorithm is as claimed. It remains to prove the approximation ratio.

Let $S^*$ be a solution to the {\sc Maximum Induced Subgraph with Property
$\Pi$} problem. Since $\Pi$ is hereditary and determined by the components, any independent set satisfies $\Pi$. Since $G$ is $c$-colorable, this implies that $|S^*| \geq n/c$. Since ${\cal P}$ is a weak $(1/\eps^2)$-division, $|\bnd {\cal P}| = \Oh(\eps n)$. Then $|S^*
\setminus \bnd {\cal P}| \geq (1-\Oh(c\cdot\eps)) \cdot |S^*|$. Since $\Pi$ is hereditary, $S^*
\setminus \bnd {\cal P}$ is also an induced subgraph with property $\Pi$, which only
consists of interior vertices. Since the output solution is at least as large as
$S^* \setminus \bnd {\cal P}$, it follows that the output is a $(1-\Oh(c\cdot\eps))$-approximation. The lemma follows by adjusting $\eps$ before the start of the algorithm.
\end{proof}

We now apply this lemma to planar $\delta$ hyperbolic graphs.

\begin{theorem} \label{thm:misp}
Let $\Pi$ be a graph property that is hereditary and determined by the components. Suppose {\sc Maximum Induced Subgraph with Property $\Pi$} can be solved in $\mathcal{T}(k,n)$ time on $n$-vertex graphs of treewidth $k$. For any $\delta \geq 0$ and any $\eps > 0$, the class of planar $\delta$-hyperbolic graphs has a $(1-\eps)$-approximation algorithm for {\sc Maximum Induced Subgraph with Property $\Pi$} running in time $2^{\Oh(\delta)} \cdot n \log^{\separatorlogexponentplus{2}} n + \mathcal{T}(\Oh(\delta^2 + \delta\log (1/\eps)), 2^{\Oh(\delta)}/\eps^2) \cdot 2^{\Oh(\delta)} \cdot n/\eps^2$ on $n$-vertex graphs in the class.
\end{theorem}
\begin{proof}
Recall that by \cref{prp:tw-hyper}, the treewidth of any $n$-vertex planar $\delta$-hyperbolic graph is $\Oh(\delta \log n)$. Moreover, planar graphs are $4$-colorable by the Four Color Theorem~\cite{AppelH1989,RobertsonSST97}. Furthermore, a weak $r$-division ${\cal P}$ of which each group is planar $\delta$-hyperbolic can be computed in $2^{\Oh(\delta)} \cdot n\log^{\separatorlogexponentplus{2}} n$ time by \cref{cor:ourdivision}. Recall that according to~\cref{rem:ourdivision}, it holds that $|\bnd({\cal P})| = \Oh(2^{\Oh(\delta)} n/\sqrt{r})$. This affects the approximation factor in the analysis of \cref{lem:div-to-algo} and $\eps$ must be adjusted by a factor~$2^{\Oh(\delta)}c$. Adjusting its analysis accordingly, the theorem follows by~\cref{lem:div-to-algo}.
\end{proof}

This allows us to straightforwardly prove Theorem~\ref{thm:mis}.

\theoremmis*
\begin{proof}
Note that {\sc Maximum Independent Set} is a {\sc Maximum Induced Subgraph with Property $\Pi$} problem. The corresponding $\Pi$ is hereditary and determined by the components. It is well known that {\sc Maximum Independent Set} can be solved in $2^{\Oh(\tw)} n$ time on $n$-vertex graphs of treewidth $\tw$~\cite{ArnborgP89}. Additionally, a tree decomposition of width $\Oh(\tw)$ can be computed in $2^{\Oh(\tw)} n$ time~\cite{BodlaenderDDFLP16,Korhonen21} (or even faster on planar graphs~\cite{KammerT16,GuX19}). We then apply \cref{thm:misp} and notice that the running time is $2^{\Oh(\delta)} \cdot n \log^{\separatorlogexponentplus{2}} n + 2^{\Oh(\delta^2)} \cdot \frac{1}{\eps^{\Oh(\delta)}} \cdot 2^{\Oh(\delta)} \cdot \frac{1}{\eps^2} \cdot \Oh(n/\eps^2)$. The theorem follows.
\end{proof}

We note that \cref{thm:misp} has other applications. We refer to Chiba et al.~\cite{ChibaNS81} for further suggestions, but highlight one well-known problem here, namely {\sc Maximum Induced Forest}. This is the complementary problem to {\sc Minimum Feedback Vertex Set}.

\begin{theorem} \label{thm:mif}
For any $\delta \geq 0$ and any $\eps > 0$, the class of planar $\delta$-hyperbolic graphs has a $(1-\eps)$-approximation algorithm for {\sc Maximum Induced Forest} running in time $2^{\Oh(\delta)} \cdot n \log^{\separatorlogexponentplus{2}} n + 2^{\Oh(\delta^2)} n/\eps^{\Oh(\delta)}$.
\end{theorem}
\begin{proof}
We follow the same arguments as in the proof of \cref{thm:mis}. Note that the
{\sc Maximum Induced Forest} problem can be solved in $2^{\Oh(k)} n^3$ time on
$n$-vertex graphs of treewidth $k$ (by the straightforward reduction to the {\sc
Minimum Feedback Vertex Set} problem~\cite{BodlaenderCKN15}).
\end{proof}

\subsection{Near-linear Time FPTAS for the {\sc Traveling Salesperson} Problem}
Now, we present the near-linear time FPTAS for the {\sc Traveling Salesperson} problem.

Recall that a tour in a graph $G$ is a closed walk in $G$ that visits
every vertex of $G$ at least once. For a tour $R$, we use $|R|$ to denote the
length of the tour. By $\tau(G)$, we denote the length of a shortest tour of $G$.

We require the following notion in the analysis. Recall that, for any graph $G$
and a set $T \subseteq V(G)$ of even size, a \emph{$T$-join} is a multiset $F$
of edges in $G$ such that in the multigraph $(V,F)$, each vertex of $T$ has odd
degree, and all other vertices have even degree. As shown by
Edmonds~and~Johnson~\cite{EdmondsJ73}, any connected graph $G$ has a $T$-join,
and any minimum-size $T$-join is a set of edge-disjoint paths in $G$. We can
observe that if $G$ is a tree, then we can construct the paths as follows: root
the tree arbitrarily and iteratively add the path in the tree between the two
vertices of $T$ that have the lowest-common-ancestor and have not been
previously connected.

We say that a separator $Z$ of a graph $G$ is \emph{connected} if $G[Z]$ is
connected. Recall that the separator returned by \cref{thm:separator} is
connected. The following argument is reminiscent of ideas in Grigni et al.~\cite{GrigniKP95} and Klein~\cite{Klein08} as well as of the Christofides-Serdyukov algorithm~\cite{Christofides76,Serdyukov78} by its use of $T$-joins.

\begin{lemma} \label{lem:tour-sep}
Let $G$ be an $n$-vertex graph and let $Z \subset V(G)$ be a non-empty connected separator of $G$. Consider any bi-partition of the connected components of $G-Z$ and let $A, B \subseteq V(G)$ denote the sets of vertices of the two parts.
\begin{enumerate}[label=(\roman*)]
\item\label{lem:tour-sep-merge} If $R_A$ and $R_B$ are tours of $G[A \cup Z]$ and $G[B \cup Z]$ respectively, then the union of these tours is a tour of $G$.
\item\label{lem:tour-sep-split} If $R$ is a tour of $G$, then $G[A \cup Z]$ and $G[B \cup Z]$ have tours $R_A$ and $R_B$ respectively such that $|R_A|+|R_B| \leq |R| + 4\cdot |Z|$.
\end{enumerate}
\end{lemma}
\begin{proof}
To see~\ref{lem:tour-sep-merge}, note that since $Z$ is non-empty, the tours $R_A$ and $R_B$ share a vertex (namely, any vertex of $Z$). Merge the tours on this vertex to obtain a tour $R$ of $G$.

To see~\ref{lem:tour-sep-split}, note that since the separator is connected,
$G[Z]$ has a spanning tree $S$. By the same reasoning, $G[A \cup Z]$ and $G[B
\cup Z]$ are connected and thus each admits a tour. Apply the following
construction with respect to $A$. Let ${\cal W}$ be the set of walks that is
obtained from $R$ by removing all vertices of $B$ and all edges of $E(G[Z])$
from $R$. Let $H$ be the multi-graph on $A \cup Z$ induced by the walks of
${\cal W}$ and the edges of $S$. Let $T$ be the set of vertices in $H$ that have
odd degree. Since the ends of the non-closed walks in ${\cal W}$ are all in $Z$,
$T \subseteq Z$ and $|T|$ is even. Let $F \subseteq E(S)$ be a $T$-join in $S$
that is a union of edge-disjoint paths, and let ${\cal L}$ be this set of paths. Recall from the previous discussion that ${\cal L}$ exists, since $S$ is a tree. Let $H'$ be the multi-graph on $A \cup Z$ induced by the walks of ${\cal W}$, the edges of $S$, and the paths of ${\cal L}$. Note that $H'$ is connected, because $S$ is connected and all walks in ${\cal W}$ have a vertex in $Z$, and that every vertex in $H'$ has even degree by the definition of a $T$-join. Hence, $H'$ has an Eulerian tour $R_A$. Observe that $R_A$ is a tour of $G[A \cup Z]$ by construction. Moreover, $H'$ has at most $2\cdot |Z|$ more edges than ${\cal W}$, since ${\cal L}$ is a set of edge-disjoint paths in $S$.

By the same construction, we obtain a tour $R_B$ of $G[B \cup Z]$. Observe that the total length increases by at most $4 \cdot |Z|$, since the sets ${\cal W}$ of walks constructed with respect to $A$ and $B$ are edge disjoint.
\end{proof}

\begin{lemma} \label{lem:tour-analysis}
Let ${\cal G}$ be a graph class that admits a connected separator that satisfies the conditions of \cref{lem:sep-pump} with functions $\alpha, {\sf sep}$ for the balance and size of the separator respectively. Suppose that ${\sf sep}(n)/\alpha(n) = \Oh(\sqrt{n})$. Let $G \in {\cal G}$. For some integer $r$, let ${\cal P}$ be the weak $r$-division that is constructed for $G$ following \cref{lem:division} using \cref{lem:sep-pump} applied to the assumed separator. Then:
\begin{enumerate}[label=(\roman*)]
\item\label{lem:tour-analysis-merge} For each $P \in {\cal P}$, let $R_P$ be a tour of $G[P]$. The union of these tours is a tour for $G$;
\item\label{lem:tour-analysis-split} $G$ has a tour $R$ of length $\tau(G) + \Oh(n/\sqrt{r})$ such that $R$ contains a closed subwalk that is a tour of $G[P]$ for each $P \in {\cal P}$.
\end{enumerate}
\end{lemma}
\begin{proof}
This immediately follows from \cref{lem:tour-sep} and from the recursive constructions and analyses of \cref{lem:sep-pump} and \cref{lem:division}. Note that the fact that the separator is connected ensures that each of the produced parts in the recursive construction again induces a connected subgraph, and thus the separator is always non-empty.
\end{proof}

We now again apply the approach of Lipton and Tarjan~\cite{LiptonT80} in a similar fashion as in \cref{lem:div-to-algo}.

\begin{lemma} \label{lem:tsp-gen}
Let ${\cal G}$ be a graph class and let $\mathsf{div},\tw: \mathbb{N} \rightarrow \mathbb{N}$ be monotone non-decreasing functions possibly depending on ${\cal G}$). Suppose that for any $n$-vertex graph $G \in {\cal G}$ has treewidth at most $\tw(n)$. Moreover, suppose that ${\cal G}$ admits a connected separator that satisfies the conditions of \cref{lem:tour-analysis}. Then, by \cref{lem:division} and \cref{lem:sep-pump}, there exists a constant integer $s$ such that for any $s < r < n$, one can compute in $\mathsf{div}(n)$ time a weak $r$-division of $G$ such that $G[P] \in {\cal G}$ for each group $P \in {\cal P}$. If the {\sc Traveling Salesperson} problem can be solved in $\mathcal{T}(k,n)$ time on $n$-vertex graphs of treewidth $k$, then for any $\epsilon > 0$, the {\sc Traveling Salesperson} problem has a $(1+\eps)$-approximation algorithm running in time $\mathsf{div}(n) + \mathcal{T}(\tw(\Oh(1/\eps^2)),\Oh(1/\eps^2)) \cdot \Oh(n/\eps^2)$.
\end{lemma}
\begin{proof}
Consider any $n$-vertex graph $G \in {\cal G}$. We may assume that $G$ is connected, because otherwise we can simply answer that $G$ has no tour. Next, compute a weak $(1/\eps^2)$-division ${\cal P}$ of $G$ in $\mathsf{div}(n)$ time using the assumed algorithm. By definition, ${\cal P}$ has $\Oh(n/\eps^2)$ groups. Each group $P$ of ${\cal P}$ has $\Oh(1/\eps^2)$ vertices and $G[P] \in {\cal G}$. By the assumption of the lemma, $G[P]$ has treewidth $\tw(\Oh(1/\eps^2))$. For each $P \in {\cal P}$, apply the assumed algorithm to $G[P]$ to obtain a tour $R_P$. Finally, apply \cref{lem:tour-analysis}\ref{lem:tour-analysis-merge} to merge the tours into a tour $R$ of $G$.

We first analyze the approximation ratio. Observe that $\tau(G) \geq n$ and, by definition, $|\bnd {\cal P}| = \Oh(\eps n)$. Hence, by \cref{lem:tour-analysis}\ref{lem:tour-analysis-split}, $G$ has a tour $R^*$ of length at most $\tau(G) + \Oh(\eps n) \leq (1+\Oh(\eps)) \cdot \tau(G)$ such that $R^{*}$ contains a closed subwalk that is a tour of $G[P]$ for each $P \in {\cal P}$. Observe that $R$ is obtained as a union of optimal tours of $G[P]$ for each $P \in {\cal P}$. Hence, $|R| \leq (1+\Oh(\eps)) \cdot \tau(G)$. The lemma follows by adjusting $\epsilon$ before the start of the algorithm.

It is immediate by \cref{lem:tour-analysis}\ref{lem:tour-analysis-merge} and the description of the algorithm that it computes a tour of $G$ in the stated time.
\end{proof}

Finally, we complete the proof of~\cref{thm:tsp}.

\theoremtsp*
\begin{proof}
Recall that by \cref{prp:tw-hyper}, 
the treewidth of any $n$-vertex planar $\delta$-hyperbolic graph is $\Oh(\delta \log n)$. It is known that the {\sc Traveling Salesperson} problem can be solved in $2^{\Oh(\tw)} \cdot n$ time on $n$-vertex graphs of treewidth $\tw$~\cite[Appendix~D]{Le18} (see also Dorn et al.~\cite[Lemma~2]{DornPBF10} for such an algorithm on planar graphs of bounded treewidth). Additionally, a tree decomposition of width $\Oh(\tw)$ can be computed in $2^{\Oh(\tw)} n$ time~\cite{BodlaenderDDFLP16,Korhonen21} (or even faster on planar graphs~\cite{KammerT16,GuX19}). We now apply the algorithm of \cref{lem:tsp-gen} with \cref{cor:oursep-pump}. The analysis is the same as in \cref{thm:mis}.
\end{proof}

\section{Lower Bound for \textsc{Independent Set}} \label{sec:lower}
The goal of this section is to prove the following theorem.

\theoremlower*

Let $\grid_n$ denote the grid of size $n\times n$ where we have added the
diagonal $(a,b)(a+1,b+1)$ in each cell. The starting point of our reduction will
be the following result of De Berg~\etal\cite{BergBKMZ20}.

\begin{theorem}[Implicit in Theorems 3.2 and 4.2 of \cite{BergBKMZ20}]
\label{thm:ETH_starter}
    There is no $2^{o(n)}$ algorithm for \textsc{Independent Set} in graphs given as (induced) subgraphs of $\grid_n$, unless ETH fails.
\end{theorem}

Observe that the maximum independent set of $H$ can be associated with a maximum independent set of any graph $H'$ that is obtained by subdividing each edge of $H$ an even number of times~\cite{Po74}. Such graphs $H'$ are called even subdivisions of $H$. More precisely, $H$ has an independent set of size $k$ if and only if $H'$ has an independent set of size $k+\frac{|V(H')|-|V(H)|}{2}$.

Let $G$ be a given subgraph of $\grid_n$. Our goal now is to construct a corresponding $\delta$-hyperbolic graph $G'$ that has an independent set of size $k_0+k$ for some fixed $k$ if and only if $G$ has an independent set of size $k_0$. This will be achieved in several steps: first we construct a simple $\Theta(\delta)$-hyperbolic graph $B$, which is then shown to contain some subdivision $\grid_B$ of $\grid_n$. This implies that $B$ contains a subdivision $\Gamma$ of $G$. We will then modify $B$ with simple gadgetry to get the desired graph $G'$ where we can prove the correspondence between independent sets in $G'$ and $G$.

\subparagraph*{Constructing a simple $\Theta(\delta)$-hyperbolic graph $B_\grid$.}

\begin{figure}[t]
    \centering
    \includegraphics[scale=0.7]{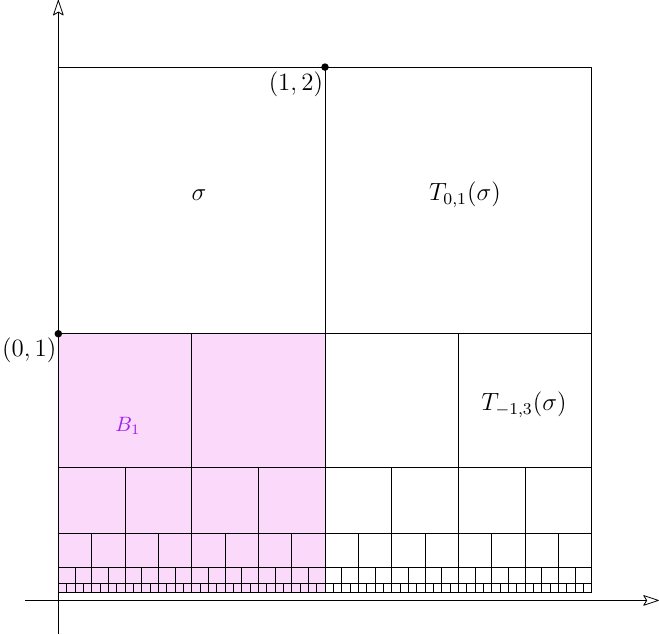}
    \caption{Binary tiling $B_0$ defined through the transformation $T_{a,b}$. Transforming the square $\sigma=[0,1]\times [1,2]$ with $T_{a,b}$ for all integers $a,b$ results in a tiling of the upper half-plane $\Uu$. The shaded part belongs to the finite subgraph $B_1$.}
    \label{fig:binarytiling}
\end{figure}

Consider the \emph{binary tiling}  of the hyperbolic plane introduced by B\"or\"oczky~\cite{Boro}, which we define next as seen in~\cite{cannon97}. 
The tiling can be defined in the open upper half-space $\Uu:=\{(x,y)\in \real^2
\mid y>0\}$ as follows. For $(a,b)\in \Z^2$ define the transformation $T_{a,b}:\Uu\rightarrow \Uu$ as
\[T_{a,b}(p) := 2^a \cdot (p+(b,0)),\]
that is, $T_{a,b}$ translates with a vector $(b,0)$ for some $b \in \Z$ and then scales all coordinates by $2^a$ for some $a \in \Z$. Let $\sigma=[0,1]\times [1,2]$ be the unit square with lower left corner $(0,1)$, see Figure~\ref{fig:binarytiling}. Observe that $\{T_{a,b}(\sigma)\mid (a,b)\in \Z^2\}$ is a tiling of $\Uu$ with squares. The drawing of this tiling gives rise to a graph where vertices are the tile corners and edges are either edges or half-edges of the tiles. We call this infinite planar graph the \emph{binary tiling graph}. One can manually prove (by checking the slim triangle property) that the binary tiling graph $B_0$ is $\delta_0$-hyperbolic for some constant $\delta_0$. Alternatively, one can use the fact that $B_0$ is \emph{quasi-isometric} to the \emph{hyperbolic plane} to show that the $B_0$ has some constant hyperbolicity $\delta_0$. We do not introduce these terms here, but refer the interested reader to Gromov's original work~\cite{gromov87}, the book chapter ~\cite[III.H]{BH99}, and the notes~\cite{cannon97}.

Observe now that in $B_0$ the shortest paths between any pair of vertices $(x,y)$ and $(x',y')$ where $0\leq x\leq x'\leq 1$ and $y,y'\leq 1$ stays in the rectangle $[0,1]\times [\min(y,y'),1]$. Consequently, the part of $B_0$ that falls in a rectangle $[0,1]\times [y_0,1]$ (with $y_0\in \{2^k\mid k\in \Z,\, k\leq 0\}$) forms a geodesic subgraph of $B_0$, thus the values of Gromov products are preserved and the part remains $\delta_0$-hyperbolic.

Let $\hat \delta$ denote the least integer multiple of $4$ that is greater than $\delta$. Consider the subgraph $B_1$ of $B_0$ induced by the vertices inside the rectangle $[0,1]\times [2^{-\lceil 9n/{\hat \delta} \rceil},1]$. (Recall that we are given a graph $G$ as a subgraph of $\grid_n$, so $n$ is the grid size.) The above argument shows that $B_1$ is a geodesic subgraph of $B_0$ and therefore it is $\delta_0$-hyperbolic.

Let us now subdivide each tile of $B_1$ into a grid of size $\hat \delta \times \hat \delta$, and adding diagonals to get a copy of $\grid_{\hat \delta}$ in each tile. More precisely, we subdivide each edge of $B_1$ with $\hat \delta -1$ vertices that are placed at equal distances to get a new graph $B_1^\delta$.
Consider some tile $\sigma$ of $B_1$. Notice that in $B_1$ the tile $\sigma$ had $5$ boundary edges (two edges along the bottom), so in $B_1^\delta$ its boundary cycle $\sigma_T$ consists of $5\hat \delta$ edges. We divide this square $\sigma$ into $\hat \delta^2$ small squares of the same size with $\hat \delta-1$ horizontal lines and $\hat \delta-1$ vertical lines. Notice that these lines land at every other vertex on the bottom of the boundary cycle $C_\sigma$. In each tile we also add a diagonal. Finally, for each edge of the resulting drawing that is not an edge of $B_1^\delta$ we subdivide the edge $2$ times. We apply the same procedure in all tiles $\sigma$ of $B_1^\delta$. 

Let $B_\grid$ denote the planar graph created this way. 
See \cref{fig:lowerbound_defs}.
We call the edges of $B_1^\delta$ \emph{boundary} edges, and all other edges of $B_\grid$ are called \emph{inner} edges. For technical reasons, we will need to work with some custom graph $B$ that arises from $B_\grid$ by subdividing the inner edges $E(B_\grid)\setminus E(B_1^\delta)$ \emph{at most once}. Let $\Bb_\grid$ denote the set of graphs that can be obtained from $B_\grid$ by subdividing inner edges at most once.
\begin{figure}
    \centering
    \includegraphics[width = \textwidth]{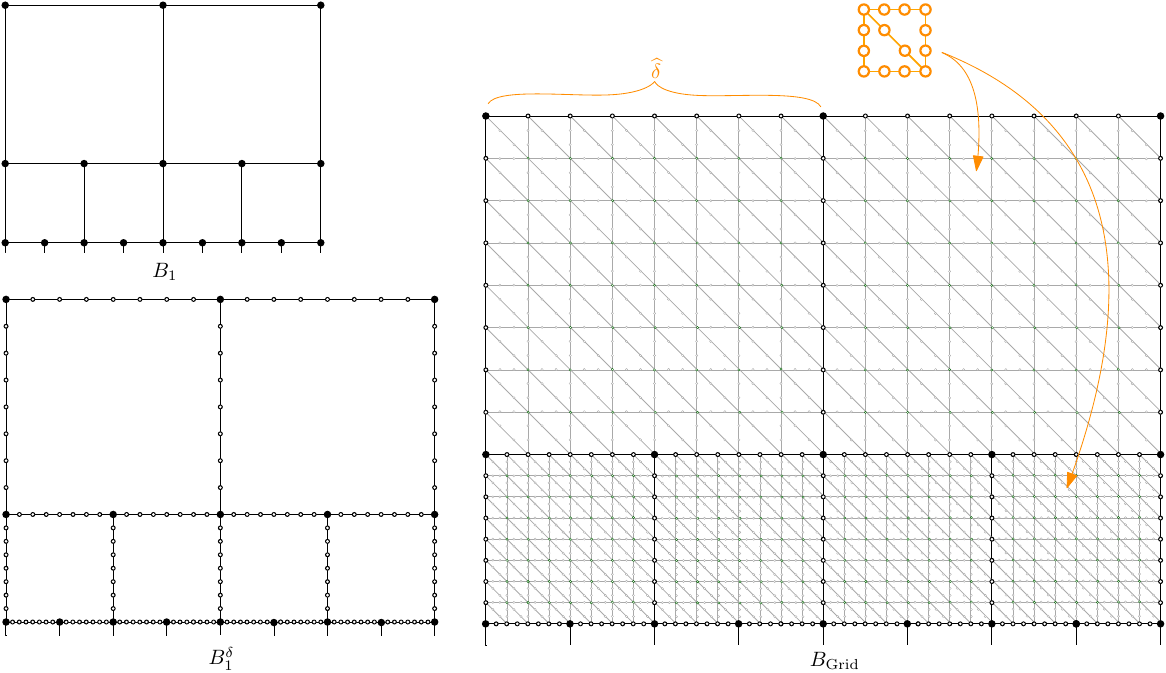}
    \caption{Partial view of $B_1$, $B_1^\delta$, and $B_\grid$ for $\hat \delta=8$.}
    \label{fig:lowerbound_defs}
\end{figure}

\begin{lemma}\label{lem:B_hyperbolicity}
    Each graph $B\in \Bb_\grid$ has hyperbolicity $\Theta(\delta)$.
\end{lemma}

\begin{proof}
    First we show that each $B\in\Bb_\grid$ has a triangle with slimness $\Omega(\hat \delta)=\Omega(\delta)$, which implies that the hyperbolicity of $B$ is $\Omega(\delta)$.
    Let $\sigma$ denote the tile in the bottom right corner of $B_1$, and let $a$ be the bottom right corner of $B_1$. Let $b,c$ be the neighbors of $a$ in $B_1$, and consider the unique triangle of geodesics given by $a,b,c$ in $B_\grid$, which is given by the boundary of $\sigma$. Notice that the boundary of $\sigma$ is also geodesic when some inner edges of $B_\grid$ are subdivided to get $B$: indeed, for any pair of vertices on the boundary the shortest path along boundary edges remains the shortest path overall, i.e., the boundary cycle is geodesic. The proposed triangle has a side of length $2\hat \delta$ and its slimness is therefore $\hat \delta$, implying hyperbolicity $\Omega(\hat \delta)=\Omega(\delta)$.
    
    To show that the hyperbolicity of $B$ is $\Oh(\delta)$, we will first prove
    that $B_1^\delta$ has hyperbolicity $\Oh(\delta)$.
    Since $B_1^\delta$ is obtained from $B_1$ by subdividing each edge $\hat
    \delta -1$ times, clearly, for all $a,b\in V(B_1)$, we have $\hat\delta \cdot \dist_{B_1}(a,b)=
    \dist_{B_1^\delta}(a,b)$. When $a \in V(B_1^\delta)\setminus V(B_1)$, then let
    $a_1,a_2\in V(B_1)$ be the endpoints of the edge in $B_1$ for which $a$ was
    a subdividing vertex. Observe that a shortest path from $a$ that has length
    at least $\hat\delta$ must also go through either $a_1$ or
    $a_2$. Thus for any $a,b\in V(B_1^\delta)\setminus V(B_1)$ we get
    \[|\dist_{B_1^\delta}(a,b) - \hat\delta \cdot \dist_{B_1}\!(a',b')| =
    \Oh(\hat \delta) = \Oh(\delta),\]
    where $a',b'\in V(B_1)$ are the nearest vertices of $B_1$ to $a$ and $b$, respectively. Using the same notation for $x\in V(B_1^\delta)$ and $x'\in V(B_1)$ and denoting the Gromov product in graph $G$ with $(.,.)_.^G$, we get for any $a,b,x\in V(B_1^\delta)$ the following:
    \[|(a,b)_x^{B_1^\delta} - \hat\delta \cdot (a',b')^{B_1}_{x'} | = \Oh(\delta).\]
    Consequently, by the $\delta_0$-hyperbolicity of $B_1$ for any quadruple $x,y,z,w\in V(B_1^\delta)$ we get:
    \begin{align*}
    \min\big( (x,y)_w^{B_1^\delta}, (y,z)_w^{B_1^\delta}\big)-
    (x,z)_w^{B_1^\delta} &\leq \hat\delta \cdot \Big( \min\big(
    (x',y')_{w'}^{B_1}, (y',z')_{w'}^{B_1}\big)- (x',z')_{w'}^{B_1}
\Big)+\Oh(\delta)\\ &\leq \hat \delta \delta_0 +\Oh(\delta)=\Oh(\delta),
    \end{align*}
    thus $B_1^\delta$ has hyperbolicity $\Oh(\delta)$.

    Moving on to some $B\in \Bb_\grid$, consider now some vertex $x\in V(B)$.
    The vertex $x$ falls in some tile $\sigma_x$ of $B_1$; let $x'$ denote the
    lower left corner of $\sigma_x$. (When $x$ is on the boundary of two tiles,
    i.e., on some edge of $B_1$, then select $x'$ so that it is to the left or
below $x$ on the same edge of $B_1$.)  Note that $\dist_{B}(x,x')=\Oh(\hat \delta)$. Due to the subdivisions in the construction of $B_\grid$ we have that for any $a,b\in V(B_1^\delta)$ any shortest path between $a$ and $b$ in $B$ is using only boundary edges, and thus
    $|\dist_B(x,y) - \dist_{B_1^\delta}(x',y')| = |\dist_B(x,y) -
    \dist_B(x',y')|  = \Oh(\hat \delta)$. It follows that the Gromov product
    $(x,y)^B_w$ in $B$ differs from the Gromov product
    $(x',y')^{B_1^\delta}_{w'}$ in $B_1^\delta$ by some additive $\Oh(\delta)$.
    Thus the hyperbolicity of $B$ is at most $\Oh(\delta)$ larger than the hyperbolicity of $B_1^\delta$, concluding the proof of the claim.
\end{proof}

\subparagraph*{Embedding a subdivision of $\grid_n$ into $B_\grid$.}

A \emph{row} of $B_1$ consists of tiles of $B_1$ that have the same $y$-projections. We number the rows top to bottom, starting with index $1$. By definition $B_1$ has $\lceil 9n/\hat \delta \rceil-1$ rows, where row $i$ has $2^i$ tiles. 
Let $i_0 = \lceil \log(4n/\hat \delta) \rceil$ be the first row with at least $4n/\hat\delta$ faces.
Our construction will use only rows $i_0+i$ for $i\in [2n]$. In each row we number the faces left to right, so face $j$ of row $i_0+i$ has lower left corner $\big((j-1)\cdot 2^{-i_0-i}, 2^{-i_0-i}\big)$; we denote this tile by $\sigma(i,j)$. 
Inside a tile $\sigma(i,j)$ the graph $B_\grid$ is further subdivided into $\hat\delta \times \hat\delta$ sub-tiles (with diagonals), where rows are also indexed top-down. 
We refer to the lower right corner of the subtile in row $x$ and column $y$ as $v^{\sigma(i,j)}_{x,y}$.

Observe that $B_1$ has $2^{\lceil 9n/\hat \delta \rceil}-2 = 2^{\Oh(n/\delta)}$
tiles, thus $B_\grid$ has $\Oh(\hat\delta^2)\cdot 2^{\Oh(n/\delta)} = 2^{\Oh(n/\delta)}$ vertices.

Consider now the graph $\grid_n$, and let $v_{a,b}$ be its vertex in row $a$ and
column $b$ for $a,b\in \{0,1,\dots,n-1\}$, 
and recall that the graph also has the diagonal
edges $v_{a,b},v_{a+1,b+1}$ for all $a,b\in \{0,1,\dots,n-2\}$.

We are now ready to define an embedding of the vertices of $\grid_n$ into the
vertices of $B_\grid$. We do this by subdividing the $n\times n$ grid into
square parts of size $\hat\delta/4 \times \hat\delta/4$, and put the
corresponding parts, spaced evenly, into some tiles of $B_\grid$. We pick tiles
only in rows of index $i_0+2i-1$
and also in the (first) row we skip even index columns. Since the
number of tiles increases exponentially as we move down the rows, we need to
increase the column indices as well so that the pictures of point $i,j$ and
$i+1,j$ are somewhat aligned vertically, see~\cref{fig:embedding_grid}.

\begin{figure}
    \centering
    \includegraphics[width = \textwidth]{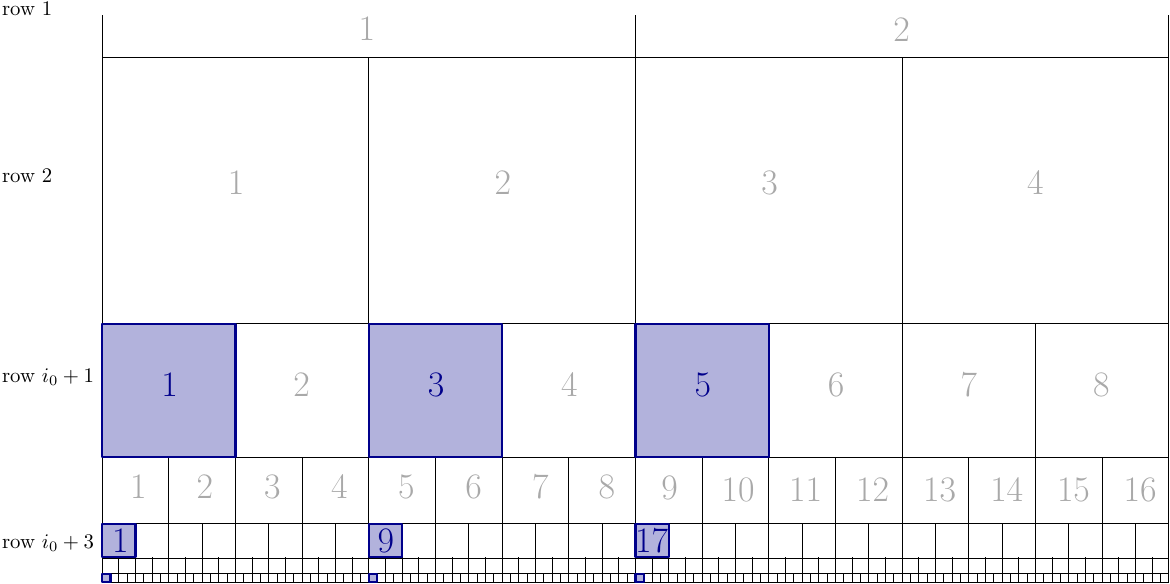}
    \caption{View of embedding $\grid_6$ into $B_\grid$ when $\hat     \delta=8$, together with the numbered faces in each row. The inner grid, edge subdivisions, and diagonal edges in $B_\grid$ are omitted for clarity. Parts of size $\hat \delta/4\times \hat \delta/4$ are assigned to the shaded tiles in $B_\grid$. }
    \label{fig:embedding_grid}
\end{figure}

\begin{figure}
    \centering
    \includegraphics[width = \textwidth]{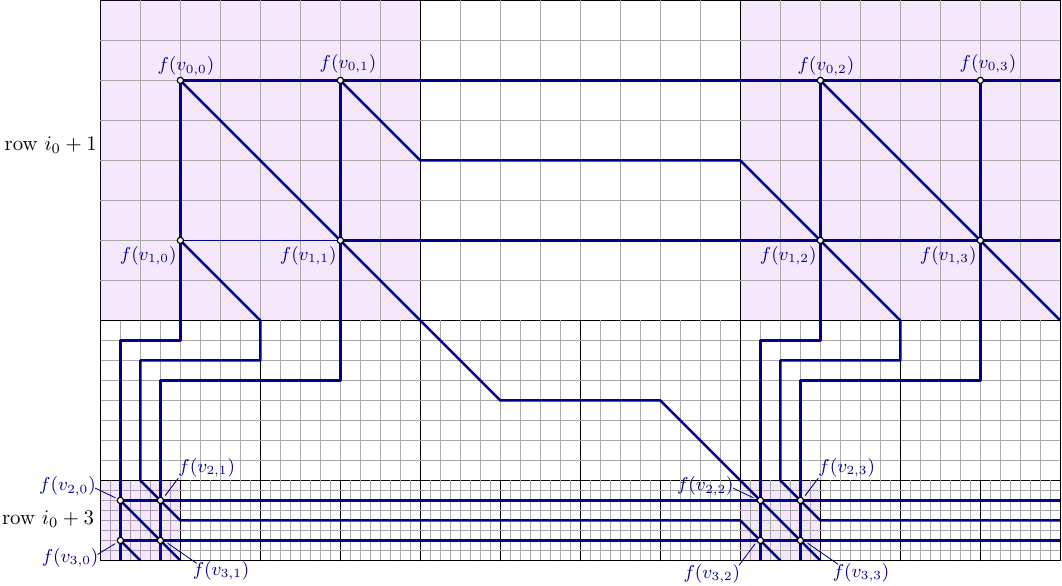}
    \caption{Partial view of embedding $\grid_n$ into $B_\grid$ when $\hat
    \delta=8$. The edge subdivisions and most diagonal edges in $B_\grid$ are omitted for clarity. Parts of size $\hat \delta/4\times \hat \delta/4$ are assigned to the shaded tiles in $B_\grid$. Inside the shaded tiles we have a natural drawing of the relevant part of $\grid_n$, and paths that leave these tiles are connected using the grids in the tiles of $B_\grid$.}
    \label{fig:embedding}
\end{figure}
Formally, let $f_V:V(\grid_n)\rightarrow V(B_\grid)$ as $f(v_{a,b}):= v^{\sigma(i,j)}_{x,y}$ where
\begin{align*}
    i &= i_0+1+2\lfloor a\cdot 4/\hat\delta \rfloor\\
    j &= 1+ 2^{i-i_0}\cdot \lfloor b\cdot 4/\hat\delta \rfloor\\
    x &= 4\cdot (a \bmod \hat\delta/4)+2\\
    y &= 4\cdot (b \bmod \hat\delta/4)+2\\
\end{align*}
and $x \bmod m = x - m\cdot \lfloor x/m \rfloor$.
Next, for each edge $uv$ of $\grid_n$ we need to assign a path $\pi_{uv}$ of
$B_\grid$ connecting $f(u)$ and $f(v)$ in a way so that the union of these
paths gives a subgraph of $B_\grid$ that is a subdivision of $\grid_n$. Such an
assignment is easy to make and is best illustrated with a figure
(see~\cref{fig:embedding}). For vertices $u,v$ in the same row of $\grid_n$ we can simply use a horizontal path to connect $f(u)$ and $f(v)$ in $B_\grid$. Similarly, when $u,v$ are vertical or diagonal neighbors and fall in the same $\hat\delta\times \hat\delta$ group, then we can use a vertical or diagonal paths to connect them inside the corresponding tile of $B_\grid$.

Diagonal neighbors that fall into two horizontally neighboring parts of $\grid_n$ are easily connected by a path that has $2$ diagonal edges, followed by a long horizontal segment and two more diagonal edges. As a result, tiles of $B_\grid$ that are between two horizontal edges will be crossed by $\hat \delta/2 -1$ horizontal paths.

For neighbors that fall into vertically neighboring parts of $\grid_n$, we can describe the corresponding paths as follows. Let $\sigma(i,j)$ and $\sigma(i+2,4j-3)$ be the corresponding tiles. Within $\sigma(i,j)$ we assign the vertical (and diagonal) edges from the bottom row of the $f(.)$ vertices along a straight vertical (respectively, straight diagonal) path, until we hit the boundary of the tile. Similarly, we use straight diagonals and verticals to represent these edges inside $\sigma(i+2,4j-3)$. 
It is then easy to see that in the two tiles $\sigma(i+1,2j-1)$ and $\sigma(i+1,2j)$ below $\sigma(i,j)$ we need to connect $\hat\delta/2 - 1$ 
points that are equally spaced along the union of their top boundary to $\hat\delta/2 - 1$ points that are placed equally along the shared boundary of $\sigma(i+1,2j-1)$ and $\sigma(i+2,4j-3)$, preserving left-to right orders. Here the odd index connections correspond to vertical edges of $\grid_n$ while the even ones correspond to diagonal connections. Clearly these connections can be achieved using $\hat \delta/2 - 1$ rectilinear paths inside the $2\hat\delta\times \hat \delta$ grid inside the tiles $\sigma(i+1,2j-1)$ and $\sigma(i+1,2j)$.
Finally, for diagonal neighbors that are in diagonally neighboring parts, let $\sigma(i,j)$ and $\sigma(i+2,j')$ be the tiles assigned to the diagonal parts. We connect the bottom right $f(.,.)$-vertex in $\sigma(i,j)$ with a straight diagonal to the middle of the tile $\sigma(i+1,2j+1)$, which is horizontally connected to the middle of $\sigma(i+1,(j'-1)/2)$, which is connected with a straight diagonal to the top left vertex in $\sigma(i+2,j')$. This concludes the description of the embedding. It is routine to check that the union of the described paths is a subdivision of $\grid_n$ which is a subgraph of $B_\grid$. Since $f$ assigns all vertices of $\grid_n$ to vertices in $B_\grid$ that are not on tile boundaries, we have proved the following observation.

\begin{observation}\label{obs:embedding}
We can construct an embedding of $\grid_n$ into $B_\grid$ such that each path $\pi_{uv}$ corresponding to edge $uv$ of $\grid_n$ contains at least one inner edge of $B_\grid$. The construction takes time that is polynomial in $|V(B_\grid)|$.
\end{observation}

\subparagraph*{Construction wrap-up and lower bound proof.}

Given the input subgraph $G$ of $\grid_n$, we can use \Cref{obs:embedding} to get a subgraph $G_1$ of $B_\grid$ that is a subdivision of $G$. Note that  
$G_1$ may not be an \emph{even} subdivision, but whenever a path $\pi_{uv}$ has an even length, we can subdivide some inner edge of $\pi_{uv}$. Subdividing the required set of inner edges in $B_\grid$ leads to a graph $B\in \Bb_\grid$, and a corresponding subgraph $G_B$ that is now an even subdivision of $G$.

Next, we will change $B$ further by adding new vertices to the parts that are not in $G_B$ as follows. For each vertex $v\in V(B)\setminus V(G_B)$ we assign two new neighboring vertices $v_1$ and $v_2$; these new vertices will be connected only to $v$ itself. The resulting graph $G'$ still has $G_B$ as a subgraph, which is an even subdivision of $G$, and since adding new degree-one vertices to a graph does not alter its hyperbolicity, we know by \Cref{lem:B_hyperbolicity} that $G'$ has hyperbolicity $\Theta(\delta)$.

\begin{lemma}\label{lem:construction_correct}
    The graph $G$ has an independent set of size $k$ if and only if $G'$ has an independent set of size $k+\frac{|V(G_B)|-|V(G)|}{2} + \frac{2}{3}|V(G')\setminus V(G_B)|$.
\end{lemma}

\begin{proof}
    Since $G_B$ is an even subdivision of $G$, it has an independent set of size $k+\frac{|V(G_B)|-|V(G)|}{2}$ if and only if $G$ has an independent set of size $k$. Therefore it is sufficient to show that $G_B$ has an independent set of size $\ell$ if and only if $G'$ has an independent set of size $\ell + \frac{2}{3}|V(G')\setminus V(G_B)|$. Indeed, if $X$ is an independent set of size $\ell$ in $G_B$, then adding the vertices $v_1,v_2$ for all $v\in V(G')\setminus V(G)$ gives an independent set of the desired size in $G'$.
    In the other direction, if $X$ is an independent set in $G'$, then we can change it to a set $X'$ that contains all vertices $v_1,v_2$ for all $v\in V(G_1)\setminus V(G)$ without decreasing the size of $X$: we simply remove all vertices from $X\cap V(G_1)$ and add all vertices of $V(G')\setminus V(G_1)$.
    The resulting independent set $X'$ and the part of $X'$ that falls in $G_B$ have a size difference of exactly $\frac{2}{3}|V(G')\setminus V(G_B)|$, which concludes the proof.
\end{proof}

We are now ready to prove \Cref{thm:lower_bound}.

\begin{proof}[Proof of \Cref{thm:lower_bound}]
    Suppose the contrary: there exists an $n^{o(\delta)}$ algorithm for \textsc{Independent Set} in $\delta$-hyperbolic graphs.
    Let $G$ be given as a subgraph of $\grid_n$, and in time that is polynomial
    in $|V(G')|=\Oh(V(B_\grid))=2^{\Oh(n/\delta)}$ we construct the graph $G'$ as described above.
    Since $G'$ is $\Oh(\delta)$-hyperbolic, we can solve \textsc{Independent
    Set} in $G'$ in $|V(G')|^{o(\delta)}=2^{\Oh(n/\delta)\cdot o(\delta)}=
    2^{o(n)}$ time. Altogether, the above algorithm solves \textsc{Independent
    Set} in $G$ in $2^{\Oh(n/\delta)} + 2^{o(n)} =2^{o(n)}$ time, which (by \Cref{thm:ETH_starter}) contradicts ETH.
\end{proof}

\section{Discussion and Open Problems} \label{sec:discussion}
We conclude this paper with a discussion of future directions that have sprung up from our work.

A first natural question is whether there is a $n^{o(\delta)}$ time lower bound to solve the {\sc Traveling Salesperson} problem on (planar) $\delta$-hyperbolic graphs under the Exponential Time Hypothesis (ETH). This would match the running time of the exact algorithm of \cref{cor:tsp-exact} and also imply that the running time of the approximation scheme of \cref{thm:tsp} is essentially best possible. One might expect that this could be possible by suitably adjusting the construction of \cref{thm:lower_bound}. However, we recall that a crucial aspect of this construction was to plant a subdivision of a grid inside a large $\delta$-hyperbolic graph, somehow without this affecting the hardness proof. While for {\sc Independent Set} this was possible by some local modifications, it seems much harder to achieve this for the {\sc Traveling Salesperson} problem.

A next algorithmic question is whether there exist fast approximation schemes (i.e., a near-linear time FPTAS) for other problems on planar $\delta$-hyperbolic graphs. This paper develops such a scheme for {\sc Maximum Independent Set} by employing ideas of the recursive separator approach of Lipton and Tarjan~\cite{LiptonT80}. It is known, however, that this approach can also lead to approximation schemes for e.g.\ {\sc Minimum Vertex Cover}~\cite{ChibaNS81} and {\sc Minimum Feedback Vertex Set}~\cite{KleinbergK01}. Both these algorithms rely on the contraction of edges to ensure the solution size is at least linear in the number of vertices. However, planar $\delta$-hyperbolic graphs are not closed under edge contraction\footnote{A well-chosen $n$-size patch of the binary tiling (call this $G$) is constant-hyperbolic, but has an $\Omega(\log n)$-sized grid subdivision as a subgraph (call this $H$). If we contract some of the edges of $H$, we get a graph $G'$ that has a grid $H'$ of size $\Omega(\log n)$ in it. $H'$ itself has hyperbolicity $\Omega(\log n)$, and it is not too hard to show that even in $G'$ the shortest paths among vertices of $H'$ stay inside $H'$. So the hyperbolicity of $G'$ is $\Omega(\log n)$. However, it arose as a contraction of $G$, which has constant hyperbolicity.} and thus this approach cannot succeed. The non-hereditary nature of (planar) $\delta$-hyperbolic graphs stands in the way of other approaches as well. For {\sc Minimum Vertex Cover} in particular, a Baker-style approach~\cite{Baker94} is unlikely to succeed for this reason. Moreover, a standard conversion trick (see e.g.~\cite{Bar-YehudaHR11,Har-Peled20} and references therein) from an approximation scheme for {\sc Maximum Independent Set} to one for {\sc Minimum Vertex Cover} does not seem feasible for the same reason, since it essentially requires reduction (by vertex deletion) to a parameterized problem kernel. A similar approach is also known for {\sc Minimum Feedback Vertex Set}~\cite{BorradaileLZ19} and thus is unlikely to translate to our setting. Hence, {\sc Minimum Vertex Cover} and {\sc Minimum Feedback Vertex Set} seem interesting problems on planar $\delta$-hyperbolic graphs.

In the same vein, we wonder about approximation schemes for other network design problems. This would extend our scheme for the {\sc Traveling Salesperson} problem.  Therefore, we ask about a near-linear time FPTAS for problems such as {\sc Subset TSP} and {\sc Steiner Tree}. We deem the existence of such schemes to be highly plausible.

Turning to our separator theorem, we notice that the geodesic cycle separator that is returned by \cref{thm:separator} has size $\Oh(\delta)$ but balance $2^{-\Oh(\delta)}/\log n$. While we know that the size bound cannot be improved by the example of a $\delta$-cylinder, we do not know of any example that shows that the balance factor should be $2^{-\Oh(\delta)}/\log n$. Therefore, we ask whether we can find, in near-linear time, a geodesic cycle separator of balance equal to some constant.

Last but not least, we consider how to determine the hyperbolicity of a graph.
This can be naively done in $\Oh(n^4)$ time by following the definition of hyperbolicity. Fournier et al.~\cite{FournierIV15} improved this to $\Oh(n^{3.69})$ and even gave a $2$-approximation that runs in $\Oh(n^{2.69})$ time. Considering the context of planarity, Borassi et al.~\cite{BorassiCH16} proved that the hyperbolicity of a general sparse graph cannot be computed in subquadratic time under the Strong Exponential Time Hypothesis (SETH). However, this general result does not exclude the existence of a subquadratic algorithm on planar graphs. A potentially encouraging sign in this direction is the known linear-time algorithm for outerplanar graphs~\cite{CohenCDL17}. Hence, we repeat the question of Cohen et al.~\cite{CohenCDL17} and ask for a linear-time algorithm that computes the  hyperbolicity of planar graphs.

\bibliographystyle{plain}
\bibliography{hyprbolicrefs}

\begin{thebibliography}{10}

\bibitem{Abu-AtaD16}
Muad Abu{-}Ata and Feodor~F. Dragan.
\newblock {Metric tree-like structures in real-world networks: an empirical
  study}.
\newblock {\em Networks}, 67(1):49--68, 2016.

\bibitem{AdcockSM16}
Aaron~B. Adcock, Blair~D. Sullivan, and Michael~W. Mahoney.
\newblock {Tree decompositions and social graphs}.
\newblock {\em Internet Math.}, 12(5):315--361, 2016.

\bibitem{AppelH1989}
Kenneth Appel and Wolfgang Haken.
\newblock {\em {Every Planar Map is Four-Colorable}}, volume~98 of {\em
  Contemporary Mathematics}.
\newblock American Mathematical Society, 1989.

\bibitem{ArnborgP89}
Stefan Arnborg and Andrzej Proskurowski.
\newblock {Linear time algorithms for {NP}-hard problems restricted to partial
  k-trees}.
\newblock {\em Discret. Appl. Math.}, 23(1):11--24, 1989.

\bibitem{Baker94}
Brenda~S. Baker.
\newblock {Approximation Algorithms for NP-Complete Problems on Planar Graphs}.
\newblock {\em J. {ACM}}, 41(1):153--180, 1994.

\bibitem{BandeltM86}
Hans{-}J{\"{u}}rgen Bandelt and Henry~Martyn Mulder.
\newblock {Distance-hereditary graphs}.
\newblock {\em J. Comb. Theory, Ser. {B}}, 41(2):182--208, 1986.

\bibitem{Bar-YehudaHR11}
Reuven Bar{-}Yehuda, Danny Hermelin, and Dror Rawitz.
\newblock {Minimum vertex cover in rectangle graphs}.
\newblock {\em Comput. Geom.}, 44(6-7):356--364, 2011.

\bibitem{BlasiusFFK20}
Thomas Bl{\"{a}}sius, Philipp Fischbeck, Tobias Friedrich, and Maximilian
  Katzmann.
\newblock {Solving Vertex Cover in Polynomial Time on Hyperbolic Random
  Graphs}.
\newblock In Christophe Paul and Markus Bl{\"{a}}ser, editors, {\em 37th
  International Symposium on Theoretical Aspects of Computer Science, {STACS}
  2020, March 10-13, 2020, Montpellier, France}, volume 154 of {\em LIPIcs},
  pages 25:1--25:14. Schloss Dagstuhl - Leibniz-Zentrum f{\"{u}}r Informatik,
  2020.

\bibitem{BlasiusFK16}
Thomas Bl{\"{a}}sius, Tobias Friedrich, and Anton Krohmer.
\newblock {Hyperbolic Random Graphs: Separators and Treewidth}.
\newblock In Piotr Sankowski and Christos~D. Zaroliagis, editors, {\em 24th
  Annual European Symposium on Algorithms, {ESA} 2016, August 22-24, 2016,
  Aarhus, Denmark}, volume~57 of {\em LIPIcs}, pages 15:1--15:16. Schloss
  Dagstuhl - Leibniz-Zentrum f{\"{u}}r Informatik, 2016.

\bibitem{BlasiusFK18}
Thomas Bl{\"{a}}sius, Tobias Friedrich, and Anton Krohmer.
\newblock {Cliques in Hyperbolic Random Graphs}.
\newblock {\em Algorithmica}, 80(8):2324--2344, 2018.

\bibitem{Bodlaender93}
Hans~L. Bodlaender.
\newblock {A Tourist Guide through Treewidth}.
\newblock {\em Acta Cybern.}, 11(1-2):1--21, 1993.

\bibitem{Bodlaender96}
Hans~L. Bodlaender.
\newblock {A Linear-Time Algorithm for Finding Tree-Decompositions of Small
  Treewidth}.
\newblock {\em {SIAM} J. Comput.}, 25(6):1305--1317, 1996.

\bibitem{Bodlaender98}
Hans~L. Bodlaender.
\newblock {A Partial \emph{k}-Arboretum of Graphs with Bounded Treewidth}.
\newblock {\em Theor. Comput. Sci.}, 209(1-2):1--45, 1998.

\bibitem{Bodlaender05}
Hans~L. Bodlaender.
\newblock {Discovering Treewidth}.
\newblock In Peter Vojt{\'{a}}s, M{\'{a}}ria Bielikov{\'{a}}, Bernadette
  Charron{-}Bost, and Ondrej S{\'{y}}kora, editors, {\em {SOFSEM} 2005: Theory
  and Practice of Computer Science, 31st Conference on Current Trends in Theory
  and Practice of Computer Science, Liptovsk{\'{y}} J{\'{a}}n, Slovakia,
  January 22-28, 2005, Proceedings}, volume 3381 of {\em Lecture Notes in
  Computer Science}, pages 1--16. Springer, 2005.

\bibitem{BodlaenderCKN15}
Hans~L. Bodlaender, Marek Cygan, Stefan Kratsch, and Jesper Nederlof.
\newblock {Deterministic single exponential time algorithms for connectivity
  problems parameterized by treewidth}.
\newblock {\em Inf. Comput.}, 243:86--111, 2015.

\bibitem{BodlaenderDDFLP16}
Hans~L. Bodlaender, P{\aa}l~Gr{\o}n{\aa}s Drange, Markus~S. Dregi, Fedor~V.
  Fomin, Daniel Lokshtanov, and Micha\l{} Pilipczuk.
\newblock {A {$c^k n$} 5-Approximation Algorithm for Treewidth}.
\newblock {\em {SIAM} J. Comput.}, 45(2):317--378, 2016.

\bibitem{bonk2011embeddings}
Mario Bonk and Oded Schramm.
\newblock {Embeddings of Gromov hyperbolic spaces}.
\newblock {\em Selected Works of Oded Schramm}, pages 243--284, 2011.

\bibitem{BorassiCC15}
Michele Borassi, Alessandro Chessa, and Guido Caldarelli.
\newblock {Hyperbolicity measures democracy in real-world networks}.
\newblock {\em Phys. Rev. E}, 92:032812, Sep 2015.

\bibitem{BorassiCH16}
Michele Borassi, Pierluigi Crescenzi, and Michel Habib.
\newblock {Into the Square: On the Complexity of Some Quadratic-time Solvable
  Problems}.
\newblock In Pierluigi Crescenzi and Michele Loreti, editors, {\em Proceedings
  of the 16th Italian Conference on Theoretical Computer Science, {ICTCS} 2015,
  Firenze, Italy, September 9-11, 2015}, volume 322 of {\em Electronic Notes in
  Theoretical Computer Science}, pages 51--67. Elsevier, 2015.

\bibitem{Boro}
K\'aroly B\"or\"oczky.
\newblock G\"ombkit\"olt\'esek \'alland\'o g\"orb\"ulet{\H u} terekben {I}.
\newblock {\em Matematikai Lapok (in Hungarian)}, 25(3-4):265–306, 1974.

\bibitem{BorradaileLZ19}
Glencora Borradaile, Hung Le, and Baigong Zheng.
\newblock {Engineering a {PTAS} for Minimum Feedback Vertex Set in Planar
  Graphs}.
\newblock In Ilias~S. Kotsireas, Panos~M. Pardalos, Konstantinos~E.
  Parsopoulos, Dimitris Souravlias, and Arsenis Tsokas, editors, {\em Analysis
  of Experimental Algorithms - Special Event, SEA{\({^2}\)} 2019, Kalamata,
  Greece, June 24-29, 2019, Revised Selected Papers}, volume 11544 of {\em
  Lecture Notes in Computer Science}, pages 98--113. Springer, 2019.

\bibitem{BH99}
Martin~R. Bridson and Andr{\'e} Haefliger.
\newblock {\em {Metric spaces of non-positive curvature}}, volume 319 of {\em
  Grundlehren der mathematischen Wissenschaften}.
\newblock Springer Berlin, Heidelberg, 1999.

\bibitem{cannon97}
James~W Cannon, William~J Floyd, Richard Kenyon, Walter~R Parry, et~al.
\newblock {Hyperbolic geometry}.
\newblock {\em Flavors of Geometry}, 31(59-115):2, 1997.

\bibitem{ChenFHM12}
Wei Chen, Wenjie Fang, Guangda Hu, and Michael~W. Mahoney.
\newblock {On the Hyperbolicity of Small-World and Tree-Like Random Graphs}.
\newblock In Kun{-}Mao Chao, Tsan{-}sheng Hsu, and Der{-}Tsai Lee, editors,
  {\em Algorithms and Computation - 23rd International Symposium, {ISAAC} 2012,
  Taipei, Taiwan, December 19-21, 2012. Proceedings}, volume 7676 of {\em
  Lecture Notes in Computer Science}, pages 278--288. Springer, 2012.

\bibitem{ChepoiDEHV08}
Victor Chepoi, Feodor~F. Dragan, Bertrand Estellon, Michel Habib, and Yann
  Vax{\`{e}}s.
\newblock {Diameters, centers, and approximating trees of $\delta$-hyperbolic
  geodesic spaces and graphs}.
\newblock In Monique Teillaud, editor, {\em Proceedings of the 24th {ACM}
  Symposium on Computational Geometry, College Park, MD, USA, June 9-11, 2008},
  pages 59--68. {ACM}, 2008.

\bibitem{ChepoiE07}
Victor Chepoi and Bertrand Estellon.
\newblock {Packing and Covering $\delta$-Hyperbolic Spaces by Balls}.
\newblock In Moses Charikar, Klaus Jansen, Omer Reingold, and Jos{\'{e}} D.~P.
  Rolim, editors, {\em Approximation, Randomization, and Combinatorial
  Optimization. Algorithms and Techniques, 10th International Workshop,
  {APPROX} 2007, and 11th International Workshop, {RANDOM} 2007, Princeton, NJ,
  USA, August 20-22, 2007, Proceedings}, volume 4627 of {\em Lecture Notes in
  Computer Science}, pages 59--73. Springer, 2007.

\bibitem{ChibaNS81}
Norishige Chiba, Takao Nishizeki, and Nobuji Saito.
\newblock {Applications of the {L}ipton and {T}arjan's planar separator
  theorem}.
\newblock {\em Journal of Information Processing}, 4(4):203--207, 1981.

\bibitem{Christofides76}
Nicos Christofides.
\newblock {Worst-case analysis of a new heuristic for the traveling salesman
  problem}.
\newblock Technical Report Management Science Research Report 388,
  Carnegie-Mellon University, Pittsburgh, Pennsylvania, USA, 1976.

\bibitem{CohenCDL17}
Nathann Cohen, David Coudert, Guillaume Ducoffe, and Aur{\'{e}}lien Lancin.
\newblock {Applying clique-decomposition for computing Gromov hyperbolicity}.
\newblock {\em Theor. Comput. Sci.}, 690:114--139, 2017.

\bibitem{CohenCL12}
Nathann Cohen, David Coudert, and Aur{\'e}lien Lancin.
\newblock {Exact and Approximate Algorithms for Computing the Hyperbolicity of
  Large-Scale Graphs}.
\newblock Technical report, INRIA, September 2012.

\bibitem{Cohen-AddadFKL20}
Vincent Cohen{-}Addad, Arnold Filtser, Philip~N. Klein, and Hung Le.
\newblock {On Light Spanners, Low-treewidth Embeddings and Efficient Traversing
  in Minor-free Graphs}.
\newblock In Sandy Irani, editor, {\em 61st {IEEE} Annual Symposium on
  Foundations of Computer Science, {FOCS} 2020, Durham, NC, USA, November
  16-19, 2020}, pages 589--600. {IEEE}, 2020.

\bibitem{CoudertDP19}
David Coudert, Guillaume Ducoffe, and Alexandru Popa.
\newblock {Fully Polynomial {FPT} Algorithms for Some Classes of Bounded
  Clique-width Graphs}.
\newblock {\em {ACM} Trans. Algorithms}, 15(3):33:1--33:57, 2019.

\bibitem{Courcelle90}
Bruno Courcelle.
\newblock {The Monadic Second-Order Logic of Graphs. I. Recognizable Sets of
  Finite Graphs}.
\newblock {\em Inf. Comput.}, 85(1):12--75, 1990.

\bibitem{CyganNPPRW22}
Marek Cygan, Jesper Nederlof, Marcin Pilipczuk, Micha\l{} Pilipczuk, Johan
  M.~M. van Rooij, and Jakub~Onufry Wojtaszczyk.
\newblock {Solving Connectivity Problems Parameterized by Treewidth in Single
  Exponential Time}.
\newblock {\em {ACM} Trans. Algorithms}, 18(2):17:1--17:31, 2022.

\bibitem{BergBKMZ20}
Mark de~Berg, Hans~L. Bodlaender, S{\'{a}}ndor Kisfaludi{-}Bak, D{\'{a}}niel
  Marx, and Tom~C. van~der Zanden.
\newblock {A Framework for Exponential-Time-Hypothesis-Tight Algorithms and
  Lower Bounds in Geometric Intersection Graphs}.
\newblock {\em {SIAM} J. Comput.}, 49(6):1291--1331, 2020.

\bibitem{MontgolfierSV11}
Fabien de~Montgolfier, Mauricio Soto, and Laurent Viennot.
\newblock {Treewidth and Hyperbolicity of the Internet}.
\newblock In {\em Proceedings of The Tenth {IEEE} International Symposium on
  Networking Computing and Applications, {NCA} 2011, August 25-27, 2011,
  Cambridge, Massachusetts, {USA}}, pages 25--32. {IEEE} Computer Society,
  2011.

\bibitem{DemaineFHT05}
Erik~D. Demaine, Fedor~V. Fomin, Mohammad~Taghi Hajiaghayi, and Dimitrios~M.
  Thilikos.
\newblock {Subexponential parameterized algorithms on bounded-genus graphs and
  \emph{H}-minor-free graphs}.
\newblock {\em J. {ACM}}, 52(6):866--893, 2005.

\bibitem{DemaineH05}
Erik~D. Demaine and Mohammad~Taghi Hajiaghayi.
\newblock {Bidimensionality: new connections between {FPT} algorithms and
  PTASs}.
\newblock In {\em Proceedings of the Sixteenth Annual {ACM-SIAM} Symposium on
  Discrete Algorithms, {SODA} 2005, Vancouver, British Columbia, Canada,
  January 23-25, 2005}, pages 590--601. {SIAM}, 2005.

\bibitem{Dieng09}
Youssou Dieng.
\newblock {\em {D{\'e}composition arborescente des graphes planaires et routage
  compact}}.
\newblock PhD thesis, L'Universit{\'e} Bordeaux I, 2009.

\bibitem{DiengG09}
Youssou Dieng and Cyril Gavoille.
\newblock {On the Tree-Width of Planar Graphs}.
\newblock {\em Electron. Notes Discret. Math.}, 34:593--596, 2009.

\bibitem{graphtheoryDiestel}
Reinhard Diestel.
\newblock {\em Graph Theory, 4th Edition}, volume 173 of {\em Graduate texts in
  mathematics}.
\newblock Springer, 2012.

\bibitem{DornPBF10}
Frederic Dorn, Eelko Penninkx, Hans~L. Bodlaender, and Fedor~V. Fomin.
\newblock {Efficient Exact Algorithms on Planar Graphs: Exploiting Sphere Cut
  Decompositions}.
\newblock {\em Algorithmica}, 58(3):790--810, 2010.

\bibitem{EdmondsJ73}
Jack Edmonds and Ellis~L. Johnson.
\newblock {Matching, Euler tours and the Chinese postman}.
\newblock {\em Math. Program.}, 5(1):88--124, 1973.

\bibitem{FeldmannSLM20}
Andreas~Emil Feldmann, {Karthik {C. S.}}, Euiwoong Lee, and Pasin Manurangsi.
\newblock {A Survey on Approximation in Parameterized Complexity: Hardness and
  Algorithms}.
\newblock {\em Algorithms}, 13(6):146, 2020.

\bibitem{FluschnikKMNNT19}
Till Fluschnik, Christian Komusiewicz, George~B. Mertzios, Andr{\'{e}}
  Nichterlein, Rolf Niedermeier, and Nimrod Talmon.
\newblock {When Can Graph Hyperbolicity be Computed in Linear Time?}
\newblock {\em Algorithmica}, 81(5):2016--2045, 2019.

\bibitem{2020bodlaender}
Fedor~V. Fomin, Stefan Kratsch, and Erik~Jan van Leeuwen, editors.
\newblock {\em {Treewidth, Kernels, and Algorithms - Essays Dedicated to Hans
  L. Bodlaender on the Occasion of His 60th Birthday}}, volume 12160 of {\em
  Lecture Notes in Computer Science}. Springer, 2020.

\bibitem{FominLS18}
Fedor~V. Fomin, Daniel Lokshtanov, and Saket Saurabh.
\newblock {Excluded Grid Minors and Efficient Polynomial-Time Approximation
  Schemes}.
\newblock {\em J. {ACM}}, 65(2):10:1--10:44, 2018.

\bibitem{FournierIV15}
Herv{\'{e}} Fournier, Anas Ismail, and Antoine Vigneron.
\newblock {Computing the Gromov hyperbolicity of a discrete metric space}.
\newblock {\em Inf. Process. Lett.}, 115(6-8):576--579, 2015.

\bibitem{Frederickson87}
Greg~N. Frederickson.
\newblock {Fast Algorithms for Shortest Paths in Planar Graphs, with
  Applications}.
\newblock {\em {SIAM} J. Comput.}, 16(6):1004--1022, 1987.

\bibitem{Friedrich19}
Tobias Friedrich.
\newblock {From Graph Theory to Network Science: The Natural Emergence of
  Hyperbolicity (Tutorial)}.
\newblock In Rolf Niedermeier and Christophe Paul, editors, {\em 36th
  International Symposium on Theoretical Aspects of Computer Science, {STACS}
  2019, March 13-16, 2019, Berlin, Germany}, volume 126 of {\em LIPIcs}, pages
  5:1--5:9. Schloss Dagstuhl - Leibniz-Zentrum f{\"{u}}r Informatik, 2019.

\bibitem{Gao12}
Yong Gao.
\newblock Treewidth of erd{\H{o}}s-r{\'{e}}nyi random graphs, random
  intersection graphs, and scale-free random graphs.
\newblock {\em Discret. Appl. Math.}, 160(4-5):566--578, 2012.

\bibitem{GrigniKP95}
Michelangelo Grigni, Elias Koutsoupias, and Christos~H. Papadimitriou.
\newblock {An Approximation Scheme for Planar Graph {TSP}}.
\newblock In {\em 36th Annual Symposium on Foundations of Computer Science,
  Milwaukee, Wisconsin, USA, 23-25 October 1995}, pages 640--645. {IEEE}
  Computer Society, 1995.

\bibitem{gromov87}
Mikhael Gromov.
\newblock {Hyperbolic groups}.
\newblock In {\em Essays in group theory}, pages 75--263. Springer, 1987.

\bibitem{GuX19}
Qian{-}Ping Gu and Gengchun Xu.
\newblock {Near-linear time constant-factor approximation algorithm for
  branch-decomposition of planar graphs}.
\newblock {\em Discret. Appl. Math.}, 257:186--205, 2019.

\bibitem{Har-Peled20}
Sariel Har-Peled.
\newblock {Approximately: Independence Implies Vertex Cover}.
\newblock Note, Retrieved July 12, 2023, 2020.

\bibitem{HlinenyOSG08}
Petr Hlinen{\'{y}}, Sang{-}il Oum, Detlef Seese, and Georg Gottlob.
\newblock {Width Parameters Beyond Tree-width and their Applications}.
\newblock {\em Comput. J.}, 51(3):326--362, 2008.

\bibitem{ImpagliazzoP01}
Russell Impagliazzo and Ramamohan Paturi.
\newblock {On the Complexity of k-SAT}.
\newblock {\em J. Comput. Syst. Sci.}, 62(2):367--375, 2001.

\bibitem{KammerT16}
Frank Kammer and Torsten Tholey.
\newblock {Approximate tree decompositions of planar graphs in linear time}.
\newblock {\em Theor. Comput. Sci.}, 645:60--90, 2016.

\bibitem{Kisfaludi-Bak20-arxiv}
S{\'{a}}ndor Kisfaludi{-}Bak.
\newblock {Hyperbolic intersection graphs and (quasi)-polynomial time}.
\newblock {\em CoRR}, abs/1812.03960, 2018.

\bibitem{Kisfaludi-Bak20}
S{\'{a}}ndor Kisfaludi{-}Bak.
\newblock {Hyperbolic intersection graphs and (quasi)-polynomial time}.
\newblock In Shuchi Chawla, editor, {\em Proceedings of the 2020 {ACM-SIAM}
  Symposium on Discrete Algorithms, {SODA} 2020}, pages 1621--1638. {SIAM},
  2020.

\bibitem{Kisfaludi-Bak21}
S{\'{a}}ndor Kisfaludi{-}Bak.
\newblock A quasi-polynomial algorithm for well-spaced hyperbolic {TSP}.
\newblock {\em J. Comput. Geom.}, 12(2):25--54, 2021.

\bibitem{Klein08}
Philip~N. Klein.
\newblock {A Linear-Time Approximation Scheme for {TSP} in Undirected Planar
  Graphs with Edge-Weights}.
\newblock {\em {SIAM} J. Comput.}, 37(6):1926--1952, 2008.

\bibitem{klein2014optimization}
Philip~N Klein and Shay Mozes.
\newblock {Optimization algorithms for planar graphs}.
\newblock {\em preparation, manuscript at http://planarity. org}, 2014.

\bibitem{KleinMS2012}
Philip~N. Klein, Shay Mozes, and Christian Sommer.
\newblock {Structured Recursive Separator Decompositions for Planar Graphs in
  Linear Time}.
\newblock {\em CoRR}, abs/1208.2223, 2012.

\bibitem{KleinbergK01}
Jon~M. Kleinberg and Amit Kumar.
\newblock {Wavelength Conversion in Optical Networks}.
\newblock {\em J. Algorithms}, 38(1):25--50, 2001.

\bibitem{kopczynski2020hyperbolic}
Eryk Kopczy{\'n}ski.
\newblock Hyperbolic minesweeper is in p.
\newblock In {\em 10th International Conference on Fun with Algorithms (FUN
  2021)}. Schloss Dagstuhl-Leibniz-Zentrum f{\"u}r Informatik, 2020.

\bibitem{Korhonen21}
Tuukka Korhonen.
\newblock {A Single-Exponential Time 2-Approximation Algorithm for Treewidth}.
\newblock In {\em 62nd {IEEE} Annual Symposium on Foundations of Computer
  Science, {FOCS} 2021, Denver, CO, USA, February 7-10, 2022}, pages 184--192.
  {IEEE}, 2021.

\bibitem{DBLP:journals/talg/KowalikK06}
\L{}ukasz Kowalik and Maciej Kurowski.
\newblock {Oracles for bounded-length shortest paths in planar graphs}.
\newblock {\em {ACM} Trans. Algorithms}, 2(3):335--363, 2006.

\bibitem{KrauthgamerL06}
Robert Krauthgamer and James~R. Lee.
\newblock {Algorithms on negatively curved spaces}.
\newblock In {\em 47th Annual {IEEE} Symposium on Foundations of Computer
  Science {(FOCS} 2006), 21-24 October 2006, Berkeley, California, USA,
  Proceedings}, pages 119--132. {IEEE} Computer Society, 2006.

\bibitem{LauritzenS88}
Steffen Lauritzen and David~J. Spiegelhalter.
\newblock Local computations with probabilities on graphical structures and
  their applications to expert systems.
\newblock {\em Journal of the Royal Statistical Society, Series B},
  50(2):157--224, 1988.

\bibitem{Le18}
Hung Le.
\newblock {A {PTAS} for subset {TSP} in minor-free graphs}.
\newblock {\em CoRR}, abs/1804.01588, 2018.

\bibitem{LeskovecLDM09}
Jure Leskovec, Kevin~J. Lang, Anirban Dasgupta, and Michael~W. Mahoney.
\newblock {Community Structure in Large Networks: Natural Cluster Sizes and the
  Absence of Large Well-Defined Clusters}.
\newblock {\em Internet Math.}, 6(1):29--123, 2009.

\bibitem{LiptonT79}
Richard~J. Lipton and Robert~E. Tarjan.
\newblock {A Separator Theorem for Planar Graphs}.
\newblock {\em SIAM Journal of Applied Mathematics}, 36:177--189, 1979.

\bibitem{LiptonT80}
Richard~J. Lipton and Robert~Endre Tarjan.
\newblock {Applications of a Planar Separator Theorem}.
\newblock {\em {SIAM} J. Comput.}, 9(3):615--627, 1980.

\bibitem{ManiuSJ19}
Silviu Maniu, Pierre Senellart, and Suraj Jog.
\newblock An experimental study of the treewidth of real-world graph data.
\newblock In Pablo Barcel{\'{o}} and Marco Calautti, editors, {\em 22nd
  International Conference on Database Theory, {ICDT} 2019, March 26-28, 2019,
  Lisbon, Portugal}, volume 127 of {\em LIPIcs}, pages 12:1--12:18. Schloss
  Dagstuhl - Leibniz-Zentrum f{\"{u}}r Informatik, 2019.

\bibitem{Marx07a}
D{\'{a}}niel Marx.
\newblock {On the Optimality of Planar and Geometric Approximation Schemes}.
\newblock In {\em 48th Annual {IEEE} Symposium on Foundations of Computer
  Science {(FOCS} 2007), October 20-23, 2007, Providence, RI, USA,
  Proceedings}, pages 338--348. {IEEE} Computer Society, 2007.

\bibitem{NarayanS11}
Onuttom Narayan and Iraj Saniee.
\newblock {Large-scale curvature of networks}.
\newblock {\em Phys. Rev. E}, 84:066108, Dec 2011.

\bibitem{Po74}
Svatopluk Poljak.
\newblock {A note on stable sets and colorings of graphs}.
\newblock {\em Commentationes Mathematicae Universitatis Carolinae},
  15:307--309, 1974.

\bibitem{RobertsonSST97}
Neil Robertson, Daniel~P. Sanders, Paul~D. Seymour, and Robin Thomas.
\newblock {The Four-Colour Theorem}.
\newblock {\em J. Comb. Theory, Ser. {B}}, 70(1):2--44, 1997.

\bibitem{RobertsonS86a}
Neil Robertson and Paul~D. Seymour.
\newblock {Graph Minors. {II.} Algorithmic Aspects of Tree-Width}.
\newblock {\em J. Algorithms}, 7(3):309--322, 1986.

\bibitem{RobertsonS86b}
Neil Robertson and Paul~D. Seymour.
\newblock {Graph minors. V. Excluding a planar graph}.
\newblock {\em J. Comb. Theory, Ser. {B}}, 41(1):92--114, 1986.

\bibitem{Serdyukov78}
Anatoliy~I. Serdyukov.
\newblock O nekotorykh ekstremal’nykh obkhodakh v grafakh.
\newblock {\em Upravlyaemye sistemy}, 17:76--79, 1978.

\bibitem{ShavittT08}
Yuval Shavitt and Tomer Tankel.
\newblock {Hyperbolic embedding of internet graph for distance estimation and
  overlay construction}.
\newblock {\em {IEEE/ACM} Trans. Netw.}, 16(1):25--36, 2008.

\bibitem{Thorup04}
Mikkel Thorup.
\newblock {Compact Oracles for Reachability and Approximate Distances in Planar
  Digraphs}.
\newblock {\em J. ACM}, 51(6):993--1024, nov 2004.

\bibitem{tutte2001graph}
William~Thomas Tutte.
\newblock {\em {Graph theory}}, volume~21.
\newblock Cambridge university press, 2001.

\bibitem{Vatshelle12}
Martin Vatshelle.
\newblock {\em {New Width Parameters of Graphs}}.
\newblock PhD thesis, University of Bergen, Norway, 2012.

\end{thebibliography}

\end{document}